\documentclass[12pt,a4paper]{scrartcl}

\usepackage[utf8]{inputenc}
\usepackage{lmodern}
\usepackage{geometry}
\usepackage{amssymb}

\usepackage{amsmath}
\usepackage{graphicx}
\usepackage{amsthm}
\usepackage{enumitem}
\usepackage{mathtools}
\usepackage{aliascnt}  
\usepackage{mathrsfs}
\usepackage{dsfont}
\usepackage{nicefrac}
\usepackage{comment}

\renewcommand{\:}{\colon}

\usepackage{catoptions}
\makeatletter

\def\Autoref#1{%
  \begingroup
  \edef\reserved@a{\cpttrimspaces{#1}}%
  \ifcsndefTF{r@#1}{%
    \xaftercsname{\expandafter\testreftype\@fourthoffive}
      {r@\reserved@a}.\\{#1}%
  }{%
    \ref{#1}%
  }%
  \endgroup
}
\def\testreftype#1.#2\\#3{%
  \ifcsndefTF{#1autorefname}{%
    \def\reserved@a##1##2\@nil{%
      \uppercase{\def\ref@name{##1}}%
      \csn@edef{#1autorefname}{\ref@name##2}%
      \autoref{#3}%
    }%
    \reserved@a#1\@nil
  }{%
    \autoref{#3}%
  }%
}

\usepackage{hyperref}

\newcommand{\IR}{{\mathbb{R}}}
\newcommand{\IN}{{\mathbb{N}}}

\newcommand{\IC}{{\mathbb{C}}}

\newcommand{\abs}[1]{\left| #1 \right|}
\newcommand{\mm}[1]{\begin{pmatrix} #1 \end{pmatrix}}
\newcommand{\muu}[1]{\begin{align*} #1 \end{align*}}
\newcommand{\muun}[1]{\begin{align} #1 \end{align}}

\renewcommand{\d}{\,\mathrm{d}}

\newcommand{\abc}{\begin{enumerate}[label=(\alph*)]}
\renewcommand{\phi}{\varphi}
\renewcommand{\sc}[1]{\left< #1 \right>}
\newcommand{\nn}[1]{\left\Vert #1 \right\Vert}

\newcommand{\ind}{\mathds{1}}

\newcommand{\ran}{\operatorname{ran}}
\newcommand{\supp}{\operatorname{supp}}

\newcommand{\Id}{\operatorname{Id}}

\newcommand{\thref}[1]{\autoref{#1}}

\begin{document}
\newtheoremstyle{fgrs} 
                        {0.5em}    
                        {0.5em}    
                        {}         
                        {}         
                        {\bfseries}
                        {}        
                        {\newline} 
                        {}
\newtheoremstyle{mydef} 
                        {0.5em}    
                        {0.5em}    
                        {}         
                        {}         
                        {\bfseries}
                        {}        
                        {\newline} 
                        {}
\theoremstyle{mydef}

\newtheorem{deffn}{Definition}[section]
\newcommand{\furthertheorem}[3]{
\newaliascnt{#1}{deffn}  
\newtheorem{#1}[#1]{#2}  
\aliascntresetthe{#1}  
\providecommand*{#3}{#2}  
}

\furthertheorem{defn}{Definition}{\defnautorefname}
\furthertheorem{example}{Example}{\exampleautorefname}
\furthertheorem{rem}{Remark}{\remautorefname}

\newtheoremstyle{dotless} 
                        {0.5em}    
                        {0.5em}    
                        {\itshape}         
                        {}         
                        {\bfseries}
                        {}        
                        {\newline} 
                        {}
\theoremstyle{dotless}

\furthertheorem{prop}{Proposition}{\propautorefname}
\furthertheorem{thm}{Theorem}{\thmautorefname}
\furthertheorem{lemma}{Lemma}{\lemmaautorefname}
\furthertheorem{coro}{Corollary}{\coroautorefname}
\furthertheorem{hyp}{Hypothesis}{\hypautorefname}

\newcommand{\thmenum}{\leavevmode \vspace{-\baselineskip} \vspace*{-\medskipamount}}


\bibliographystyle{plain}

\newcommand{\el}{{\operatorname{el}}}
\renewcommand{\H}{\mathcal{H}}
\newcommand{\HH}{\mathcal{H}}
\newcommand{\h}{\mathfrak{h}}
\newcommand{\Fc}{\mathcal{F}}
\newcommand{\FF}{\mathfrak{F}}
\newcommand{\hs}{\mathfrak{h}}
\newcommand{\f}{{\operatorname{f}}}
\renewcommand{\c}{{\mathrm c}}
\newcommand{\Cci}{C_\c^\infty}
\newcommand{\fin}{{\operatorname{fin}}}
\newcommand{\D}{\mathcal{D}}
\newcommand{\G}{\overline G}
\newcommand{\Laux}{L_{\operatorname{aux}}}
\newcommand{\ess}{{\operatorname{ess}}}
\newcommand{\disc}{{\operatorname{disc}}}
\newcommand{\Uh}{\widehat U}
\renewcommand{\l}{\mathrm l}
\renewcommand{\r}{\mathrm r}

\renewcommand{\d}{{\mathrm d}}
\newcommand{\Ctot}{C_{\operatorname{tot}}}
\newcommand{\qtot}{q_{\operatorname{tot}}}

\newcommand{\aux}{{\operatorname{aux}}}
\newcommand{\sgn}{{\operatorname{sgn}}}
\newcommand{\Hf}{H_\f}
\newcommand{\Hfc}{\mathcal{H}_\f}
\newcommand{\Lf}{L_\f}
\newcommand{\oPi}{\PP{\Pi}}
\newcommand{\dist}{\operatorname{dist}}
\renewcommand{\Re}{\operatorname{Re}}
\renewcommand{\Im}{\operatorname{Im}}
\renewcommand{\P}{\overline P}
\renewcommand{\i}{\mathrm{i}}
\newcommand{\at}{{\operatorname{at}}}
\newcommand{\Af}{\mathfrak{A}}
\newcommand{\Wf}{\mathfrak{W}}
\newcommand{\C}{\mathcal{C}}
\newcommand{\slim}{\operatorname{s-lim}}
\newcommand{\dGamma}{\mathsf{d}\Gamma}
\renewcommand{\l}{\mathsf{l}}
\renewcommand{\r}{\mathsf{r}}
\newcommand{\Hfl}{H_{\f,\l}}
\newcommand{\Hfr}{H_{\f,\r}}
\newcommand{\IS}{\mathbb{S}}
\newcommand{\Ug}{U_g}
\newcommand{\Ugl}{\tilde U}
\newcommand{\Rc}{\mathcal{R}}


\newcommand{\R}{\mathbb{R}}
\newcommand{\p}{{\mathrm p}}
\newcommand{\s}{{\mathrm s}}
\renewcommand{\L}{\mathcal{L}}
\newcommand{\W}{\mathcal{W}}
\newcommand{\Nf}{N_\f}
\newcommand{\Pd}{P_{\operatorname{disc}}}
\newcommand{\Pdc}{\Pc}
\newcommand{\dG}{\mathsf{d}\Gamma}
\newcommand{\Dq}{\D_q^\infty }
\newcommand{\Dp}{\Cci(\IR^3)}

\newcommand{\aotimes}{\widehat\otimes}
\newcommand{\mult}[1]{{#1}}
\newcommand{\mults}[2]{{ #2 \cdot}_{#1}}
\newcommand{\Lb}{\mathcal{L}}

\newcommand{\Pw}{}
\newcommand{\ad}{\operatorname{ad}}
\newcommand{\Czi}{C^\infty_0}

\newcommand{\liim}{\mathsf{l.i.m.}}

\newcommand{\Vc}{V_{\mathsf c}}

\newcommand{\Aps}{{A_\p}}
\newcommand{\Apse}{{A_\p^{(\epsilon)}}}
\newcommand{\ee}{\eta_\epsilon}
\newcommand{\Dps}{\D_\p}
\newcommand{\Lps}{\Lambda_\p}
\newcommand{\Ls}{\Lambda}
\newcommand{\As}{{A}}
\newcommand{\Ase}{{A^{(\epsilon)}}}
\newcommand{\Ds}{{\D}}

\newcommand{\Vcd}{\mathcal{V}}
\newcommand{\Il}{I^{(\l)}}
\newcommand{\Ir}{I^{(\r)}}
\newcommand{\Vd}{V_{\mathsf{d}}}
\newcommand{\phie}[1]{ e^{\i #1 \mult{x}} }
\newcommand{\Pc}{P_{\mathsf{ess}}}
\newcommand{\phit}{\widetilde{\phi}}
\renewcommand{\S}{\mathcal{S}}
\newcommand{\Gos}{G(\omega,\Sigma)}
\newcommand{\Hos}{H(\omega,\Sigma)}
\newcommand{\const}{\mathsf{const.}}
\newcommand{\Adil}{A_{\mathsf{D}}}
\newcommand{\Df}{\FF_\fin(\Cci(\IR^3))}
\newcommand{\qce}{q_{C_1}}
\newcommand{\qcee}{q_{C^{(\epsilon)}_1}}

\newcommand{\Def}{\mathcal{D}}
\newcommand{\CQ}[1]{C_{ #1 }}
\newcommand{\PP}[1]{{ #1 }^\perp }
\newcommand{\Cb}{C_{\mathsf{b}}}
\newcommand{\Pb}{ C_{\mathsf{poly}}}
\newcommand{\q}{\mathsf q}
\newcommand{\Gt}{\tilde G}
\newcommand{\gb}{\gamma_\beta}
\newcommand{\Nfh}{\widehat{\Nf}}
\newcommand{\Poh}{\widehat{P_\Omega}}
\newcommand{\Dw}{\widetilde \D}
\newcommand{\hDk}{\hat{D}_k}
\newcommand{\Cf}[1]{C_{\f,#1}}

\numberwithin{equation}{section}

\newcommand{\phiz}{\phi_0}
\newcommand{\phir}{\phi_{\mathsf{R}}}
\newcommand{\refJ}[1]{\hyperref[Js]{($\text{J}_{#1}$)}}


\theoremstyle{fgrs}
\newtheorem*{fgr}{Fermi Golden Rule Condition}

\title{Thermal Ionization for Short-Range Potentials}
\author{David Hasler \and Oliver Siebert}
\date{\small Department of Mathematics, Friedrich Schiller University Jena \\
Jena, Germany}
\maketitle

\begin{abstract}
We study a concrete model of a confined particle  in form of a Schrö-dinger operator with a compactly supported smooth potential coupled to a bosonic field at positive temperature. We show,  that  the model  exhibits  thermal ionization for any positive temperature, provided the coupling is sufficiently small.  Mathematically, one has to rule out that zero is an eigenvalue of the self-adjoint generator of time evolution -- the Liouvillian. This will be done by using positive commutator methods with dilations in the space of scattering functions. Our proof relies on a spatial cutoff in the coupling but does otherwise not require any unnatural restrictions.
\end{abstract} 



\section{Introduction}
The phenomenon of  \textit{thermal ionization} can be viewed as a positive temperature generalization of the photoelectric effect: an atom is exposed to thermal radiation emitted by a black body of temperature $T  > 0$. Then photons of momentum $\omega$ according to \textsc{Planck}'s probability distribution of black-body radiation,
\[
\frac{1}{e^{\beta \omega}-1}, \qquad \beta \sim \frac{1}{T},
\]
 will interact with the atoms' electrons. Since there is a positive probability for arbitrary high energy photons, eventually one  with sufficiently high energy will show up exceeding   the ionization threshold of the atom. 
 
For the zero temperature situation (photoelectric effect) there can be found qualitative and quantitative statements for different simplified models of atoms with quantized fields in \cite{photoelectric1,photoelectric2,photoelectric3}. If one replaces the atom by a finite-dimensional small subsystem, the model usually exhibits the behavior of \textit{return to equilibrium}, see for example \cite{rte1}. Here the existence of a Gibbs state of the atom leads to the existence of an equilibrium state of the whole system. One is confronted with a similar mathematical problem -- disproving the degeneracy of the eigenvalue zero. The most common technique for handling this are complex dilations, or its infinitesimal analogue: positive commutators, which goes back to work of \textsc{Mourre} (cf. \cite{mourre}). There are a number of papers which also use positive commutators  in the context of return to equilibrium, see for example  \cite{rte2,jakpillet2,jakpillet3,merkli_positive}.

A first rigorous treatment of thermal ionization was given by \textsc{Fröhlich} and \textsc{Merkli} in \cite{merkli1}  and in a subsequent paper \cite{merkli2} by the same authors together with \textsc{Sigal}. The ionization appears mathematically as the absence of a time-invariant state in a suitable von Neumann algebra describing the complete system. The time-invariant states can be shown to be in one-to-one correspondence with the elements of the kernel of the Liouvillian. For the proof they used a global positive commutator method first established for Liouvillians in \cite{merkli_positive} and \cite{dj} with a conjugate operator on the field space as in \cite{jakpillet3}. Furthermore, they developed a new virial theorem. 

A similar situation in a mathematical sense also occurs if one considers a small finite-dimensional or confined system coupled to multiple reservoirs at different temperatures. Here one can prove as well the absence of time-invariant (equilibrium) states, which translates into the absence of zero as an eigenvalue of the Liouvillian (cf. \cite{jaksic_pillet_multiple_fermionic,derezinski_rte,instability_eqstates,ness_resonances}).

In \cite{merkli1} an abstract representation of a Hamiltonian diagonalized with respect to its energy with a single negative eigenvalue was considered. For the proof certain regularity assumptions of the interaction with respect to the energy were imposed. However, it is not so clear how these assumptions translate to a more concrete setting. The reference \cite{merkli2} on the other hand covers the case of a Schrödinger operator with a long-range potential but with only finitely many modes coupled via the interaction. Moreover, only a compact interval away from zero of the continuous spectrum was coupled. 

The purpose of this paper is to transfer the results of \cite{merkli1,merkli2} to a more specific model of a  Schrödinger operator with a compactly supported potential with finitely many eigenvalues.  We consider a typical coupling term and we have to impose a spatial cutoff. However, we do not need any restrictions with respect to the coupling to the continuous subspace of the atomic operator as in \cite{merkli2}. Moreover, in contrast to \cite{merkli1,merkli2} our result holds  uniformly for  bounded positive  temperatures.
This is achieved   by considering a finite  approximation of the so-called level shift operator. 
For the proof we use the same commutator on the space of the bosonic field as in \cite{jakpillet3,merkli1,merkli2}, and we also reuse the original virial theorem of \cite{merkli1,merkli2}. On the other hand we work with a different commutator on the atomic space, namely the generator of dilations in the space of scattering functions.  

The organization of the paper is as follows. In \Autoref{sec:model} we introduce the model and define the Liouvillian. In addition, we state the precise form of our main result in \thref{th:main} and all the necessary assumptions. We also give a more detailed outline of the proof in \Autoref{sec:overview proof}. \Autoref{sec:abstract virial} recalls the abstract virial theorem of \cite{merkli1,merkli2} and some related technical methods. Then we verify the requirements of the abstract virial theorem in our setting in \Autoref{section:verification}. We repeat the definition of scattering states (\Autoref{sec:scattering}) and use those for the concrete choice of the commutators.  The major difficulty here is to check that the commutators with the interaction terms are bounded. This requires  bounds which involve  the scattering functions and is elaborated in \Autoref{sec:estimates scattering}. The application of the virial theorem then yields a concrete version -- \thref{th:result virial}. This is the first key element for the proof of the main theorem. The second one is the actual proof that the commutator together with some auxiliary term is positive. This and the concluding proof of \thref{th:main} at the end can be found in \Autoref{sec:positivity}.

\section{Model and the Main Result}
\label{sec:model}

A model of a small subsystem interacting with a bosonic field at positive temperature is usually represented as a suitable $C^*$- or $W^*$-dynamical system on a tensor product algebra consisting of the field and the atom, respectively.

The field is defined  by a Weyl algebra with infinitely many degrees of freedom. To implement black-body radiation at a specific temperature $T > 0$ the GNS representation with respect to a KMS (equilibrium) state depending on $T$ is considered. For the atom the whole algebra $\L(H_\p)$ for the  atomic Hilbert space 
$\H_\p = L^2(\IR^3)$ is used and the GNS representation with respect to an arbitrary mixed reference state is performed. The combined representation of the whole system generates a $W^*$-algebra where the interacting dynamics can be defined by means of a Dyson series in a canonical way. Its self-adjoint generator -- the Liouvillian -- is of great interest when studying such systems.
The details of this construction can be found in \cite{merkli1,merkli2,mueck_phd}. Furthermore, it is shown in \cite{merkli1,merkli2} that the absence of zero as an eigenvalue of the Liouvillian implies the absence of time-invariant states of the $W^*$-dynamical system.

In this paper we start directly with the definition of the Liouvillian without repeating its derivation and the algebraic construction.
The only difference in our setting to \cite{merkli1,merkli2} is the coupling term which can be realized as an approximation of step function couplings as considered in their work. 

The purpose of this section is the definition of the concrete Liouvillian, the precise statement of the result --  absence of zero as an eigenvalue -- and the required conditions. At the end in \Autoref{sec:overview proof}, we explain the basic structure of the proof. 

We start with the three ingredients of the model: the atom, the field and the interaction, which we first discuss separately and state the required assumptions. 

\subsection{The Atom}

For the atom we consider a Schrödinger operator on the Hilbert space $\H_\p := L^2(\IR^3)$, 
\[
H_\p = -\Delta + V,
\]
where we assume that 

\begin{enumerate}[label=(H\arabic*)]
\item \label{ass:h1} $V \in \Cci(\IR^3)$,
\item \label{ass:h2} if $\psi \in L^2(\IR^3)$ satisfies
\begin{align} \label{eq:propofpot} 
\psi(x)  = - \frac{1}{4\pi} \int \frac{\abs{V(x)}^{\frac{1}{2}} V(y)^{\frac{1}{2}}  }{\abs{x-y}} \psi(y)  \d y 
, \quad \text{ for a.e.    }  x \in \IR^3  , \end{align} 
where $V^{1/2} := |V|^{1/2} \text{sgn} V$, then $\psi =  0$.
\end{enumerate}

We note that the above assumptions have the following immediate consequences. 

\begin{prop} If  \ref{ass:h1} holds, then  
\label{th:H implications}
\begin{itemize}
\item[(a)] $H_\p$ is essentially self-adjoint on $\Cci(\IR^3)$ with domain 
$\Def(\Delta)$,
\item[(b)] $H_\p$   has essential spectrum $[0,\infty)$, 
\item[(c)] the discrete spectrum of $H_\p$, denoted by $\sigma_\d(H_\p)$, is finite,
\item[(d)] $H_\p$ has no positive eigenvalues,
\item[(e)] $H_\p$ has no singular spectrum. 
\end{itemize} 
 If  \ref{ass:h1} and \ref{ass:h2} hold, then 
\begin{itemize}
\item[(f)] zero is not an eigenvalue of $H_\p$. 
\end{itemize}
\end{prop} 
\begin{proof} (a) follows from Kato-Rellich since $V$ is infinitesimally bounded with respect to $-\Delta$. (b) is 
shown for example in Example 6  of  Section XIII.4  in \cite{rs4}. (c) follows from \cite[Theorem XIII.6]{rs4}.
(d) follows from the Kato-Agmon-Simon
theorem (\cite[Theorem XIII.58]{rs4}).
(e)  follows from  \cite[Theorem XIII.21]{rs4}. (f) Suppose $\varphi$ were an eigenvector with eigenvalue zero. Then using the integral
representation  of the resolvent of the Laplacian, see for example \cite[IX.7]{rs2},  we find $\varphi(x) = - (4\pi)^{-1}  \int_{\R^3}  |x-y|^{-1} V(y)  \varphi(y) dy$.   
From this  it is straightforward to verify that $\psi = |V|^{1/2} \varphi$  would be   nonzero solution of  \eqref{eq:propofpot}. 
\end{proof} 
We shall denote by  $\Pc$ be the spectral projection to the essential spectrum of $H_\p$. As an immediate consequence of the 
above proposition we find that \ref{ass:h1} and \ref{ass:h2} imply  $\Pc = P_\text{ac}$,
where   $P_\text{ac} $ denotes the projection onto the absolutely continuous spectral subspace  of $H_\p$.  
As we are in a statistical physics setting, we have to consider density matrices as states of our system. They form a subset of the Hilbert-Schmidt operators, so we will work in a space isomorphic to the latter,
\[
\H_\p \otimes \H_\p. 
\]

\begin{rem} We note   that   \ref{ass:h2} is  mostly satisfied for the  potentials which we consider. 
To this end,  assume  \ref{ass:h1} and let $V_\alpha = \alpha V$ for $\alpha \geq 0$.   Let  $K_\alpha$ be the operator with integral kernel 
$-|V_\alpha(x)|^{1/2} V_\alpha^{1/2}(y) /(4\pi|x-y|)$. Clearly $K_\alpha = \alpha K_1$,
and  it is straightforward 
to see, that  $K_1$ is Hilbert-Schmidt and hence a compact operator. We conclude that 
for any $k > 0$, there exist only finitely many $\alpha \in [0,k]$ such that   $K_\alpha$
has eigenvalue one or equivalently   that   \ref{ass:h2} is violated for $V_\alpha$. 
 Furthermore, we note  that  \ref{ass:h2} 
is   equivalent to the statement that the so-called Fredholm deterimant 
$\text{det}_2(1 -  K_0)$ is nonzero, see \cite{newton}.  
 For definition of $\text{det}_2$ we refer the reader to 
\cite{rs4}.
Now  Fredholm determinants 
are   finite analytic (cf. \cite{simon_trace})
 and $\text{det}_2(1) \neq 0$,  which could be used to  obtain 
a  second argument why   potentials which do 
not satisfy   \ref{ass:h2} are rather rare. 
Finally, we note  the property    \ref{ass:h2} is referred to  in  \cite{newton} as  energy  0 being    an exceptional point. 
\end{rem}

\begin{rem}
For the result and our proof one actually needs just finitely many derivatives of $V$. Therefore, one could weaken \ref{ass:h1}. However it seems a bit tedious in the proof to keep track to which order exactly derivatives are required. 
\end{rem}
\subsection{The Quantized Field}
The quantized field will be described by operators on Fock spaces. 
Let $\FF(\hs)$ denote the bosonic Fock space over a Hilbert space $\hs$, that is, 
\[
\FF(\hs) = \bigoplus_{n=0}^\infty \hs^{\otimes_\s n},
\]
where  $\otimes_\s n$ denotes  the $n$-fold symmetric tensor product of Hilbert spaces, and $\hs^{\otimes_\s 0} := \IC$. For $\psi \in \FF(\hs)$, we write $\psi_n$ for the $n$-th element in the direct sum and we use the notation $\psi = (\psi_0, \psi_1, \psi_2, \ldots)$. The vacuum vector is defined as $\Omega := (1,0,0,\ldots)$. For a dense subspace $\mathfrak{d} \subseteq \hs$ we define the dense space of finitely many particles in $\FF(\hs)$ by
\begin{align*}
\FF_\fin(\mathfrak{d}) := &\{  (\psi_0, \psi_1, \psi_2, \dots)  \in \FF(\hs) \: \psi_n \in \mathfrak{d}^{\otimes_\s n} \text{ for all }  n \in \IN_0 \\
&\text{ and there exists }  N \in \IN_0 \: \psi_n = 0 \text{ for } n \geq N \},
\end{align*}
where here   $\otimes_\s n$  represents the $n$-fold symmetric tensor product of vector spaces. 
Again we will work on the space of density matrices and thus use the space
\begin{align} \label{eq:defgluinghilb} 
\FF:= \FF(L^2(\IR \times \IS^2 )) \cong  \FF(L^2(\IR^3)) \otimes \FF(L^2(\IR^3))  ,
\end{align} 
where $L^2(\IR \times \IS^2 ) := L^2 (\IR \times \IS^2,\d u \times \d \Sigma)$ and $\d \Sigma$ denotes the standard spherical measure on $\IS^2$. We note that the canonical identification
\eqref{eq:defgluinghilb},  outlined  in \autoref{rem:origin}
below,  is referred to as  `gluing' and was first introduced  by \textsc{Jak{\v{s}}i{\'{c}}} and \textsc{Pillet} in \cite{jakpillet2}.
For $\psi \in \FF$ notice that $\psi_n$ can be understood as an $L^2$ function in $n$ symmetric variables $(u,\Sigma) \in \IR \times \IS^2$. 

Let $\H$ be a Hilbert space.  By $\mathcal{L}(\H)$ we shall denote the space of 
all bounded linear operators on $\H$. On $\H \otimes \FF$ we define a so-called generalized  annihilation operator for a function $F \in L^2(\IR \times \IS^2, \L( \H ))$ by
\begin{align*}
(a(F) \psi )_n&(u_1, \Sigma_1, \ldots , u_n, \Sigma_n )  \\ &=
\sqrt{n+1} \int  F(u,\Sigma)^*
\psi_{n+1}(u,\Sigma, u_1, \Sigma_1, \ldots , u_n, \Sigma_n  ) \d(u,\Sigma) .
\end{align*}
Note that $a(F) \Omega = 0$. The generalized creation operator $a^*(F)$ is defined as the adjoint of $a(F)$. By definition $F \mapsto a^*(F)$ is linear whereas $F \mapsto a(F)$ is anti-linear. In the scalar case where  $\H = \IC$,  we have $\H \otimes \FF = \FF$ and one obtains the usual   creation and annihilation operators satisfying  the canonical commutation relations, for $f,g \in L^2(\IR \times \IS^2)$,
\[
[a(f),a^*(g)] = \sc{f,g}, \qquad [a(f),a(g)] = 0, \qquad [a^*(f),a^*(g)] = 0.
\]
We also define the field operators as $\Phi(F) = a(F) + a^*(F)$. 

For a measurable function  $M \: \IR \times \IS^2 \to \IR$ we introduce the second quantization $\dG(M)$, which is a self-adjoint operator on $\FF$, given for $\psi \in \FF$ by
\[
(\dG(M) \psi)_n(u_1, \Sigma_1, \ldots, u_n,\Sigma_n) = \left(\sum_{i=1}^n M(u_i,\Sigma_i) \right) \psi_n(u_1, \Sigma_1, \ldots, u_n,\Sigma_n),  n \in \IN_0.
\]
In particular we will use the number operator,
\begin{align*}
\Nf := \dG(1).
\end{align*}

\subsection{Liouvillian with Interaction}
\label{subsec:liouvillian}
We assume to have an interaction term with a smooth spatial cutoff.  For $(\omega,\Sigma) \in \IR_+ \times \IS^2$, where $\IR_+ := (0,\infty)$, we define a bounded multiplication operator on $\H_\p$ by
\begin{align}
\label{eq:G definition}
 G(\omega,\Sigma)(x)  = \kappa(\omega) \chi(x) \Gt(\omega,\Sigma)(x), \qquad x \in \IR^3,
\end{align}
where  $\kappa$ is a   function on $\IR_+$  and $\chi \in \S(\IR^3)$ -- the space of Schwartz functions, and for each $(\omega,\Sigma)$, $\Gt(\omega,\Sigma)$ is a function on $\IR^3$, satisfying the following conditions.
\begin{enumerate}[label=(I\arabic*)]
\item \label{assumption:I1}
\textit{Spatial cutoff:}  
 For all  $n \in \{0,1,2,3\}$
and $\alpha \in \IN_0^3$   the partial derivatives  
 $ \partial_x^\alpha  \partial^n_\omega \Gt$ exist  and are continuous on $\IR_+ \times \IS^2  \times \IR^3$,
and 
 there exists a polynomial $P$ and an $M \in \IN_0$ such that  for all $(\omega,\Sigma,x) \in \IR_+ \times \IS^2 \times  \IR^3$ 
\begin{align} \label{I3polcond} 
\abs{  \partial_x^\alpha  \partial^n_\omega \Gt(\omega,\Sigma)(x) } \leq P(\omega) \langle x \rangle^M ,  
\end{align} 
where $\sc{x} := (1 + x^2)^{1/2}$. 
\item
\label{assumption:I2}
\textit{UV cutoff:} $\kappa$ decays faster than any polynomial, that is, for all $n \in \IN$,
\[
\sup_{\omega \geq 1} \omega^n \abs{ \kappa(\omega) } < \infty.
\]
\item
\label{assumption:I3}
\textit{Regularity and infrared behavior:} 
$\kappa \in C^3(\IR_+)$ and one of the following two properties holds:
\begin{enumerate}[label=(\roman*)]
\item \label{I3:first}
 there exist $k, C \in (0,\infty)$, $p > 2$ such that 
\[
\abs{ \partial^j_\omega \kappa(\omega)   }  \leq  C \omega^{p-j}, \quad  \omega \in (0, k) , 
\quad   j = 0,\ldots,3 , 
\]
\item  \label{I3:second}
there exist   $J \in \IN_0$  and  $\kappa_0 \in C^s([0,\infty))$, with   $s = \max\{0,3-J\}$, such that 
$$
\kappa(\omega)  = \omega^{-\frac{1}{2} + J } \kappa_0(\omega), \quad   \omega >0  ,
$$
and there  exists an extension  $\tilde{G}_0$  of  $\tilde{G}$ to  $[0,\infty) \times \IS^2 \times  \IR^3$
such that  for all $j \in \{0,\ldots,s\}$ and all $\alpha \in \IN_0^3$  the partial derivatives 
 $ \partial_x^\alpha \partial^j_\omega \Gt_0$ exist,  are continuous,     satisfy  \eqref{I3polcond} for all  $(\omega,\Sigma,x) \in [0,\infty) \times \IS^2 \times  \IR^3$, 
and 
\begin{align*} 
&  \partial_{\omega}^j (\kappa_0(\omega) \chi \tilde G_0(\omega,\Sigma))(x)  |_{\omega=0} =    (-1)^{j+J+1} \partial_{\omega}^j  \overline{  (\kappa_0(\omega) \chi \tilde G_0(\omega,\Sigma))(x)}|_{\omega=0} .
\end{align*} 
\end{enumerate}
\end{enumerate}
Let $\beta > 0$ be the inverse temperature and let
\[
\rho_\beta(\omega) := \frac{1}{e^{ \beta \omega} - 1 }
\]
be the probability distribution for black-body radiation. 
To describe the interaction in the positive temperature setting it is convenient to introduce a  map  $\tau_\beta$ as follows. 
For a function $F \: \IR_+ \times \IS^2 \to \mathcal{L}(H_\p)$   we define a function $\tau_\beta F \: ( \R \setminus \{ 0 \} ) \times \IS^2 \to \mathcal{L}(H_\p)$  by 
\begin{align}
\label{eq:taubeta defn}
(\tau_\beta F)(u, \Sigma) &:= \begin{cases} u \sqrt{1+\rho_\beta(u)} F(u, \Sigma), & u > 0, \\
									 u \sqrt{\rho_\beta(-u)}  F(-u ,\Sigma)^* , & u < 0 . \end{cases}
\end{align}
It is straightforward to verify that  $\tau_\beta$ maps the following spaces into each other 
\begin{align*} 
\tau_\beta \: L^2(\IR_+  &\times \IS^2, (\omega^2 + \omega ) \d \omega \d \Sigma ,  \L(\H_\p) ) \rightarrow   L^2(\IR \times \IS^2,   \d \omega \d \Sigma ,  \L(\H_\p)) , 
\end{align*} 
where we adapted the convention to use  again the symbol  for the  restriction.

Let $\C_\p$ be the complex conjugation on $\H_\p$ given by $\C_\p \psi(x) := \overline{\psi(x)}$, and for an operator $T \in \L(\H_\p)$ we define  $\overline{T} := \C_\p T \C_\p$.

Let 
\[
\Dw := \Cci(\IR^3) \otimes \Cci(\IR^3) \otimes \FF_\fin(\Cci(\IR \times \IS^2)),
\]
which is a dense subspace of the composite Hilbert space
\[
\H := \H_\p \otimes \H_\p \otimes \FF.
\]
On $\Dw$ we define for $\lambda \in \IR$ the Liouvillian by
\begin{align}
\label{eq:L definiton}
L_\lambda := (H_\p \otimes \Id_\p - \Id_\p \otimes H_\p) \otimes \Id_\f  + \Id_\p \otimes \Id_\p \otimes \dG(\mult{u}) + \lambda  W  ,
\end{align}
where $W := \Phi(I)$, and
\muu
{
I(u,\Sigma) &:= \Il(u,\Sigma) \otimes \Id_\p + \Id_\p \otimes \Ir(u,\Sigma),    \\
\Il(u, \Sigma) &:= \tau_\beta(G)(u,\Sigma), \qquad \Ir(u, \Sigma) :=  - e^{-\beta u  /2}   \tau_\beta(\G^*)(u,\Sigma).
}
\newcommand{\Hft}{{\tilde{\H}_\f}
}
Note that by \ref{assumption:I1}--\ref{assumption:I3}, $G \in L^2(\IR_+ \times \IS^2, (\omega^2 + \omega ) \d \omega \d \Sigma ,  \L(\H_\p) )$ and therefore, $I \in L^2(\IR \times \IS^2, \Lb(\H_\p \otimes \H_\p))$ and the expression $\Phi(I)$ is well-defined.
It will be shown below (\autoref{th:gjn check}) that $L_\lambda$ is indeed essentially self-adjoint.
\begin{rem} \label{rem:origin}  Let us  give a unitarily equivalent 
definition  of  the  Liouvillian \eqref{eq:L definiton}, which uses the usual  notations from physics
and which  does not involve the 
`gluing' construction.   Let   $c^*(k)$ and $c(k)$, $k \in \IR^3$, denote the usual
creation and annihilation operator-valued distributions  of   $\FF(L^2(\IR^3))$.    Let the free  field energy operator on $\FF(L^2(\IR^3))$
be given by 
$\Hf = \int |k| c^*(k) c(k) dk$.  For the coupling function introduced in \eqref{eq:G definition}
we define  $\widehat{G}(\omega \Sigma) = G(\omega,\Sigma)$ for all $(\omega,\Sigma) \in \R_+ \times \IS^2$. It follows 
from  \ref{assumption:I1}--\ref{assumption:I3} that $\widehat{G}  \in L^2(\IR^3, (1+ \abs k^{-1}) \d k , \L( \H_\p ))$.
We consider  the following operator in  $\H_\p \otimes \H_\p \otimes   \FF(L^2(\IR^3))  \otimes  \FF(L^2(\IR^3))    $,
\begin{align*}
\widehat{L}_\lambda = (H_\p \otimes &  \Id_\p - \Id_\p  \otimes H_\p) \otimes \Id_\f \otimes  \Id_\f   + \Id_\p \otimes \Id_\p \otimes (\Hf \otimes   \Id_\f - 
  \Id_\f  \otimes \Hf) 
\\  + 
\lambda \int \d k \bigg\{  & \left( \sqrt{ 1 + \rho_\beta(|k|)} \widehat{G}(k) \otimes \Id_\p - \sqrt{\rho_\beta(|k|)} \Id_\p \otimes \overline{\widehat{G}}^*(k) \right) \otimes c^*(k) \otimes \Id_\f \\
 +& \left( \sqrt{ 1 + \rho_\beta(|k|) } \widehat{G}^*(k) \otimes \Id_\p - \sqrt{\rho_\beta(|k|)} \Id_\p \otimes \overline{\widehat{G}}(k) \right) \otimes c(k) \otimes \Id_\f \\
 +& \left( \sqrt{  \rho_\beta(|k|)} \widehat{G}^*(k) \otimes \Id_\p - \sqrt{ 1+ \rho_\beta(|k|)} \Id_\p \otimes \overline{\widehat{G}}(k) \right)  \otimes \Id_\f \otimes c^*(k)  \\
 +& \left( \sqrt{  \rho_\beta(|k|)} \widehat{G}(k) \otimes \Id_\p - \sqrt{1 + \rho_\beta(|k|)} \Id_\p \otimes \overline{\widehat{G}}^*(k) \right)  \otimes \Id_\f \otimes c(k)  \bigg\}  , 
\end{align*}
which is defined on the dense subspace 
$\hat{ \mathcal{D} } := \Cci(\IR^3) \otimes \Cci(\IR^3) \otimes \FF_\fin( \Def(|\cdot|)) \otimes  \FF_\fin(\Def(|\cdot|))$. 
One can show that  $\widehat{L}_\lambda$ is unitarily equivalent to \eqref{eq:L definiton}. 
To this end,  we  define the unitary transformation 
\[
\psi  \: L^2(\R^3) \longrightarrow L^2(\R_+ \times \mathbb{S}^2 )\]
with $\psi(f)(\omega,\Sigma) :=  \omega   f(\omega \Sigma)$, inducing 
the unitary transformation of Fock spaces  $\Gamma(\psi)$, see for example 
\cite[Section X]{rs2} , 
\begin{align*} 
&  \Gamma(\psi)  \: \FF(L^2(\IR^3)) \longrightarrow  \FF(L^2(\R_+ \times \mathbb{S}^2 ) )  ,  
\end{align*} 
and the unitary transformation 
\begin{align*}
&\tau \: L^2(\R_+ \times \mathbb{S}^2 ) \oplus L^2(\R_+ \times \mathbb{S}^2) \longrightarrow L^2(\R \times \mathbb{S}^2), \\ &\tau(f_1,f_2)(u, \Sigma) := \begin{cases} f_1(u, \Sigma), & u > 0, \\ - f_2(-u,\Sigma), & u < 0,  \end{cases}
\end{align*}
inducing the unitary transformation of Fock spaces  $\Gamma(\tau)$.
Furthermore,  let $V_\mathfrak{h}$ denote the canonical unitary map 
$$
V_\mathfrak{h} \:  \FF(\mathfrak{h}) \otimes    \FF(\mathfrak{h})  \longrightarrow   \FF(\mathfrak{h}   \oplus \mathfrak{h} )   ,
$$
characterized by mapping the tensor product of the vacua to the vacuum and  satisfying  $V_\mathfrak{h} (b(f)\otimes \Id + \Id \otimes b(g)   ) V_\mathfrak{h}^* = b(f,g)$, here $b(\cdot)$ 
stands for the usual annihilation operator on the corresponding Fock spaces.  As a consequence of the   definition it follows that 
 $U := \Gamma(\tau)   \circ V_{  L^2(\R_+ \times \mathbb{S}^2   ) } \circ  (  \Gamma(\psi)  \otimes \Gamma(\psi) ) $ is an  unitary transformation 
\begin{equation}  \nonumber 
U \: \FF(L^2(\IR^3)) \otimes \FF(L^2(\IR^3))    \longrightarrow  \FF(L^2(\IR \times \IS^2 ))   , 
\end{equation} 
which 
satisfies  the   property 
\begin{equation} \label{eq:propgenequiv1} 
  a(u ,\Sigma)  =  U \left(  \ind_{u > 0} u  c (u \Sigma) \otimes  \Id  +   \Id \otimes \ind_{u < 0}  u c(- u \Sigma) \right) U^*  , \quad (u,\Sigma) \in \IR \times \IS^2 , 
\end{equation} 
 where   $a(u,\Sigma)$ denotes  the usual
 annihilation operator-valued distribution  of   $\FF(L^2( \allowbreak \R  \times \mathbb{S}^2 )  ) $.
Using    \eqref{eq:propgenequiv1}  it is    straightforward to verify that
 $ L_\lambda = U \widehat{L}_\lambda   U^*$ on $\Dw$. 

\end{rem}

\subsection{Main Result}
\label{subsection:main}

For the proof of our main result we need an  additional assumption. The  instability of the eigenvalues should be  visible in second order in perturbation theory with respect to the coupling constant. This term is also called level shift operator and the corresponding positivity assumption \textit{Fermi Golden Rule condition}. In \Autoref{subsec:applications}  an example is provided where this is satisfied. 
\begin{fgr}
For $E \in \sigma_\d(H_\p)$ let $p_E := \ind_{\{E\}}(H_\p)$  be the spectral projection corresponding to the eigenvalue $E$.
For  $\varepsilon > 0$ and  $E \in \sigma_\d(H_\p)$ let  $\gb(E,\varepsilon)$ be the largest 
number such that 
\begin{align*}
p_E \left( F^{(1)}_\beta(E,\varepsilon) +  F^{(2)}_\beta(E,\varepsilon) \right) p_E   &\geq \gb(E,\varepsilon) p_E,
\end{align*}
where
\begin{align*}
F^{(1)}_\beta(E,\varepsilon) &:= \int_{0}^\infty \int_{\IS^2}  \frac{\omega^2}{e^ {\beta \omega}  - 1}  G(\omega,\Sigma)  \frac{ \Pc}{(H_\p - E - \omega)^2 + \varepsilon^2} G(\omega,\Sigma)^*    \d \Sigma  \d \omega , \\
F^{(2)}_\beta(E,\varepsilon) &:= \int_0^{\infty} \int_{\IS^2}   \frac{\omega^2}{1 - e^ {-\beta \omega} }  G(\omega,\Sigma)^*  \frac{ \Pc}{(H_\p - E + \omega)^2 + \varepsilon^2} G(\omega,\Sigma)    \d \Sigma  \d \omega .
\end{align*}
Furthermore we set $\gb(\varepsilon) := \inf_{E \in \sigma_\d(H_\p)} \gb(E,\varepsilon)$. 
\begin{enumerate}[label=(F)]
\item \label{fgrc} By the Fermi Golden Rule Condition for $\beta$ we mean that there exists an  $\varepsilon > 0$ such that $\gamma_\beta(\varepsilon ) > 0$. 
\end{enumerate}

\end{fgr}
Notice that $\gb$ might depend on $\beta$. In particular, this is the case if, as in \cite{merkli1,merkli2}, \ref{fgrc} is verified using only the term $\varepsilon F^{(1)}_\beta(E,\varepsilon)$ (for $\varepsilon \to 0$), which decays exponentially to zero if $\beta \to \infty$. However, one can also obtain results uniformly in $\beta$ for $\beta \geq \beta_0 > 0$ (low temperature) by proving that $F^{(2)}_\beta(E,\varepsilon)$  is positive for a fixed $\varepsilon > 0$, which will be done in the next section,  \autoref{th:application2}. 

Let  us now state the main result of this paper.
%
\begin{thm} 
\label{th:main}
Assume that \ref{ass:h1},  \ref{ass:h2} and  \ref{assumption:I1}--\ref{assumption:I3} hold. Let $\beta_0 > 0$ and  $\varepsilon > 0$. 
Then there exists a constant $ C > 0$ such that the following holds.  If   $\beta \geq \beta_0$ and  $0 < \abs \lambda <  C \min\{ 1 , \gb^2(\varepsilon)\}$    the operator  $L_\lambda$ given in \eqref{eq:L definiton} does not have  zero as an eigenvalue. 
\end{thm}


\subsection{Application}
\label{subsec:applications}

\renewcommand{\q}{\hat {\mathsf x}}

In the following we present an example of a QED system with a linear coupling term (Nelson Model) where the conditions for the main theorem are satisfied. In particular, one can verify the Fermi Golden Rule condition \ref{fgrc} in this case. We summarize this in the following corollary to \autoref{th:main}.
\begin{coro}
\label{th:application2}  Assume that  \ref{ass:h1}, \ref{ass:h2} hold.
Let 
\begin{align}
\label{eq:eikx3}
G(\omega,\Sigma)(x) =  \kappa(\omega ) e^{ \i \omega \Sigma x } \chi( x), \qquad \text{for all } \quad  x \in \IR^3, \, (\omega, \Sigma  ) \in \R_+ \times \IS^2 , 
\end{align}
where  $\chi \in \S(\IR^3)$ is nonzero, and  $\kappa$ is a nonzero   function on $\IR_+$ satisfying \ref{assumption:I2} and one of the following two conditions: 
\begin{enumerate}[label=(\alph*)]
\item \label{I3 application first} part \ref{I3:first} of  \ref{assumption:I3}  holds,
\item \label{I3 application second} $\chi$ is real-valued and there exist positive constants $c$ and $C$ such that 
\begin{align}
\label{eq:kappa other case}
\kappa(\omega) = C \omega^{-\frac{1}{2}} e^{-c \omega^2}
\end{align}
for all $\omega \geq 0$. 
\end{enumerate}
Then for any $\beta_0 > 0$ there exists a  $\lambda_0 > 0 $, such that  whenever $0 < \abs \lambda < \lambda_0$ 
and  $\beta \geq \beta_0$ the operator  $L_\lambda$, or equivalently  $\widehat{L}_\lambda$ defined in \autoref{rem:origin}, 
does not have zero as an eigenvalue. 
\end{coro}
\begin{proof}[Proof of \autoref{th:application2}]
Derivatives with respect to $\omega$ and $x$ yield only polynomial growth in $x$ and $\omega$, respectively. Thus, \ref{assumption:I1} is satisfied. The conditions \ref{ass:h1}, \ref{ass:h2}, \ref{assumption:I2} and  \ref{assumption:I3} \ref{I3:first} (if \ref{I3 application first} holds), are satisfied by assumption. In the case of \ref{I3 application second} note that $G$ in \eqref{eq:eikx3} can be multiplied with any phase $e^{\i \varphi}$, $\varphi \in \IR$, which just yields unitary equivalent Liouvillians by means of the unitary transformation $\Id_\p \otimes \Id_\p \otimes \Gamma(e^{\i \varphi})$ (cf. \cite[Section X]{rs2} for the definition of $\Gamma$). Therefore, we can assume without loss of generality that instead of \eqref{eq:kappa other case}, we have
\[
\kappa(\omega) = \i C  \omega^{-\frac{1}{2}} e^{-c \omega^2}.
\]
The condition \ref{assumption:I3} \ref{I3:second} is actually satisfied in this case. One has to verify
\[
 \partial_{\omega}^j ( \i e^{-c \omega^2} \chi(x)  e^{ \i \omega \Sigma x } )  |_{\omega=0} =    (-1)^{j+1} \partial_{\omega}^j  \overline {\i e^{-c \omega^2} \chi(x)  e^{ \i \omega \Sigma x } }|_{\omega=0}
\]
 for all $\Sigma \in \IS^2$, $x \in \IR^3$ and  $j=0, \ldots ,3$, which is true.

It remains the verification of the Fermi Golden Rule condition. This will
follow from  \autoref{th:fgrc verification2}, below, since $\sigma_\d(H_\p)$ is finite. 
\end{proof}

\begin{prop}
\label{th:fgrc verification2} Suppose the assumptions of \autoref{th:application2} hold, and   let $E \in \sigma_\d(H_\text{p})$. Then 
for any $\varepsilon > 0$ there exists a  $\gamma > 0$ (independent of $\beta$)  such that 
\begin{align*}
p_E F^{(2)}_\beta(E,\varepsilon) p_E
 &\geq \gamma p_E. 
\end{align*}
\end{prop}
\begin{proof}
 Let $\varphi_E$ denote a normalized eigenvector  of $-\Delta +V$ with eigenvalue $E$.  
First observe  that 
\begin{align*}
&\sc{ \phi_E, F^{(2)}_\beta(E,\varepsilon) \phi_E} \\ &\qquad \geq \int_0^{\infty} \int_{\IS^2}   \omega^2 \sc{\phi_E, G(\omega,\Sigma)^*  \frac{ \Pc}{(H_\p - E + \omega)^2 + \varepsilon^2} G(\omega,\Sigma) \phi_E}   \d \Sigma  \d \omega.
\end{align*}
The integrand is continuous in $(\omega,\Sigma)$ and non-negative. 
Thus, by the  finite dimensionality of the eigenspace with eigenvalue $E$  and continuity  it suffices to show that for some $(\omega,\Sigma)$ we have 
\[
\sc{G(\omega,\Sigma) \phi_E,  \frac{ \Pc}{(H_\p - E + \omega)^2 + \varepsilon^2} G(\omega,\Sigma) \phi_E} \not= 0,
\]
which follows  if we can show
\begin{align}
\label{eq:phi_E not zero}
\Pc G(\omega,\Sigma) \phi_E \not= 0.
\end{align}
 We now  claim that \eqref{eq:phi_E not zero} holds for some   $(\omega,\Sigma)$.
Otherwise, by the finite dimensionality of the range of  $\Id-P_\text{ess}$,  the space  
\begin{equation} \label{eq:infspaceclaim} 
\{ G(\omega, \Sigma) \phi_E  \:  (\omega,\Sigma) \in (0,\infty) \times \IS^2 \} 
\end{equation} 
would be finite-dimensional as well. But this leads to contradiction. 
By unique  continuation of eigenfunctions (\cite[Theorem XIII.57]{rs4})  $\phi_E(x) \not= 0$ for a.e. $x \in \IR^3$, and thus,  $\chi \varphi_E$ is nonzero. 
By assumption  $\kappa$ does not vanish on some nonempty open interval $I \subset (0,\infty)$. 
It now follows  that 
\[
\{ \tilde{G}(\omega, \Sigma) \chi \phi_E  \:  (\omega,\Sigma ) \in  I \times \IS^2 \} 
\]
is a subspace of \eqref{eq:infspaceclaim}  and has 
infinite dimension, by well-known methods (e.g. by calculating  Wronskians and  Vandermonde determinants). 
\end{proof}

Instead of using $F^{(2)}_\beta(E,\varepsilon)$ one could  also verify \ref{fgrc} with the first term $F^{(1)}_\beta(E,\varepsilon)$ in the limit $\varepsilon \to 0$. This does not improve the qualitative statement of \autoref{th:application2} and has the drawback that $\gb(\varepsilon) \to 0$ as $\beta \to \infty$ as in \cite{merkli1,merkli2}. However, in certain situations 
 it  might give the  dominant contribution. In this context we would like to mention the ``zero temperature 
result'' about the  leading order contribution of  the ionization probability in
the photoelectric effect  \cite{photoelectric3}.

\subsection{Overview of the Proof}
\label{sec:overview proof}

\renewcommand{\q}{\hat{\mathsf k}}
\newcommand{\po}{{\hat{\mathsf q}}}

The first step is to find a suitable conjugate operator $A$ consisting of a part $\Aps$ on the particle space $\H_\p$ and a part $A_\f$ on the field space $\FF$. 

For the latter we make the same choice as established for the first time in \cite{jakpillet3} and later also used by \textsc{Merkli} and co-authors in \cite{merkli_positive,merkli1,merkli2}, namely the second quantization of the generator of translations, 
\[
A_\f = \dG(\i \partial_u).  
\]
Let $P_\Omega$ denote the orthogonal projection onto the one-dimensional subspace containing the vacuum $\Omega$ (``vacuum subspace''). Formally, we obtain on  $\FF$  that
\[ \i [\dG(\mult{u}), A_\f] = \dG( \Nf ) \geq \PP{P_\Omega},
\]
which yields a positive contribution on the space orthogonal to  the vacuum.
The Hilbert space $\mathcal{H}$ can be further decomposed by means of the projection 
\begin{equation}
\label{eq:Pi defn}
\Pi := \ind_{L_0 = 0} = \ind_{L_\p = 0}\otimes P_\Omega
\end{equation} 
as 
\begin{equation}
\label{eq:ranpi}
\H = \ran (\Id_{\H_\p \otimes \H_\p} \otimes P_\Omega^\perp ) \oplus \ran \Pi \oplus  \ran ( \ind_{L_\p \not= 0} \otimes P_\Omega).
\end{equation} 
To obtain a positive operator on $\ran \Pi$, we proceed again as in \cite{merkli1,merkli2} and consider a bounded operator $A_0$ on the whole space $\H$.  The Fermi Golden Rule Condition \ref{fgrc} then implies that 
\[
\i \Pi [L_\lambda, A_0] \Pi  > 0.
\]
 The details can be found in \Autoref{subsec:error}.

Let $\Pd$ denote the spectral projection to the discrete spectrum of $H_\p$. Note that  $ \PP{\Pd}  = \Pc $. The third space in \eqref{eq:ranpi} can be decomposed further by use of
\begin{align}
\nonumber 
\ind_{L_\p \not= 0} &= ( \Pc \otimes \Pc) \oplus  (\Pc \otimes \Pd) \\ &\qquad  \oplus (\Pd \otimes \Pc)  \oplus  \ind_{L_\p \not= 0}(\Pd \otimes \Pd ) .\label{eq:Lp not zero} 
\end{align}
We start with the space generated by the first projection $\Pc \otimes \Pc$. In contrast to \cite{merkli2} the conjugate operator in the particle space will be defined as follows. We first diagonalize the non-negative part of $H_\p$ by means of   generalized eigenfunctions associated  to the positive (continuous) spectrum, the \textit{scattering functions}, which we recall in \Autoref{sec:scattering}. This will establish a unitary map $\Vc$ between the non-negative eigenspace of $H_\p$ and $L^2(\IR^3)$ with the property that $\Vc^* H_\p \Vc = \q^2$, where $\q = (\q_1, \q_2, \q_3)$ denotes the vector of multiplication operators with the respective components.  Let 
\[\Adil := \frac{1}{4} (\po \q + \q \po)\]
be the generator of dilations, where $\po :=  \i \nabla = (\i \partial_1,\i \partial_2,\i \partial_3) $
and where we used the notation 
 $\po \q := \sum_{j=1}^3 \po_j \q_j$ and  $\q \po  := \sum_{j=1}^3  \q_j \po_j$.
Then
\[
\Aps := \Vc ^* \Adil  \Vc
\]
has the effect that 
\[
\i [H_\p, \Aps] = \Vc^* \q^2 \Vc = \Pc H_\p,
\]
which is strictly positive on  $\ran \Pc$.  We combine $A_\f$ and $\Aps$ to an operator on $\H$ by
\[
\As = (\Aps \otimes \Id_\p - \Id_\p \otimes \Aps) \otimes \Id_\f + \Id_\p \otimes \Id_\p \otimes A_\f,
\]
which yields 
\[
\i[L_0, \As] = ( \Pc H_\p \otimes \Id_\p + \Id_\p \otimes \Pc H_\p ) \otimes \Id_\f + \Id_\p \otimes \Id_\p \otimes \Nf.
\]
As $\As$ is unbounded, it is necessary to use a virial theorem for the positive commutator method to work. We will indeed use the same abstract versions developed in  \cite{merkli1,merkli2} which are repeated in \Autoref{sec:abstract virial}. In order to be able to apply the virial theorem \autoref{thm:firstvirial} it is necessary that the commutators are bounded on the atomic space (see \eqref{thm:virthm:eq2} and \eqref{thm:virthm:eq3}). Thus one has to include a regularization in $\Aps$. The exact definition of $\As$ and of a regularized version $\Ase$, as well as the verification of the conditions for the virial theorems, can be found in \Autoref{section:verification}. 

For the space corresponding to the sum of the remaining three   projections in \eqref{eq:Lp not zero} we choose an operator $Q$ on $\H_\p \otimes \H_\p$ given as a bounded continuous function of $L_\p$, which vanishes at the origin. We add a suitable operator $T$ depending on the interaction and $\lambda$ to accomplish 
\[
\sc{\psi, (Q \otimes P_\Omega + T) \psi} = 0
\]
for all $\psi \in \ker L_\lambda$. Now  the distance between the essential and the discrete spectrum is strictly positive and  the  distances  between the distinct discrete eigenvalues of $H_\p$ are   bounded from below 
by a positive number. Therefore,  $Q \otimes P_\Omega$ is   strictly positive on 
\[
\ran ( \ind_{L_\p \not= 0}  \otimes P_\Omega ) .
\]
The operator $T$ will  be viewed as an  error term which will  be estimated by $\Nf$.

Finally, there will arise further error terms from the the commutator of the interaction with $\As$ and $A_0$, respectively. The general idea to control them, is  to  estimate them in terms of   $\Nf$ on $\FF$ and 
in terms of bounded terms on $\H_\p \otimes \H_\p \otimes \ran P_\Omega$, respectively. It is for the 
  latter that we need   the decompositions   \eqref{eq:ranpi}  and  \eqref{eq:Lp not zero} as well as 
 the corresponding positive operators mentioned  above. On $\ran \Pc \otimes \Pc$ we estimate them by $\po^{-2}$ and then use that by the uncertainty principle 
\[
\q^2 - \lambda \po^{-2} > 0
\]
for $\lambda > 0$ sufficiently small (cf. \cite[X.2]{rs2}). 

\section{Abstract Virial Theorems}
\label{sec:abstract virial}
In this section we recall the abstract virial theorems of \cite{merkli1,merkli2}. They are based on Nelson's commutator theorem, which can be used for proving self-adjointness of operators which are not bounded from below. An important notion will be that of a GJN triple. 
\begin{defn}[GJN triple]
Let $\HH$ be a Hilbert space, $\mathcal{D} \subset \HH $ a core for a self-adjoint operator $Y \geq \Id$, and $X$ a symmetric operator on $\mathcal{D}$. We say the triple
$(X,Y,\mathcal{D})$ satisfies the Glimm-Jaffe-Nelson (GJN) condition, or that $(X,Y,\mathcal{D})$ is a GJN-triple, if there is a constant $C < \infty$, such that
for all $\psi \in \mathcal{D}$:
\begin{align}
\| X \psi \| & \leq C \| Y \psi \| , \label{eq:gjn1} \\
 \pm \i \{ \sc{ X \psi , Y \psi} - \sc{ Y \psi , X \psi } &\leq C \sc{ \psi , Y \psi } .  \label{eq:gjn2}
\end{align}
\end{defn}

\begin{thm}[GJN commutator theorem, {\cite[Theorem X.37]{rs2}}] 
\label{thm:invariance_gjn_domain}
If $(X,Y,\mathcal{D})$ satisfies the GJN condition, then $X$ determines a self-adjoint operator (again denoted by $X$),
such that $\Def(X) \supset \Def(Y)$. Moreover, $X$ is essentially self-adjoint on any core for $Y$, and \eqref{eq:gjn1} is valid for all
$\psi \in \Def(Y)$.
\end{thm}
A consequence of  the GJN commutator theorem   is that the unitary group generated by $X$ leaves the domain of $Y$ invariant. The concrete formulation stated in the next theorem  is taken from \cite{merkli1}.
\begin{thm}[Invariance of domain, \cite{froehlich_invariance}]
Suppose $(X,Y,\D)$ satisfies the GJN condition. Then, for all $t \in \IR$, $e^{\i tX}$ leaves $\Def(Y)$ invariant, and there is a constant $\kappa \geq 0$ such that
\muu
{
\nn{Y e^{\i t X} \psi } \leq e^{\kappa \abs t} \nn{Y \psi}, \qquad \psi \in \Def(Y).
}
\end{thm}

Based on the GJN commutator theorem, we  can now describe  the setting for a general virial theorem. Suppose one is given  a self-adjoint operator $\Lambda \geq \Id$ with
core $\mathcal{D} \subset \HH$, and let  $L, A, N, D, C_n$, $n \in \{0,1,2,3\}$, be symmetric operators  on $\mathcal{D}$  satisfying the relations 
\begin{align} \label{abstractviri1}
&\sc{ \varphi, D \psi } = \i ( \sc{ L \varphi, N \psi } - \sc{ N \varphi, L \psi } )
\end{align}
and
\begin{align}
& C_0 = L, \\
&\sc{ \varphi, C_{n+1} \psi } = \i ( \sc{ C_{n} \varphi, A \psi  } - \sc{ A \varphi, C_{n} \psi } ) , \quad n \in \{0,1,2\},  \label{eq:defn_Cn}
\end{align}
where $\varphi, \psi \in \mathcal{D}$. Furthermore we shall assume:
\begin{enumerate}[label=(V\arabic*)]
\item \label{vi:all gjn} $(X,\Lambda,\mathcal{D})$ satisfies the GJN condition for $X=L,N,D,C_n$, $n\in \{0,1,2,3\}$. Consequently all these operators determine self-adjoint operators, which we denote by the same letters.
\item \label{vi:eitA invariant}  $A$ is self-adjoint, $\mathcal{D} \subset \Def(A)$, and $e^{\i  t A}$ leaves  $\Def(\Lambda)$ invariant.
\end{enumerate}
\begin{thm}[Abstract virial theorem, {\cite[Theorem 3.2]{merkli1}}] \label{thm:firstvirial}   Let  $\Lambda \geq \Id$ be a self-adjoint operator in $\HH$ with
core $\mathcal{D} \subset \HH$, and let  $L, A, N, D, C_n$, $n \in \{0,1,2,3\}$, be  symmetric on $\mathcal{D}$ satisfying relations  \eqref{abstractviri1}--\eqref{eq:defn_Cn}. Assume  \ref{vi:all gjn}  and \ref{vi:eitA invariant}. Furthermore, 
assume that $N$ and $e^{\i  t A}$ commute, for all $t \in \R$, in the strong sense on $\mathcal{D}$, and that there  exist $ 0 \leq p < \infty$ and $C < \infty$ such that
\begin{align}
 \nn{D \psi} &\leq C \nn{N^{1/2} \psi} , \label{thm:virthm:eq1} \\
  \nn{C_1\psi}  &\leq C\nn{ N^p \psi}, \label{thm:virthm:eq2}  \\
  \nn{C_3\psi}  &\leq C \nn{N^{1/2}  \psi}, \label{thm:virthm:eq3}
\end{align}
for all $\psi \in \mathcal{D}$. Then, if $\psi \in \Def(L)$ is an eigenvector of $L$, there is a sequence
of approximating eigenvectors $( \psi_n )_{n \in \mathbb{N}}$ in $\Def(L)\cap \Def(C_1)$ such that $\lim_{n \to \infty} \psi_n = \psi$ in $\HH$, and
\[
\lim_{n \to \infty} \sc{ \psi_n, C_1 \psi_n } = 0 .
\]
\end{thm}
%

\section{Definition of the Commutator and Verification of the Virial Theorems}
\label{section:verification}
In this section  we introduce generalized eigenstates associated to scattering states of the atomic Hamiltonian $H_\p$. 
Using  these  scattering states  we will then define explicit  realizations for the operators  $L, A, N, D, C_n$, $n \in \{0,1,2,3\}$, of \Autoref{sec:abstract virial}. 
We  then  verify  the assumptions of the abstract virial theorem  in order to obtain a   concrete virial theorem \autoref{th:result virial}.  This theorem 
 will be one of the two ingredients for  the proof of the main result of this paper. In this section we shall always assume that  the potential $V$ satisfies  \ref{ass:h1} and  \ref{ass:h2} and that 
\ref{assumption:I1}--\ref{assumption:I3} hold. 

\subsection{Scattering States}
\label{sec:scattering}
In this part we recall the theory of generalized eigenstates, which are associated 
to scattering states, and their corresponding spectral decomposition.  The scattering states $\phi(k,\cdot)$, $k \in \IR^3$,  can be  defined as generalized eigenvectors,
\[
(-\Delta + V) \phi(k,\cdot) = k^2 \phi(k,\cdot),
\]
or as solutions of the so-called Lippmann-Schwinger equation, 
\begin{align}
\label{eq:lippmann-schwinger}
\phi(k,x) = e^{\i k x} - \frac{1}{4\pi} \int_{\IR^3} \frac{e^{\i \abs k \abs{x-y}}}{\abs{x-y}} V(y) \phi(k,y) \d y.
\end{align}
We discuss their properties in the following proposition which is from \cite{ikebe}, see also 
 \cite[Theorem XI.41]{rs3}   and \cite{newton}. 
In particular,  the scattering functions can be used for a spectral decomposition of the continuous spectrum of $H_\p$. 
\begin{thm}[{ \cite[Theorem XI.41]{rs3},\cite{ikebe},\cite{newton} }]
\label{th:scattering functions properties}
Suppose   \ref{ass:h1} and  \ref{ass:h2} hold.

\begin{enumerate}[label=(\alph*)]
\item For all $k \in \IR^3$ there exists a unique solution  $\phi(k,\cdot)$ of 
    \eqref{eq:lippmann-schwinger} which obeys $|V|^{1/2} \varphi(k ,\cdot ) \in L^2(\R^3)$. Moreover   for every $k \in  \R^3$  the function $x \mapsto \varphi(k, x)$ is continuous. 
\item  For  $f \in L^2(\IR^3)$ the generalized Fourier transform 
\[
(\Vc f)(k) := f^\#(k) := (2\pi)^{-3/2} \liim \int \overline{\phi(k,x)} f(x) \d x,
\]
where $\liim \int g(x) \d x :=  L^2$-$\lim_{R \to \infty} \int_{\abs x < R} g(x) \d x$, exists. 
\item We have $\ran \Vc = L^2(\IR^3)$ and for  $f \in L^2(\IR^3)$
\[
\nn{\Vc f} = \nn{\Pc f}. 
\]
In particular,  $\Vc$ is a partial isometry and  $\Vc|_{\ran \Pc} \: \ran \Pc \rightarrow L^2(\IR^3)$ is a unitary operator, and $\Vc \Vc^* = \Id$. 
\item 
 For  $f \in L^2(\IR^3)$ we have the spectral decomposition
\[
(\Pc f)(x) = \liim (2\pi)^{-3/2} \int f^\#(k) \phi(k,x) \d k .
\]
\item 
If $f \in \Def(H_\p)$, then
\[
(H_\p f)^\#(k) = k^2 f^\#(k),
\]
in other words, $\Vc H_\p \Vc^* = \q^2$. 
\end{enumerate}
\end{thm}
The basic strategy of the  proof of the theorem is to  introduce  the method of modified square integrable scattering functions, which can be found in \cite{ikebe,rs3}, originally developed by \textsc{Rollnik}.
In particular, one introduces  the so-called modified Lippmann-Schwinger equation
\begin{equation} \label{eq:modilip} 
\phit(k,x) = \abs{V(x)}^{1/2} e^{\i k x} + (L_{\abs k} \phit(k,\cdot))(x), \quad k, x  \in \R^3 , 
\end{equation} 
where
\begin{equation} \label{eq:deofL}
L_\kappa \psi(x) :=  - \frac{1}{4\pi} \int \frac{\abs{V(x)}^{1/2} e^{\i \kappa \abs{x-y}} V(y)^{1/2} }{\abs{x-y}} \psi(y) \d y, \quad \kappa \geq 0, 
\end{equation} 
and $V(y)^{1/2} :=   \abs{V(y)}^{1/2} \sgn V(y) $. 
We note that is elementary to see that  by \ref{ass:h1} the operator $L_\kappa$ is a Hilbert-Schmidt operator, see \cite[Theorem 1.22]{simon}. 
 If  for fixed $k \in \IR^3$ the function   $\varphi(k,\cdot)$ obeys  \eqref{eq:lippmann-schwinger} 
and $\tilde{\varphi}(k,\cdot ) := |V|^{1/2} \varphi(k,\cdot)$ is an $L^2$-function, then $\tilde{\varphi}(k,\cdot)$   obeys  \eqref{eq:modilip}, provided  \ref{ass:h1} holds (in fact it  holds for a larger class of potentials  \cite{rs3}).  
   On the other hand,  if for a fixed  $k \in \R^3$, the modified  Lippmann-Schwinger equation \eqref{eq:modilip}  has a unique $L^2$-solution $\tilde{\varphi}(k,\cdot)$, then, as outlined in \cite{rs3}, the original
Lippmann-Schwinger equation  \eqref{eq:lippmann-schwinger}  has a unique solution $\varphi(k,\cdot)$ satisfying $|V|^{1/2} \varphi(k,\cdot)  \in L^2(\R^3)$. It is given by 
\begin{align}
\label{eq:recovery phi}
\phi(k,x) = e^{\i k x} - \frac{1}{4\pi} \int \frac{e^{\i \abs k \abs{x-y}}}{\abs{x-y}} V(y)^{1/2} \phit(k,y) \d y . 
\end{align}

\begin{proof} 
By the assumptions  on the potential the operator $L_\kappa$ defined as in \eqref{eq:deofL} is a Hilbert-Schmidt operator for all $\kappa \geq 0$. 
 Let $\mathcal{E}$ denote the set of all  $\kappa \in (0,\infty)$ such that $\psi = L_\kappa \psi$ has a nonzero
solution in $L^2$.    We claim that    $\mathcal{E}$ is the empty set if  \ref{ass:h1}  holds. 
To this end, let  $\kappa > 0$ and assume   $\tilde{\varphi}  = L_\kappa \tilde{\varphi}$ for some $L^2$-function $\tilde{\varphi}$. Now consider  
\begin{align*}
\phi(x) :=  - \frac{1}{4\pi} \int \frac{e^{\i \kappa \abs{x-y}}}{\abs{x-y}} V(y)^{1/2} \phit(y) \d y  =   - \frac{1}{4\pi} \int \frac{e^{\i \kappa  \abs{x-y}}}{\abs{x-y}} V(y) \phi(y) \d y . 
\end{align*}
It follows that $\varphi(x) = o(|x|^{-1})$ as  $|x| \to \infty$, and that 
$
- \Delta \varphi + V \varphi = \kappa^2 \varphi  . 
$
According \cite{kato1} this implies that $\varphi$ vanishes identically outside 
a sufficiently large sphere. Hence by the unique continuation theorem it follows that $\phi = 0$ and $\tilde{\phi} = \abs{V}^{1/2} \phi = 0$. This is a contradiction, and  
we conclude that  the set $\mathcal{E}$ is empty for potentials which we consider. 
Thus by the Fredholm alternative whenever  $k\neq 0$
there is  a unique $L^2$ solution  $\tilde{\phi}$ of the modified Lippmann-Schwinger equation 
$
 \tilde{\phi}(x) = |V|^{1/2} e^{ikx} + (L_{|k|}\tilde{\phi}      )(x)   
$.
As mentioned above  it  follows that  the original 
Lippmann-Schwinger equation     \eqref{eq:lippmann-schwinger}   has a unique solution $\varphi$ satisfying $|V|^{1/2}\varphi \in L^2$
given by \eqref{eq:recovery phi}.  In the case $k=0$ we argue  analogously using   \ref{ass:h2}. 
This shows the first part of  (a).  The continuity follows in view of \eqref{eq:lippmann-schwinger}   from dominated convergence. 
(b)--(e) now follow  from  \cite[Theorem XI.41]{rs3}, where we have seen in \autoref{th:H implications} that the essential and the absolutely continuous spectrum of $H_\p$ coincide.
\end{proof}

Furthermore, we can extend $\Vc$ to a unitary operator by including the eigenfunctions into consideration. 
For this, we denote by $\phi_n$, $n=1, \ldots, N$, the eigenvectors of $H_\p$.
We define
\[
\Vd \: L^2(\IR^3) \longrightarrow \ell^2(N), \qquad (\Vd\psi)_i := \sc{\phi_i, \psi}. 
\]
Obviously $\Vd|_{\ran \Pd} \: \ran \Pd \rightarrow \ell^2(N) $ is a unitary operator  and $\Vd|_{\ran \Pc}  = 0$. Thus,
\begin{align}
\label{eq:Vcd}
\Vcd := \Vd \oplus \Vc \: L^2(\IR^3) \longrightarrow \ell^2(N) \oplus L^2(\IR^3)
\end{align}
is unitary.

\subsection{Setup for the Virial Theorems}
First, we describe the setting on the particle space $\H_\p$. We consider a dense subspace given by
\begin{equation} \label{eq:deofdomainD} 
\Dps := \Vc^* \Dp \oplus \ran \Pd .
\end{equation} 
Note that $\Dps$ is dense since $\Vc^* \Dp \subseteq \ran \Pc$ is dense in $\ran \Pc$.
Now, based on the definition of the generator of dilations,
\[
\Adil = \frac{1}{4} (\q \po + \po \q),
\]
we define on $\Dps$ the conjugate operator $\Aps$ and a regularized version $\Apse$,
\[
\Aps := \Vc^*  \Adil  \Vc, \qquad\Apse :=  \Vc^*  \ee \Adil \ee  \Vc, 
\]
where
\[
\ee(k) := e^{-\epsilon k^2}.
\]
Note that $\eta_0 \equiv 1$ and $\Aps = A_\p^{(0)}$. 
It is clear that both $H_\p$ and $\Apse$, $\epsilon \geq 0$, leave $\Dps$ invariant. Thus we can define $\ad^{(n)}_{\Apse}(H_\p)$ on $\Dps$ for all $n \in \IN$. Furthermore, the bounding operator is chosen as 
\begin{align*}
\Lps &:= \Vc^* (\q^2  + \po^2) \Vc + \Id_\p.
\end{align*}
Next, on the field space we set
\begin{align}
A_\f &:= \dG(\i \partial_u), \label{eq:defofAf}  \\
\Lambda_\f &:= \dG(\mult{u} ^2 + 1).  \label{eq:defofLambdaf} 
\end{align}
Now, we can define on the dense subspace of the composite space $\H$,
\begin{align}
\Ds &= \Dps \otimes \Dps \otimes  \Df,  \label{eq:defofmathcalD} 
\end{align}
the operators
\begin{align}
\Ls &= \Lps \otimes \Id_\p  \otimes \Id_\f + \Id_\p \otimes \Lps \otimes \Id_\f + \Id_\p \otimes \Id_\p \otimes \Lambda_\f,   \label{eq:defofLambda}   \\
\Ase &= (\Apse \otimes \Id_\p - \Id_\p \otimes \Apse) \otimes \Id_\f + \Id_\p \otimes \Id_\p \otimes A_\f, \qquad \epsilon \geq 0,  \nonumber \\
D &= \i [L_\lambda,\Nfh],  \label{eq:defofD} 
\end{align}
where $\Nfh := \Id_\p \otimes \Id_\p \otimes \Nf$.

For operators $X,Y$ with a dense domain $\D_0$ we define multiple commutators by $\ad^{(0)}_Y(X) = X$ and $\ad^{(n+1)}_Y(X) = \i [\ad^{(n)}_Y(X),Y]$ in the form sense on $\D_0 \times \D_0$, provided the  right-hand side is determined by a densly defined bounded or an essentially self-adjoint operator, in which case we  denote the corresponding  extension by the same symbol. Furthermore, we set $\ad_Y(X) := \ad^{(1)}_Y(X)$.

Now we define for $n \in \{1,2,3\}$, $\epsilon \geq 0$,
\begin{align}
C_n^{(\epsilon)} &:= \ad^{(n)}_{\Ase}(L_\lambda)  \label{eq:defofCnepsilon} \\  &= \delta_{n,1} \Id_\p \otimes \Id_\p \otimes N_\f + \ad_{\Apse}^{(n)}(H_\p) \otimes \Id_\p + (-1)^{n+1} \Id_\p \otimes \ad_{\Apse}^{(n)}(H_\p) + \lambda W_n^{(\epsilon)},  \nonumber \\
\Cf{n} &:= \ad^{(n)}_{\Id_\p \otimes \Id_\p \otimes A_\f}(L_\lambda)  \nonumber \\  &=  \delta_{n,1} \Id_\p \otimes \Id_\p \otimes N_\f  + \lambda W_n^{(\f)} , \label{eq:defofCfn} 
\end{align}
where
\begin{align}
\label{eq:cn_formula 2}
W_n^{(\epsilon)} &:= \ad^{(n)}_{\Ase}(\Phi(I)) =  \Phi( I_n^{(\epsilon)} (u,\Sigma)), \\  
W_n^{(\f)} &:= \ad^{(n)}_{\Id_\p \otimes \Id_\p \otimes A_\f}(\Phi(I)) =  \Phi( I_n^{(\f)} (u,\Sigma)),  \label{eq:cn_formula 3}
\end{align}
with
\begin{align*}
I_n^{(\epsilon)} (u,\Sigma) &:= \sum_{k=0}^n  \binom{n}{k}  \big(  (-\i \partial_u)^k \tau_\beta( \ad^{(n-k)}_{ \Apse }(G) ) \otimes \Id_\p  \\ &\qquad -  (-\i \partial_u)^k  e^{-\beta u  /2}   \Id_\p  \otimes \tau_\beta( \ad^{(n-k)}_{\Apse} (\G^*) )  \big), \\
I_n^{(\f)} (u,\Sigma) &:= \sum_{k=0}^n  \big(  (-\i \partial_u)^k \tau_\beta( G) \otimes \Id_\p    -  (-\i \partial_u)^k  e^{-\beta u  /2}   \Id_\p  \otimes \tau_\beta( \G^*)  \big),
\end{align*}
and we use the shorthand notation $I_n:= I_n^{(0)}$, $W_n := W_n^{(0)}$,  $C_1 := C_1^{(0)}$. We note that 
the above identites follow from a straightforward calculation. 
We will see in \autoref{th:G commutator conditions}  that the expressions in the field operators in \eqref{eq:cn_formula 2} and \eqref{eq:cn_formula 3} are indeed  well-defined and belong to $L^2(\IR \times \IS^2, \L(\H_\p))$. 
Furthermore, it will be proven in \autoref{th:gjn check}, that $C_n^{(\epsilon)}$ and $\Cf{n}$, $n \in \{1,2,3\}$, are in fact  essentially self-adjoint on $\Ds$ and we denote their self-adjoint extensions by the same symbols.
Moreover, it will be shown below in \autoref{th:total estimate steps} that $C_1$ is actually bounded from below. Thus, we can assign to $C_1$ a quadratic form $\qce$.

\subsection{Verification of the Assumptions of the Virial Theorems}
\label{subsec:verification assumptions}

In the given setting just described we can now start to prove the assumptions of the virial theorem \autoref{thm:firstvirial}. Above all, we have to check the GJN condition for the different commutators. The most difficult part will be the discussion of the interaction terms $W_n^{(\f)}$ and $W_n^{(\epsilon)}$,  $n \in \{1,2,3\}$, $\epsilon > 0$, and $W_1$. Here, the expressions in the field operators need to be sufficiently bounded. These bounds will be collected in the following proposition,  which will be poven in  \Autoref{sec:estimates scattering}. 
\begin{prop}
\label{th:G commutator conditions}
Let $\partial_u$ denote the weak  derivative of a $\L(\H_\p)$-valued function in the sense of the strong operator topology. 
For all $m \in \{0,1,2,3\}$, $n \in \IN_0$, $j \in \{1,2,3\}$, and for all $(u,\Sigma)$, $\epsilon \geq 0$, the operators
\begin{enumerate}[label=(\arabic*)]
\item 
\label{item:no derivative}
$\partial_u^m  \ad^{(n)}_{\Apse}(  \tau_\beta(G)(u,\Sigma)  )$,
\item 
\label{item:qj derivative}
$\partial_u^m   \ad_{\Vc^* \q_j \Vc}(  \ad^{(n)}_{\Apse}(  \tau_\beta(G)(u,\Sigma)  ) )$,
\item 
\label{item:pj derivative}
 $ \partial_u^m \ad_{\Vc^* \po_j \Vc}(  \ad^{(n)}_{\Apse}(  \tau_\beta(G)(u,\Sigma)   ) )$, 
\item
\label{item:pj additional}
 $ \partial_u^m  \Vc^* \po_j \Vc    \ad^{(n)}_{\Aps }(  \tau_\beta(G)(u,\Sigma)   )$,
\end{enumerate}
are well-defined, and the corresponding functions $\IR \times \IS^2 \rightarrow \L(H_\p)$ of $(u,\Sigma)$  belong to $L^2(\IR \times \IS^2, \L(H_\p))$.  Moreover, there exists a constant $C$ independent of $\beta$ such that for $n,m,s \in \{0,1\}$,
\begin{equation} \label{eq:mainineqbound} 
\nn{  \partial_u^m  (\Vc^* \po_j \Vc)^s    \ad^{(n)}_{\Aps }(  \tau_\beta(G)  ) }_{L^2(\IR \times \IS^2, \L(H_\p))} \leq C(1+ \beta^{-\frac{1}{2}}).
\end{equation} 
The result also holds true if we replace $G$ by $G^*$.
\end{prop}
With the help of Proposition  \ref{th:G commutator conditions}
we can now verify the necessary GJN conditions.
\begin{prop}
\label{th:gjn check}
The following triples are GJN:
\begin{enumerate}[label=(\arabic*)]
\item $(\Aps, \Lps, \Dps)$, \label{Aps Lps Dps}
\item $(\Apse, \Lps, \Dps)$, $\epsilon > 0$,

\item $(H_\p, \Lps, \Dps)$, \label{Hp Lps Dps}
\item $(L_\lambda, \Ls, \Ds)$, $\lambda \in \IR$,\label{Ll essentially self-adjoint}
\item $(\Nfh, \Ls, \Ds)$, \label{N Ls D}
\item $(D, \Ls, \Ds)$,
\item $(C_i^{(\epsilon)}, \Ls, \Ds)$,  $\epsilon > 0$, $i \in \{1,2,3\}$, \label{Cie Ls D}
\item $(\Cf{i}, \Ls, \Ds)$, $i \in \{1,2,3\}$, \label{Cf Ls D}
\item $(C_1, \Ls, \Ds)$.
\end{enumerate}
In particular, $L_\lambda$ is essentially self-adjoint on $\Ds$ for any $\lambda \in \IR$ due to \ref{Ll essentially self-adjoint}. Moreover, $D$, $C_1^{(\epsilon)}$, $C_3^{(\epsilon)}$, $\epsilon > 0$, and $C_1^{(\f)}$, $C_3^{(\f)}$ are relatively bounded by $\Nfh^{1/2}$. 
\end{prop}
\begin{proof}
\begin{enumerate}[label=(\arabic*),wide,  labelindent=0pt]
\item
Using that $(\Adil, \q^2 + \po^2, \Dp)$ is a GJN triple (cf. \cite{merkli2}) and $\Vc^*$ an isometry, we have, for $\psi \in \Dps$, 
\begin{align*}
\nn{ \Aps \psi} = \nn{\Vc^* \Adil \Vc \psi} = \nn{\Adil \Vc \psi} \leq C \nn{ (\q^2 + \po^2) \Vc  \psi} \leq  C' \nn{ \Lambda_\p \psi},
\end{align*}
for some constants $C,C'$, and 
\begin{align*}
\pm \i &(\sc{\Aps \psi, \Lps \psi} - \sc{\Lps \psi, \Aps \psi})  \\\ &=  \pm \i \left(\sc{\Adil \Vc \psi, (\q^2 + \po^2) \Vc \psi} - \sc{(\q^2 + \po^2)\Vc \psi, \Adil\Vc \psi}  \right)  \\
&\leq C \sc{\Vc \psi, (\q^2 + \po^2) \Vc \psi} \\
&\leq C \sc{\psi,  \Lps \psi}.
\end{align*}

\item
On $\Dps$ we have
\[
\Apse  = \Vc^* \ee \Adil \ee \Vc = \Vc^* \ee^2 \Adil \Vc + \Vc^* \ee [\Adil,\ee] \Vc. 
\]
The operator $\Vc^* \ee^2 \Adil \Vc$ is relatively bounded by $\Vc^* \Adil \Vc$, and thus also by $(\q^2 + \po^2 )\Vc$ as we have already seen in the proof of \ref{Aps Lps Dps}. Furthermore, as derivatives of $\ee$ are also bounded, the operator
\[
[\Adil,\ee] = \frac{1}{2} \sum_{j}  [ \po_j, \ee] \q_j 
\] 
is relatively bounded by $| \q |$, and thus $\ee [\Adil,\ee] \Vc$ is also relatively bounded by $(\q^2 + \po^2 )\Vc$ . This shows the first GJN condition. 

For the second one, we compute
\begin{align*}
\pm \i \left(\sc{\Apse \psi, \Lps \psi} - \sc{ \Lps \psi, \Apse \psi} \right) = \pm \sc{ \Vc \psi, \i [\q^2 + \po^2, \ee \Adil \ee] \Vc \psi}.
\end{align*}
Thus, it suffices to show $ \pm \i [\q^2 + \po^2, \ee \Adil \ee]  \leq C(\q^2 + \po^2)$ for some constant $C$  on $\Cci(\IR^3)$. First, 
\[
[\q^2, \ee \Adil \ee] = \ee [\q^2, \Adil] \ee = \sum_{j} \ee\q_j [\q_j, \po_j] \q_j \ee = -\i \ee \q^2 \ee.
\]
Second, in order to obtain $\pm \i [\po^2, \ee \Adil \ee]  \leq C(\q^2 + \po^2)$, we use the basic operator inequality
\begin{align}
\label{eq:op estimate}
A^*B + B^*A \leq A^* A + B^* B
\end{align}
to show that
\begin{align}
\label{eq:po_j ee Adil ee}
[\po_j, \ee \Adil \ee] =  \ee [\po_j, \Adil] \ee + [\po_j, \ee] \Adil \ee  + \ee \Adil[\po_j, \ee], \qquad j \in \{1,2,3\},
\end{align}
is relatively  bounded by $|\q|$. This is clear for the first term in \eqref{eq:po_j ee Adil ee} as $\ee [\po_j, \Adil] \ee \allowbreak = \frac{\i}{2} \ee^2 \q_j$, and for the two remaining ones it follows analogously to the proof of the first GJN condition, since second derivatives of $\ee$ are bounded.

\item
We have, with regard to the first GJN condition,
\begin{align}
\nn{H_\p \psi} = \nn{ \Vc^* \q^2 \Vc \psi + H_\p \Pd \psi} \leq C \nn{\Vc^* (\q^2 + \po^2) \Vc \psi} + \sup_{\lambda \in \sigma_\d(H_\p)} \abs \lambda \nn{\psi}, \label{ineqproofGJN}
\end{align}
as $\q^2$ is relatively bounded by $\po^2 + \q^2$. Furthermore,
\begin{align*}
&\sc{ H_\p \psi, \Lps \psi} - \sc{ \Lps \psi,   H_\p \psi} \\ &= \sc{\Vc^* \q^2 \Vc \psi, \Vc^* (\po^2 + \q^2) \Vc \psi} - \sc{\Vc^* (\po^2 + \q^2) \Vc \psi,  \Vc^* \q^2 \Vc \psi}  
\\ &= \sc{\q^2 \Vc \psi, \po^2 \Vc \psi} - \sc{ \po^2  \Vc \psi,  \q^2 \Vc \psi} .
\end{align*}
Using now that $\pm \i [\po^2,\q^2] \leq C(\q^2 + \po^2)$ for some constant $C$, we get also the second GJN condition. 

\item
As $H_\p$ is relatively bounded by $\Lps$  in view of  \eqref{ineqproofGJN}, $L_0 = (H_\p \otimes \Id_\p - \Id_\p \otimes H_\p) \allowbreak \otimes \Id_\f$ is relatively bounded by $\Ls$.
Next, by  \ref{assumption:I2}, \ref{assumption:I3}, and   Lemma \ref{lemma:gluing_integrable}  we know that the interaction terms $I^{(\l)}, I^{(\r)}   \in L^2(\IR \times \IS^2)$. Hence $W$ is bounded by $\Nfh^{1/2}$ and thus bounded by $\Id_\p \otimes \Id_\p \otimes \Lambda_\f^{1/2}$. Therefore, the first GJN condition is satisfied.

Next, as $(H_\p, \Lps, \Dps)$ is GJN, we get a constant $C$, such that for all $\psi \in \Ds$,
\[
\pm \i ( \sc{L_0 \psi, \Ls \psi} - \sc{ \Ls  \psi, L_0 \psi}) \leq C\sc{ \psi, (\Lps \otimes \Id_\p  + \Id_\p \otimes \Lps)  \otimes \Id_\f  \psi},
\]
which yields the second GJN condition for $L_0$.

Again by  \autoref{lemma:gluing_integrable} and \ref{assumption:I1}--\ref{assumption:I3} we know that $(u,\Sigma) \mapsto (u^2 + 1) I^{(\l)}(u,\Sigma)$ and $(u,\Sigma) \mapsto (u^2 + 1) I^{(\r)}(u,\Sigma)$  are in $L^2(\IR \times \IS^2, \L( \H_\p \otimes \H_\p))$. We have
\begin{align*}
&\abs{ \sc{ \Phi( I^{(\l)}) \psi ,  \Id_\p \otimes \Id_\p \otimes \Lambda_\f \psi } - \sc{  \Id_\p \otimes \Id_\p \otimes \Lambda_\f \psi , \Phi( I^{(\l)})  \psi } }
\\ &\qquad = \abs{ \sc{\psi,(  a\left(( \mult{u} ^2 + 1) I^{(\l)}) - a^*(( \mult{u} ^2 + 1) I^{(\l)}) \right) \psi} } 
\\ &\qquad\leq C \nn{\Id_\p \otimes \Id_\p \otimes \Nf^{1/2} \psi} \nn{\psi}
\\ &\qquad\leq C' \sc{ \psi, \Ls \psi}
\end{align*}
for some constants $C,C'$. The same thing can be shown for the commutator with $\Phi(I^{(\r)})$. 

It remains to consider the commutator of $W$ with the $\Lps$ terms. One has to show that the  commutators
\[
[\Phi(I^{(\l)} ), \Vc^* (\po^2 + \q^2) \Vc], \qquad [\Phi(I^{(\r)} ), \Vc^* (\po^2 + \q^2) \Vc]
\]
are form-bounded by $\Lps$, that is, in the sense of \eqref{eq:gjn2}. This follows  from \autoref{th:G commutator conditions} since  we can write in the weak sense on $\Ds$,
\[
[\Phi(I^{(\l)} ), \Vc^* \po^2 \Vc] = \sum_{j} [\Phi(I^{(\l)} ), \Vc^* \po_j \Vc]  \Vc^* \po_j \Vc +   \Vc^* \po_j \Vc [\Phi(I^{(\l)} ), \Vc^* \po_j \Vc] 
\]
and analogously for terms involving  $\Phi(I^{(\r)})$ as well as  $\Vc^* \q^2 \Vc$. 

\item

Clearly, $\nn{ \Nf \psi } = \nn{\dG(1) \psi} \leq \nn{\dG(u^2 +  1) \psi}$ for any $\psi \in \Df$, which shows that $\Nfh$ is relatively bounded by $\Ls$ on $\Ds$. Furthermore, $[\Nfh, \Lambda] = 0$ on $\Ds$. 

\item
We have
\[
D =  \frac{\i \lambda}{\sqrt 2} ( a(I) - a^*(I) ).
\]
Thus, the proof works as the one for $L_\lambda$. 
\item
We first consider the  atomic part. One can show by induction that for all $n \in \IN$ there exists $f \in \S(\IR^3)$ such that 
\begin{align}
\label{eq:f_1 P f_2 form}
\ad^{(n)}_{\Apse}(H_\p) = \Vc^* f(\q) \Vc.
\end{align} 
Clearly, for $n=1$,
\begin{equation}    \label{eq:commaphp} 
\ad_{\Apse}(H_\p) = \i \Vc^* [\q^2, \ee \Adil \ee] \Vc = \i \Vc^*  \ee  [\q^2, \Adil ] \ee\Vc = \Vc^*  \ee  \q^2 \ee\Vc, 
\end{equation} 
which has the form \eqref{eq:f_1 P f_2 form}. Next, as $[f(\q), \po_j] = -\i \partial_j f(\q)$,
\begin{align*}
[ \Vc^* f(\q)  \Vc, \Apse] &= \Vc^*\ee  [ f(\q) ,  \Adil]  \ee \Vc = \frac{1}{2} \sum_{j} \Vc^*\ee  [ f(\q),  \po_j ] \q_j  \ee \Vc 
\end{align*}
yields again  the form \eqref{eq:f_1 P f_2 form}.

In particular, \eqref{eq:f_1 P f_2 form} implies that $\ad^{(n)}_{\Apse}(H_\p)$, $n \in \IN$, are bounded and
\[
\pm \i [\ad^{(n)}_{\Apse}(H_\p), \Lambda_\p] \leq C \Lambda_\p
\] 
on $\Dps$ for some constant $C$, since
\[
\pm \i [ f(\q),\q^2 + \po^2] = \pm \i [ f(\q),\po^2] \leq C(\q^2 + \po^2),
\]
because  $[ f(\q),\po_j]$, $j \in \{1,2,3\}$, is bounded.

This together with \ref{N Ls D} implies that $(\ad^{(n)}_\Ase (L_0), \Ls, \Ds)$, $n\in \{1,2,3\}$, are GJN triples.  

It remains to verify the GJN conditions for  $(W_n^{(\epsilon)}, \Ls, \Ds)$, $n\in \{1,2,3\}$. Analogously to the proof \ref{Ll essentially self-adjoint} for the GJN condition of $L_\lambda$ we have to show that the expressions in the field operators and the commutators with $\Vc^* \q_j \Vc$ and $\Vc^* \po_j \Vc$, $j \in \{1,2,3\}$ are integrable. That is, we have to show that
\begin{align*}
(u,\Sigma) &\mapsto  \partial_u^m \tau_\beta(\ad^{(n-m)}_{\Apse} (G))(u,\Sigma), \\
(u,\Sigma) &\mapsto  \partial_u^m [\tau_\beta(\ad^{(n-m)}_{\Apse} (G))(u,\Sigma), X],
\end{align*}
$X \in \{ \Vc^* \q_j \Vc, \Vc^* \po_j \Vc   \: j \in \{1,2,3\}\}$, 
$n \in \{ 1,2,3 \}$, $m \in \{ 0, \ldots,n \}$, 
are in $L^2(\IR \times \IS^2)$. This follows from \autoref{th:G commutator conditions}.

\item
Analogously to the proof of \ref{Cie Ls D} it suffices to show that
\begin{align*}
(u,\Sigma) &\mapsto  \partial_u^n \tau_\beta(G)(u,\Sigma), \\
(u,\Sigma) &\mapsto  \partial_u^n [\tau_\beta(G)(u,\Sigma), X], 
\end{align*}
$X \in \{ \Vc^* \q_j \Vc, \Vc^* \po_j \Vc   \: j \in \{1,2,3\}\}$, 
$n \in \{ 1,2,3 \}$, 
are in $L^2(\IR \times \IS^2)$. This follows again from \autoref{th:G commutator conditions}.

\item
We first consider again the free part. We have on $\Dps$,
\begin{align*} 
\i [H_\p, \Aps] =  \i [\Vc^* \q^2 \Vc, \Aps] =    \i  \Vc^* [\q^2, \Adil] \Vc =  \Vc^* \q^2 \Vc.
\end{align*}
Then we can show as in the proof of \ref{Hp Lps Dps} that also $(\ad_{\Aps}(H_\p), \Lps, \Dps)$, is GJN and so is $(\ad_\As (L_0), \Ls, \Ds)$.

It remains to verify the GJN conditions for  $(W_1, \Ls, \Ds)$. As in \ref{Cie Ls D} and \ref{Cf Ls D} one has to show that
\begin{align*}
(u,\Sigma) &\mapsto  \partial_u^n \tau_\beta(\ad^{(1-n)}_{\Aps} (G))(u,\Sigma), \\
(u,\Sigma) &\mapsto  \partial_u^n [\tau_\beta(\ad^{(1-n)}_{\Aps} (G))(u,\Sigma), X], 
\end{align*}
$X \in \{ \Vc^* \q_j \Vc, \Vc^* \po_j \Vc   \: j \in \{1,2,3\}\}$, 
$n \in \{0,1\}$,
are in $L^2(\IR \times \IS^2)$, which follows again from \autoref{th:G commutator conditions}. \qedhere 
\end{enumerate} 
\end{proof}
The statements of \autoref{th:gjn check} allow the application of the virial theorem \autoref{thm:firstvirial} for the regularized conjugate operator $\Ase$, $\epsilon >0$. In order to remove the regularization and transfer the result to $C_1$, one has to consider the limit $\epsilon \to 0$ for the corresponding quadratic forms $\qcee$ of $C_1^{(\epsilon)}$. This is the content of the following lemma. 
\begin{lemma}
\label{th:W1e lemma}
Assume that $\psi \in \D(\Nfh^{1/2})$ and $\qcee(\psi) \leq 0$ for all $  \epsilon \in (0,1)$. Then $\psi \in \D(\qce)$, and 
\begin{align}
\label{eq:qce to show}
\qce(\psi) = \lim_{\epsilon \to 0} \qcee(\psi).
\end{align}
\end{lemma}
\begin{proof} Let us first recall that by  definition 
\begin{align}
C_1^{(\epsilon)}   &=  \Id_\p \otimes \Id_\p \otimes N_\f + \ad_{\Apse}(H_\p) \otimes \Id_\p +  \Id_\p \otimes \ad_{\Apse}(H_\p) + \lambda W_1^{(\epsilon)} .  \label{defofceeeeee} 
\end{align}
\underline{Step 1:}  We show that 
 \begin{equation} \label{eq:dompropofpsi} 
\psi \in \D(\Vc^* | \q | \Vc \otimes \Id_\p +    \Id_\p \otimes \Vc^* | \q | \Vc) = \D(\qce).
\end{equation} 
Step 1 will follow once we have established  that 
\begin{align}
\label{eq:ee q2 ee uniformly bounded}
\sc{\psi, ( \Vc^* \ee \q^2 \ee \Vc \otimes \Id_\p + \Id_\p \otimes \Vc^* \ee \q^2\ee  \Vc) \otimes \Id_\f \psi} \leq C
\end{align}
for a constant $C$ independent of $\epsilon > 0$.
To this end, we use that by  standard estimates for creation and annihilation operators, we obtain for any $\delta > 0$,
\begin{align}
\label{eq:W1pe estimate}
\pm \i W_1^{(\epsilon)} \leq \frac{1}{\delta}  \Nfh + \delta w_{1}^{(\epsilon)} \otimes \Id_\f
\end{align}
in the form sense on $\D(\Nfh^{1/2})$, where we introduced the  follwoing 
 bounded operators on $\H_\p \otimes \H_\p$,
\begin{align*}
w_{1}^{(\epsilon)} &:=  \int I_1^{(\epsilon)} (u,\Sigma)^* I_1^{(\epsilon)}(u,\Sigma) \d(u, \Sigma), \qquad \epsilon \geq 0 . 
\end{align*}
To estimate this expression we use that 
\[
\int \ad_{\Apse}( \tau_\beta( G) (u,\Sigma))^*   \ad_{\Apse}(  \tau_\beta( G) (u,\Sigma)) \d (u,\Sigma)  \leq  C (  \Vc^* \ee \q^2 \ee \Vc  + \Id) ,  
\]
in the form sense for some constant $C$ independent of $\epsilon$, where we multiplied out the commutators, used \eqref{eq:op estimate} and the fact that the functions $(u,\Sigma) \mapsto \partial_u \tau_\beta( G) (u,\Sigma)$ and  $(u,\Sigma) \mapsto \Vc^* \po_j \Vc \tau_\beta( G)$ belong to $L^2(\IR \times \IS^2, \Lb(\H_\p))$ due to \autoref{th:G commutator conditions}. 
 This yields 
\begin{align}
\label{eq:upper bound}
w_{1}^{(\epsilon)} \leq C(\Vc^* \ee \q^2 \ee \Vc \otimes \Id_\p + \Id_\p \otimes \Vc^* \ee \q^2\ee  \Vc + \Id_\p \otimes \Id_\p).
\end{align}
Now using  \eqref{eq:commaphp} and  \eqref{eq:W1pe estimate}  to estimate  \eqref{defofceeeeee}    we obtain in the form sense on $\D(\Nfh^{1/2})$ for any $\epsilon > 0$ 
\begin{align*}
C_1^{(\epsilon)}  &\geq  (1 - C \abs \lambda \delta) (\Vc^* \ee \q^2 \ee \Vc \otimes \Id_\p + \Id_\p \otimes \Vc^* \ee \q^2 \ee \Vc + \Id_\p \otimes \Id_\p) \otimes \Id_\f \\  &\qquad + \left(1 - \frac{\abs \lambda }{\delta} \right) \Nfh. 
\end{align*}
Making $\delta > 0$ sufficently small such that $C \abs \lambda  \delta < 1$ and using that by assumption  $C_1^{(\epsilon)}  \leq 0$ and $\psi \in \D(\Nfh^{1/2})$, we arrive at \eqref{eq:ee q2 ee uniformly bounded}.

\vspace{0.4cm} 

\noindent 
\underline{Step 2:}  We have  \eqref{eq:qce to show}. 

\vspace{0.4cm} 

\noindent
First observe,  that by dominated convergence, we have 
\[
\nn{| \q |  \ee \Vc \psi} \longrightarrow \nn{ | \q | \Vc \psi}, ~\epsilon \to 0 .
\]
Thus Step 2 will follow once we have shown that   
\begin{align}
\label{eq:second to prove}
\sc{\psi, W_1^{(\epsilon)} \psi} \to \sc{\psi, W_1 \psi}.
\end{align}
Thus we have to show the convergence of  the field operator of a  commutator. To this end, we note that 
from  \autoref{th:G commutator conditions} we know that 
\[
(u,\Sigma) \mapsto  \Vc^* \po_j \Vc  H(u,\Sigma), 
\]
for  $H \in \{ \tau_\beta( G), \partial_u \tau_\beta( G) \}$, 
are in  $L^2(\IR \times \IS^2, \L(\H_\p))$. Now,
\begin{align}
\Apse H(u,\Sigma)  &= \Vc^* \ee \Adil \ee \Vc H(u,\Sigma) \nonumber \\
 &=   \Vc^* \ee \frac{1}{4}(3 \i + 2 \q \po)  \ee \Vc H(u,\Sigma) \nonumber \\
 &= \Vc^* \ee \frac{1}{4}\left(3 \i \ee  + 2 \sum_{j=1}^3 \q_j ( \i  ( \partial_{k_j} \ee )  + \ee \po_j)  \right)  \Vc H(u,\Sigma). \label{eq:Ape H}
\end{align}
Observe that   $\ee$  and  $ \partial_{k_j} \ee $  are  bounded uniformly in $\epsilon \in (0,1)$.  From  \eqref{eq:Ape H} we see  using $\psi \in \D( \Id_\p \otimes \Nf)$  that 
\begin{align*}
& \sc{ \psi ,  a^*( \Apse H )  \psi }  \\
  &= \frac{3}{4}  \sc{ \Vc^*  \eta_\epsilon^2   \Vc \otimes \Id_\f  \psi ,   a^*( \i  H ) \psi} + \frac{1}{2}  \sum_{j=1}^3  \sc{ \Vc^*   ( \partial_{k_j} \ee )  \q_j  \eta_\epsilon   \Vc \otimes \Id_\f \psi , a^*( \i  H) \psi}\\
& \quad  + \frac{1}{2}  \sum_{j=1}^3  \sc{ \Vc^*  \eta_\epsilon  \q_j  \eta_\epsilon   \Vc \otimes \Id_\f \psi , a^*(\Vc^* \po_j \Vc H) \psi }\\
&\to   
 \frac{3}{4}  \sc{ \Vc^*    \Vc \otimes \Id_\f \psi ,   a^*( \i  H ) \psi} 
  + \frac{1}{2}  \sum_{j=1}^3  \sc{ \Vc^*    \q_j    \Vc \otimes \Id_\f \psi , a^*(\Vc^* \po_j \Vc H) \psi } , 
\end{align*}
where for the limit we used that   $ \partial_{k_j} \ee $ tends to zero as $\epsilon \downarrow 0$, 
 \eqref{eq:dompropofpsi},  and dominated convergence. Similarly, we find 
\begin{align*}
& \sc{ \psi ,  a^*( H  \Apse  )  \psi }   \\ &\to   
 \frac{3}{4}  \sc{ a( \i  H )  \psi ,    \Vc^*    \Vc \otimes \Id_\f  \psi} 
  + \frac{1}{2}  \sum_{j=1}^3  \sc{a( H \Vc^* \po_j \Vc )  \psi ,  \Vc^*    \q_j    \Vc  \otimes \Id_\f \psi } , 
\end{align*}
Thus we find 
\begin{align*}
 &\i \sc{ \psi ,  a^*(  \ad_{\Apse}(H)  )  \psi }\\
&\to    \frac{3}{4}  \sc{ \Vc^*    \Vc \otimes \Id_\f \psi ,   a^*( \i  H ) \psi} 
  + \frac{1}{2}  \sum_{j=1}^3  \sc{ \Vc^*    \q_j    \Vc \otimes \Id_\f \psi , a^*(\Vc^* \po_j \Vc H) \psi } \\
& \quad -  \frac{3}{4}  \sc{ a( \i  H )  \psi ,    \Vc^*    \Vc  \otimes \Id_\f \psi} 
  - \frac{1}{2}  \sum_{j=1}^3  \sc{a ( H \Vc^* \po_j \Vc )  \psi ,  \Vc^*    \q_j    \Vc \otimes \Id_\f \psi } \\
& \quad =  \i  \sc{ \psi ,  a^*(  \ad_{\Aps}(H(u,\Sigma))  )  \psi } , 
\end{align*}
where the last  equality follows  by verifying  the 
identiy on the dense space $\Dps$, defined in  \eqref{eq:deofdomainD}, using a straightforward 
calculation,  and  then extending it to $\psi$   using  \autoref{th:G commutator conditions}. This shows \eqref{eq:second to prove} in view of the definition of $W_1^{(\epsilon)}$ and $W_1$, see \eqref{eq:cn_formula 2}.
\end{proof}

Now we can prove the main result of this section, the concrete virial theorem for our setting.
\begin{thm}[Concrete virial theorem]
\label{th:result virial}
Assume there exists $\psi \in \Def(L_\lambda)$ with $L_\lambda \psi = 0$. Then $\psi \in \Def(\qce)$ and $\qce(\psi) \leq 0$.
\end{thm}
\begin{proof} 
\underline{Step 1:}  We have $\psi \in \D(\Nfh^{1/2})$. 

\vspace{0.4cm} 

To show Step 1, we apply the virial theorem  for $\Lambda$ as in \eqref{eq:defofLambda}, 
   $\mathcal{D}$ as in \eqref{eq:defofmathcalD},  as well as the operators 
  $L=L_\lambda$,  $A =  \Id_\p \otimes \Id_\p \otimes A_\f$, $N=\Nfh +1$, $D$ as in \eqref{eq:defofD}, $C_n = \Cf{n}$, $n \in \{0,1,2,3\}$ as in \eqref{eq:defofCfn}, which are symmetric on $\mathcal{D}$
and satisfy \eqref{abstractviri1}--\eqref{eq:defn_Cn} as a consequence of the 
 definition.  Next we  verify the additional  assumptions
of the virial theorem.  
We have shown  in \autoref{th:gjn check} that \ref{vi:all gjn} is satisfied.  Furthermore, on $\Df$, 
\muu
{
\Lambda_\f e^{\i t A_\f} = e^{\i t A_\f} (\Lambda_\f + \dG(2\mult{u}t +  t^2)), \qquad t \in \IR.
}
Thus, for some fixed $t$ there is a constant $C$ such that $\nn{ \Lambda_\f e^{\i t A_\f} \psi} \leq C\nn{ \Lambda_\f \psi}$ for all $\psi \in \Df$. As $\Df$ is a core for $\Lambda_\f$, we find $e^{\i t A_\f} \Def(\Lambda_\f) \subseteq \Def(\Lambda_\f)$.  
Hence, we conclude that $\Def(\Lambda)$ is invariant under the unitary group associated to $A$ and hence \ref{vi:eitA invariant} holds.
 Furthermore, \eqref{thm:virthm:eq1}--\eqref{thm:virthm:eq3}  are satisfied since $D$, $\Cf{1}$ and $\Cf{3}$ are relatively bounded by $\Nf^{1/2}$ (by \autoref{th:gjn check}).
Thus we can apply the virial theorem, and obtain  a sequence $(\psi_n)$ in $\Def(L_\lambda) \cap \Def(\Cf{1})$ such that $\psi_n \to \psi$ and 
\[
\lim_{n \to \infty} \sc{\psi_n, \Cf{1} \psi_n} = 0.
\]
Now,  it follows from lower semi-continuity of closed quadratic forms that $\psi$ is in the form domain of $\Cf{1}$. As $W_1^{(\f)}$ is relatively bounded by $\Nfh^{1/2}$ (as shown in the proof of \autoref{th:gjn check} \ref{Cf Ls D}), we conclude that $\psi \in \D(\Nfh^{1/2})$.

\vspace{0.4cm} 

\noindent 
\underline{Step 2:}  $\psi \in \Def(\qcee)$ and 
$\qcee(\psi) \leq  0 $.

\vspace{0.4cm} 

To show Step 2, we apply the virial theorem  once more with  $\Lambda$, 
   $\mathcal{D}$,   $L$,  and   $N$ as in Step 1. But now with  
$A = \Ase$, $\epsilon > 0$ and    $C_n =  C_n^{(\epsilon)}$, $n \in \{0,1,2,3\}$ as in \eqref{eq:defofCnepsilon}.
Again,  the  operators are  symmetric on $\mathcal{D}$
and satisfy \eqref{abstractviri1}--\eqref{eq:defn_Cn} as a consequence of the 
 definition.  Let us now   verify the additional  assumptions
of the virial theorem.  
 Again, by \autoref{th:gjn check}, \ref{vi:all gjn} and \eqref{thm:virthm:eq1}--\eqref{thm:virthm:eq3} are satisfied. Furthermore, by  \autoref{th:gjn check} we know  that $(\Apse, \Lps, \Dps)$ is a GJN triple. 
 Hence \autoref{thm:invariance_gjn_domain} shows that $\exp( \i t \Apse )$, $t \in \IR$, leaves $\Def(\Lambda_\p)$ invariant. This and the invariance  for  the unitary group generated by  $A_\f$, as already shown in Step 1,  imply   \ref{vi:eitA invariant}.  Thus we can apply the  virial theorem and obtain   
 a sequence $(\psi_n)_{n\in \IN}$ in $\Def(C_1^{(\epsilon)}) \cap \Def(L_\lambda)$ such that $\psi_n \to \psi$, $n \to \infty$, and $\lim_{n\to\infty} \sc{\psi_n, C_1^{(\epsilon)} \psi_n} = 0$. Now, since  $\qcee$ is closed and  thus lower semi-continuous, we obtain $\psi \in \Def(\qcee)$ and 
\[
\qcee(\psi) \leq \lim_{n \to \infty} \qcee(\psi_n) = 0. 
\]

\vspace{0.4cm} 

\noindent 
\underline{Step 3:}   The theorem  follows from  Step 1, Step 2, and  \autoref{th:W1e lemma}.
\end{proof}

\section{Estimates on the Scattering Functions}
\label{sec:estimates scattering}
The aim of this section is to prove  that the commutators of the interaction with the dilation operator in scattering space are sufficiently bounded. To achieve this, we use the Born series expansion of the scattering functions, that is, we expand them using the recursion formula of the Lippmann-Schwinger equation \eqref{eq:lippmann-schwinger}. Then we get the Born series terms,  and a remainder term since we perform only finitely many recursion steps. The idea is that the remainder term decays fast enough for the momentum $\abs k \to \infty$ after  sufficiently many iteration  steps.

\subsection{Born Series Expansion and Technical Preparations}

First we show that that the scattering functions as well as their derivatives with respect to the wave vector $k$  are bounded. For that we use the method of modified square integrable scattering functions as described in \Autoref{sec:scattering}. Remember that $\phi(k,\cdot)$, $k \in \IR^3$, denote the continuous scattering functions on $\IR^3$ and $V$ a potential satisfying the assumptions \ref{ass:h1}, \ref{ass:h2}.  As $V$ is compactly supported we may assume that $\supp V$ is contained in a ball around the origin of radius $R$.

 We now set $\phit(k,x) := \abs{V(x)}^{1/2} \phi(k,x)$. Then $\phit(k,\cdot) \in L^2(\IR^3)$ for all $k$, and
the function satisfies the modified  Lippmann-Schwinger equation
\[
\phit(k,x) = \abs{V(x)}^{1/2} e^{\i k x} + (L_{\abs k} \phit(k,\cdot))(x),
\]
where $L_\kappa$ is defined in  \eqref{eq:deofL}. 
We recall  from \Autoref{sec:scattering} that we can recover the original scattering function from the modified one by
\begin{align}
\label{eq:recovery phi2}
\phi(k,x) = e^{\i k x} - \frac{1}{4\pi} \int \frac{e^{\i \abs k \abs{x-y}}}{\abs{x-y}} V(y)^{1/2} \phit(k,y) \d y .
\end{align}
Now we extend the results of boundedness of the first derivative of the scattering functions in \cite[Lemma 1.1.3]{newton} to derivatives of arbitrary order. 
\begin{prop}
\label{th:phi derivatives bounded}
Assume  \ref{ass:h1} and  \ref{ass:h2}.  
Let $\hDk = \frac{k}{|k|} \nabla_k$ and let $\partial_{k_j}$ be the derivative with respect to the $j$-th component of $k$.
For all $n \in \IN_0$, $m \in \{0,1\}$ there is a polynomial $P$
such that for all $x$ and $k \neq 0$,
\[
\abs{ \partial^m_{k_j} \hDk^n \varphi(k,x) } \leq P(\abs x) .
\]
\end{prop}
\begin{proof}
We get by the modified Lippmann-Schwinger equation
\[
\phit(k,\cdot) = (\Id - L_{\abs k})^{-1} ( \abs{V}^{1/2} e_k ),
\]
where $e_k(x) := e^{\i k x}$. First we claim  that $(\Id - L_{\abs k})^{-1}$ is uniformly bounded in $k \in \IR^3$.
For this we note that $\kappa \mapsto L_{\kappa}$ is continuous on $[0,\infty)$, which is easy to see, 	cf.  \cite[Theorem 1.22]{simon}.
Moreover,  $\lim_{\kappa \to 0} \| L_\kappa \| \to 0$  follows from a result of Zemach and Klein, see  \cite[Theorem 1.23]{simon}. Since $L_\kappa$ is Hilbert-Schmidt
and hence compact,  it follows from the Fredholm alternative that  $\Id - L_\kappa$ is invertible, 
provided $\psi =  L_{\kappa} \psi$ has no non-trivial solutions  in $L^2$. 
But as in the proof of Theorem \ref{th:scattering functions properties}, such non-trivial 
solutions are ruled out for all $\kappa \geq 0$. 
Since the inverse is a continuous map on the space of bounded  invertible operators, the claim about the bounded resolvent now follows. 

Observe that $\hDk |k| = 1$ and 
$\hDk (k/|k|) = 0$.  Thus, for any $n \in \IN_0$, the operators  $\hDk^n (\Id - L_{\abs k})^{-1} $ are again uniformly bounded   in $k \not= 0$, since differentiation with $\hDk$ yields just higher powers of $(\Id - L_{\abs k})^{-1}$ and radial derivatives of $L_{\abs k}$, which are again bounded operators since $V$ decays fast enough. 
Similarly one sees that  $\partial_{k_j} \hDk^n (\Id - L_{\abs k})^{-1} $ is 
uniformly bounded   in $k \not= 0$. Note that the expression is not differentiable at 
the origin.  Furthermore, for any $n \in \IN$,
\[
\sup_{k \not= 0}  \nn{ \hDk^n  (\abs{V}^{1/2}  e_k )}_2 < \infty
\]
as $V$ is compactly supported. Thus, we have shown, for all $n \in \IN_0$,
\[
\sup_{k \not= 0} \nn{ \hDk^n  \phit(k,\cdot) }_2 < \infty.
\]
Now we can differentiate \eqref{eq:recovery phi2}, estimate the integral with Cauchy-Schwarz, and use that
\begin{align}
\label{eq:square V estimate}
\int \frac{\abs{V(y)}}{\abs{x-y}^2} \d y \leq \| V \|_\infty \int_{B_R(x)} \frac{\d y}{y^2} 
\leq 
\| V \|_\infty \left( \int_{B_{3R}(0)} \frac{\d y}{y^2} +   \frac{\abs{ B_R(0)}}{R^2}  \right)
 < \infty
\end{align}
is bounded uniformly in $x$, where the second inequality  can be seen by considering the cases $|x| < 2R$ and $|x| \geq 2R$.  
\end{proof}
Next, we  perform the Born series expansion.  Similar to \cite{ikebe} it is convenient to introduce a symbol for the integral operator in the Lippmann-Schwinger equation. We consider a slightly bigger class of operators to cover also derivatives with respect to $k$. 
Let $\Cb(\IR^3)$ denote the bounded continuous functions on $\IR^3$ and $\Pb(\IR^3)$ the polynomially bounded continuous functions, that is,
\[
\Pb(\IR^3) = \{ \psi \in C(\IR^3) \:  \, \exists n \in \IN_0 , \,  \exists C > 0 , \, \forall x  \in \IR^3 \: \abs{\psi(x)} \leq C (1+\abs x)^n \}.
\]
\begin{defn}
\label{def:Tk}
For $\kappa \geq 0$, $\psi \in \Pb(\IR^3)$ and $n \in \IN_0$, we define 
\[
T^{(n)}_{V,\kappa} \psi(x) := \int \frac{e^{\i \kappa \abs{x-y}}  }{\abs{x-y}^{1-n}} V(y) \psi(y) \d y = \int \frac{e^{\i \kappa \abs v}}{\abs v^{1-n}} V(v+x) \psi(v+x) \d v 
\]
(where the second equality follows by a simple change of variables). 
Furthermore, we set $T_{V,\kappa} := T^{(0)}_{V,\kappa}$.
\end{defn}
In the next proposition we collect a few  elementary properties of these operators. 
\begin{prop}
\label{th:Tk properties}
For all $\kappa \geq 0$,
\begin{enumerate}[label=(\alph*)]
\item $T_{V,\kappa}$, $T^{(-1)}_{V,\kappa}$  are bounded operators from $\Cb(\IR^3)$ to $\Cb(\IR^3)$ with an operator norm, which can be  uniformly bounded in $\kappa$, 
\item for all $n \in \IN_0$, $T^{(n)}_{V,\kappa}$ maps $\Pb(\IR^3)$ to $\Pb(\IR^3)$. 
\end{enumerate}
\end{prop}
\begin{proof}
\begin{enumerate}[label=(\alph*)]
\item
It follows that for $\psi \in \Cb(\IR^3)$, $x\in \IR^3$, $\kappa \geq 0$,
\begin{align*}
\abs{ T^{-1}_{V,\kappa}\psi(x) } &\leq \nn{\psi}_\infty \int \frac{ \abs{ V(y) }  }{\abs{x-y}^2} \d y   \leq  \nn{V}_\infty  \nn{\psi}_\infty  \int_{B_R(0)} \frac{1}{\abs{x-y}^2} \d y,  \\
\abs{ T_{V,\kappa}\psi(x) } &\leq \nn{\psi}_\infty \int \frac{ \abs{ V(y) }  }{\abs{x-y}} \d y  \leq  \nn{V}_2  \nn{\psi}_\infty  \left( \int_{B_R(0)} \frac{1}{\abs{x-y}^2} \d y \right)^{1/2} .
\end{align*}
The integral is bounded independent of $x$, see \eqref{eq:square V estimate}.
\item
There exists a constant $C > 0$ such that for all $x\in \IR^3$, $n \in \IN_0$, $\kappa \geq 0$,
\begin{align}
\abs{ T^{(n)}_{V,\kappa} \psi(x) } &\leq   C \int_{B_R(0)} \frac{ \abs{ V(y) }  }{\abs{x-y}^{1-n}}  (1+ \abs y^m)  \d y    \label{eq:boundonTs} \\
&\leq C (1+ \abs R^m)  \nn{V}_2 \left( \int_{B_R(x)} \frac{1}{\abs{y}^{2-2n}} \d y \right)^{1/2}.  \nonumber 
\end{align}
The last integral can be estimated by
\[
\int_{r=0}^{R+\abs x} r^{2n}  \d r,
\]
which is bounded by a polynomial in $\abs x$. \qedhere
\end{enumerate}
\end{proof}

Using the  notation  introduced above and iterating the Lippmann-Schwinger equation \eqref{eq:lippmann-schwinger} we arrive at the following proposition. 
\begin{prop}\label{prop:expansionofscatter} 
For all $N \in \IN_0$, $k,x \in \IR^3$, we have
\[
\phi(k,x) = \sum_{n=0}^N \phiz^{(n)}(k,x) + \phir^{(N+1)}(k,x),
\]
where we defined  for $n \in \IN_0$
\begin{align*}
\phiz^{(n)}(k,x) &:= (-4\pi)^{-n}  T_{V,\abs k}^n e_k(x), \\
\phir^{(n)}(k,x) &:= (-4\pi)^{-n}  T_{V,\abs k}^n  \phi(k,\cdot) (x),
\end{align*}
where $e_k(x) := e^{\i k x}$.
\end{prop}

As an immediate consequence of  an iterated application of  the first and 
second identity in \autoref{def:Tk} we find 
the following lemma, which we will use. 

\begin{lemma} \label{iteratedformula2} 
Let  $V_1,\ldots, V_p \in \Cci(\IR^3)$, $n_1,\ldots,n_p \in \IN_0$,  and $\psi \in \Pb(\IR^3)$.
 Then for all $k, x \in \R^3$, we have
\begin{align} \nonumber 
(T_{V_1,\abs k}^{(n_1)}  &\cdots   T_{V_p,\abs k}^{(n_p)} \psi )(x_0)   \\ &=   \int 
 \left\lbrace  \prod_{l=1}^p \frac{e^{\i  \abs k  \abs{x_{l-1}-x_l} }}{\abs {x_{l-1}-x_l}^{1-n_l}}  V_l(x_l) \right\rbrace 
\psi( x_p ) 
 \d(x_1, \ldots, x_p) , \label{eq:tk formula0}  
\\ &=   \int  \left\lbrace  \prod_{l=1}^p \frac{e^{\i  \abs k  \abs{u_l} }}{\abs {u_l}^{1-n_l}}  V_l\left(x_0  + \sum_{s=1}^l u_s \right) \right\rbrace 
\psi\left( x_0 + \sum_{s=1}^p u_s \right) 
 \d(u_1, \ldots, u_p)     \label{eq:tk formula-1} , 
\end{align} 
and the special case 
\begin{align} 
(T_{V_1,\abs k}^{(n_1)}  &\cdots   T_{V_p,\abs k}^{(n_p)} e_k )(x) \label{eq:tk formula000}
 =   e^{\i k x} \int  \left\lbrace  \prod_{l=1}^p \frac{e^{\i  ( \abs k  \abs{u_l} + k u_l)}}{\abs {u_l}^{1-n_l}}  V_l\left(x  + \sum_{s=1}^l u_s \right) \right\rbrace 
 \d(u_1, \ldots, u_p).
\end{align}
\end{lemma}

\subsection{Estimates of the  Terms of the Born Series}

In this subsection we prove decay estimates for the inner products of an abstract coupling function $\chi \in \S(\IR^3)$ with (derivatives with respect to $k$ of) the functions $\phiz^{(n)}(k,\cdot)$, $k \in \IR^3$, $n \in \IN_0$. These estimates   will be collected in \thref{th:final:phiz} and \thref{th:final:phiz2}. In fact one can  show an arbitrary fast decay for \textit{any} $n \in \IN$. The main tool will be a standard stationary phase argument as given in the next lemma.  For this recall 
the notation  $\sc{k} = (1 + |k|^2)^{1/2}$.  
\begin{lemma}[Stationary phase]
\label{th:stationary phase}
For any $n \in \IN$ there exists a constant $C$ such that for all $g \in \Cci(\IR^3)$ and $k \in \IR^3$, we have
\[
\abs { \int e^{\i  x k} g(x) \d x  } \leq \frac{C}{\sc{k}^n}  \sup_{\abs \alpha \leq n} \nn{D^\alpha g}_1.
\]
\end{lemma}
\begin{proof}
For all $k \in \IR^3$, $j=1, \ldots, n$,
\begin{align*}
\i k_j\int e^{\i xk} g(x) \d x &= \int \partial_j  e^{\i xk} g(x) \d x  \\
&= \lim_{R \to \infty} \int_{S_R(0)} e^{\i xk} g(x) \d x -  \int e^{\i xk} \partial_j   g(x) \d x.  
\end{align*}
The first term clearly vanishes. Now we can repeat this procedure $n$ times. 
\end{proof}
We proceed by computing the derivatives as well as  the effect of  multiple applications of the dilation operator in the variable $k$ acting on the terms of the Born series. The idea is that the application of $k \nabla_k$ or $\nabla_k$ on terms of the form
\begin{align}
\label{eq:T V1 Vp}
T_{V_1,\abs k} \cdots   T_{V_p,\abs k} e_k, \qquad V_1, \ldots, V_p \in \Cci(\IR^3), ~ p \in \IN,
\end{align}
yields again a linear combination of such terms multiplied with polynomials in $x$ and $k$ (see \autoref{th:k nabla lemma} and \autoref{th:k nabla lemma 2}). We  want to  remember that the Born series terms can be written in the form \eqref{eq:T V1 Vp}.
This procedure can be repeated multiple times and the resulting expressions  can then be estimated with the stationary phase argument. 
\begin{lemma}
\label{th:k nabla lemma}
Assume $V_1,\ldots, V_p \in \Cci(\IR^3)$. Then we can write for all $k \in \IR^3$,
\[
k \nabla_k \left( T_{V_1,\abs k} \cdots   T_{V_p,\abs k} e_k \right)
\]
as a sum of
\begin{align}
\label{eq:tk first term}
\i (k \hat {\mathsf x} ) T_{V_1,\abs k}  \cdots   T_{V_p,\abs k} e_k,
\end{align}
where $\hat {\mathsf x}$ denotes the multiplication in $x$, 
and  terms of the form
\begin{align}
\label{eq:tk second term}
\sum_{l=1}^p Q T_{W_1,\abs k} \cdots   T_{W_l,\abs k} e_k,
\end{align}
where $Q$ denotes the multiplication in $x$ with a polynomial of maximal degree  one, and $W_l \in \Cci(\IR^3)$, $l=1, \ldots, p$.
\end{lemma}
\begin{proof}
Recall the formula \eqref{eq:tk formula000} in \autoref{iteratedformula2}, 
\[
(T_{V_1,\abs k}  \cdots   T_{V_p,\abs k} e_k )(x) =  e^{\i k x} \int  \left\lbrace  \prod_{l=1}^p \frac{e^{\i (\abs k  \abs{u_l} + k u_l)}}{\abs {u_l}}  V_l\left(x  + \sum_{s=1}^l u_s \right) \right\rbrace  \d(u_1, \ldots, u_p). \]
Differentiation with respect to the first factor on the right-hand side
 yields \eqref{eq:tk first term}. Under the integral we use 
\begin{align}
k \nabla_k   \prod_{l=1}^p \frac{e^{\i ( k u_l +  \abs k  \abs{u_l})}}{\abs {u_l}} &= \i \sum_{l'=1}^p ( k u_{l'} + \abs k \abs{u_{l'}} )  \prod_{l=1}^p \frac{e^{\i ( k u_l +  \abs k  \abs{u_l})}}{\abs {u_l}} \nonumber \\
&=  \sum_{l'=1}^p (u_{l'}  \nabla_{u_{l'}} +1)  \prod_{l=1}^p \frac{e^{\i ( k u_l +  \abs k  \abs{u_l})}}{\abs {u_l}}. \label{eq:k nabla partial}
\end{align}
We can now use  in \eqref{eq:tk formula000}  integration by parts   to shift the derivatives to the $V_l$ terms. Any boundary terms vanish as we consider compactly supported functions. 
Thus, we arrive at
\begin{align*}
k \nabla_k &\int \left\lbrace \prod_{l=1}^p \frac{e^{\i (\abs k  \abs{u_l} + k u_l)}}{\abs {u_l}}   V_l\left(x  + \sum_{s=1}^l u_s \right) \right\rbrace  \d(u_1, \ldots, u_p) \\
&= - \int \left\lbrace \prod_{l=1}^p \frac{e^{\i (\abs k  \abs{u_l} + k u_l)}}{\abs {u_l}} \right\rbrace   \sum_{l'=1}^p (u_{l'}  \nabla_{u_{l'}} +2)   \prod_{l=1}^p V_l\left(x  + \sum_{s=1}^l u_s \right)  \d(u_1, \ldots, u_p)\\
&= - \int \left\lbrace   \prod_{l=1}^p \frac{e^{\i (\abs k  \abs{u_l} + k u_l)}}{\abs {u_l}}  \right\rbrace \sum_{l'=1}^p \left\lbrace \prod_{\substack{l=1 \\ l\not= l'}}^p   V_l\left(x  + \sum_{s=1}^l u_s \right)  \right\rbrace        \\ &\qquad  \left(  \sum_{s=1}^{l'} u_{s}  \nabla + 2 \right) V_{l'}\left(x  + \sum_{s=1}^{l'} u_s \right) \d(u_1, \ldots, u_p) , 
\end{align*}
where the last equality follows by  calculating the derivatives by means of the product 
and chain rule and by reordering the summation. 
We can now write, with $W_{l'}(y) := y \nabla V_{l'}(y)$,
\begin{align*}
 \left( \sum_{s=1}^{l'} u_{s}  \nabla \right) V_{l'}\left(x  + \sum_{s=1}^{l'} u_s \right)   = W_{l'}\left(x  + \sum_{s=1}^l u_s \right) - x \nabla V_{l'}\left(x  + \sum_{s=1}^l u_s \right).
\end{align*}
Since  $W_{l'}$ and the derivatives of $V_{l'}$ are again in $\Cci(\IR^3)$ for all $l'$, we obtain expressions of the form \eqref{eq:tk second term}.
\end{proof}
\begin{lemma}
\label{th:k nabla lemma 2}
Let $p \in \IN$. 
Assume that $V_1,\ldots, V_p \in \Cci(\IR^3)$, $n_1, \ldots, n_p \in \IN_0$. Then for $j \in \{1,2,3\}$, $k \not= 0$,
\begin{align*}
\partial_{k_j} \left( T^{(n_1)}_{V_1,\abs k}   \cdots    T^{(n_p)}_{V_p,\abs k} e_k \right) &= \i  T^{(n_1)}_{V_1,\abs k}   \cdots    T^{(n_p)}_{ \hat{\mathsf x}_j  V_p,\abs k}e_k + \i \frac{k_j}{\abs k } \sum_{i=1}^p X_1^{(i)}   \cdots   X_p^{(i)} e_k, \\
\hDk \left( T^{(n_1)}_{V_1,\abs k}   \cdots    T^{(n_p)}_{V_p,\abs k} e_k \right) &= \i T^{(n_1)}_{V_1,\abs k}   \cdots    T^{(n_p)}_{ \frac{k \hat{\mathsf x}  }{\abs k}  V_p,\abs k} 
 e_k + \i \sum_{i=1}^p X_1^{(i)}   \cdots   X_p^{(i)} e_k, 
\end{align*}
where $\hat{\mathsf x}_j$ stands for multiplication by $x_j$ and   
\[
X_l^{(i)} := \begin{cases} T^{(n_l)}_{V_l,\abs k}, & i \not= l, \\ T^{(n_l+1)}_{V_l,\abs k}, & i=l. \end{cases}
\]
\end{lemma}
\begin{proof}
This follows by direct computation of the derivative  by means of the product rule 
of the expression \eqref{eq:tk formula000}.
\end{proof}
Finally, we use  the previous estimates for  the following two propositions. 
\begin{prop}
\label{th:final:phiz}
 For all $s \in \IN_0$, $p,m,n \in \IN$, $X \in \{ \Id, \nabla_k, \nabla_{k'} \}$, $Y \in \{ k \nabla_k + k' \nabla_k', \eta(k) k  \nabla_k \}$, where  $\eta \in \S(\IR^3)$, 
 there are constants $n_1, n_2 \in \IN$, $C$,  such that for all $k,k' \not= 0$, $\chi \in \S(\IR^3)$,
\begin{align*}
\abs{  X  Y^s \sc{ \phi_0^{(p)}(k,\cdot), \chi \phi_0^{(m)}(k',\cdot) } } \leq \frac{C}{1+ \abs{k-k'}^n} \sup_{\abs \alpha \leq n_1} \nn{\sc{\cdot}^{n_2} \partial_x^\alpha \chi}_1.
\end{align*}
\end{prop}
\begin{proof}
First, let $Y = k \nabla_k + k' \nabla_k'$. Using an induction argument in $s$ 
we obtain from  \thref{th:k nabla lemma} that  we can write 
\[
(k \nabla_k + k' \nabla_k')^s \sc{ \phi_0^{(p)}(k,\cdot), \chi \phi_0^{(m)}(k',\cdot) }
\]
as linear combination of terms of the form
\begin{align}
\label{eq:k k prime expressions}
 (k-k')^{\alpha} \sc{ T_{{V}_1,\abs k}   \cdots    T_{{V}_p,\abs k} e_k , P \chi T_{{W}_1,\abs {k'}}   \cdots    T_{{W}_m,\abs {k'}} e_{k'} },
\end{align}
for some polynomial $P$, multi-index $\alpha$, ${V}_1, \ldots, {V}_p$, ${W}_1, \ldots, {W}_m \in \Cci(\IR^3)$.
Then we obtain the desired estimate for $X =  \Id$ by the stationary phase argument of \thref{th:stationary phase}, which  can  be seen using  
\eqref{eq:tk formula000}
and observing that $\chi$ is a Schwartz function and that the potential $V$ has compact support.  For $X \in \{  \nabla_k ,  \nabla_{k'}\}$  we apply  \thref{th:k nabla lemma 2} 
to the expressions in 
\eqref{eq:k k prime expressions}, with the result that 
we can write 
\[
X (k \nabla_k + k' \nabla_k')^s \sc{ \phi_0^{(p)}(k,\cdot), \chi \phi_0^{(m)}(k',\cdot) }
\]
 for $k,k' \not= 0$ as linear combinations of terms
\[
 (k-k')^{\alpha} f(k,k') \sc{ T_{{V}_1,\abs k}^{(n_1)}   \cdots    T_{{V}_p,\abs k}^{(n_p)}  e_k , P \chi T_{{W}_1,\abs {k'}}^{(n'_1)}   \cdots    T_{{W}_m,\abs {k'}}^{(n'_m)} e_{k'} },
\]
with now a bounded function $f$ on $\IR^3 \times \IR^3$ and $n_1 + \ldots + n_p$, $n'_1 + \ldots + n'_m \in \{0,1\}$. The desired estimate in this case  follows now from 
the same stationary phase argument  using \eqref{eq:tk formula000}   as before.  Finally, 
for $Y = \eta(k) k \nabla_k$ one proceeds similarly  but now using   only \autoref{th:k nabla lemma 2}. 
\end{proof}

\begin{prop}
\label{th:final:phiz2}
 For all $s \in \IN_0$, $p,n \in \IN$, there exist constants $n_1, n_2 \in \IN$, $C$, such that for all $k$, $X \in  \{ \Id, \nabla_k \}$,  and   $\chi \in \S(\IR^3)$
 we have 
\begin{align*}
\abs{ X  (k \nabla_k)^s \sc{ \phi_0^{(p)}(k,\cdot), \chi } } \leq \frac{C}{1+ \abs{k}^n} \sup_{\abs \alpha \leq n_1} \nn{\sc{\cdot}^{n_2} \partial_x^\alpha \chi}_1.
\end{align*}
\end{prop}
\begin{proof}
Analogously to the proof of  \autoref{th:final:phiz} we  inductively apply    \autoref{th:k nabla lemma} to find  that
\[
(k \nabla_k)^s \sc{ \phi_0^{(p)}(k,\cdot), \chi }
\]
can be written as
\begin{align}
\label{eq:k k prime expressions 2}
k^{\alpha} \sc{ T_{{V}_1,\abs k}   \cdots    T_{{V}_p,\abs k} e_k, P \chi},
\end{align}
for some polynomial $P$, multi-index $\alpha$, ${V}_1, \ldots, {V}_p \in \Cci(\IR^3)$. Then after inserting  \eqref{eq:tk formula000} and
 using \thref{th:k nabla lemma 2} with  $X= \nabla_k$,  the stationary phase argument yields again the desired estimate. 
\end{proof}

%
%
%
\subsection{Estimates of the Remainder Terms}

Now we prove arbitrarily  fast polynomial decay for the remainder terms of sufficiently high order. We obtain results for remainder terms in \thref{th:final:phir} and scalar products of Born series terms with remainder terms in \thref{th:phir phi0}. The main tool will be the following lemma, where the basic idea is due to \textsc{Klein} and \textsc{Zemach} (cf. \cite{ZemachKlein}). It 
essentially  follows   from  a stationary phase argument together with a suitable coordinate transformation. 
\begin{lemma}
\label{th:eve lemma}
For $V \in \Cci(\IR^3)$, $n_1, n_2 \in \IN_0$, $R > 0$ such that $\supp V \subseteq B_R(0)$, there exists a constant $C$ such that for all $\kappa \geq 0$,
\begin{align}
\label{eq:eve lemma}
\sup_{\abs x, \abs{x'} \leq R} \abs{ \int \frac{e^{ \i \kappa \abs{x-y}}}{\abs{x-y}^{1-n_1}}  V(y)  \frac{ e^{  \i \kappa \abs{x'-y}}}{\abs{x'-y}^{1-n_2}} \d y } \leq  \frac{C}{1 + \kappa}.
\end{align}
\end{lemma}
\begin{proof}
For the proof we use Prolate Spheroidal coordinates, see \cite[appendix]{ZemachKlein} and \cite[p. 661]{MorseFeshbach}. Let $D = \frac{1}{2} \abs{x-x'}$. For
\[
\xi \in [D,\infty), \qquad \eta \in [-1,1], \qquad \phi \in [0,2\pi),
\]
we set
\[
\Phi(\xi,\eta, \phi) = \frac{1}{2}(x + x') +  \mathcal R \mm{ \sqrt{(\xi^2 - D^2)(1-\eta^2)} \cos \phi \\ \sqrt{(\xi^2 - D^2)(1-\eta^2)} \sin \phi \\ \xi \eta },
\]
where $\mathcal R$ is the rotation matrix transforming $e_3$ into $\frac{x - x'}{\abs{x - x'}}$. A straightforward computation then shows that 
\begin{align*}
& \xi = \frac{1}{2}\left( \abs{x - \Phi(\xi,\eta, \phi)}  + \abs{x' - \Phi(\xi,\eta, \phi)}\right) ,\\
& \eta = \frac{1}{2D} \left( \abs{x - \Phi(\xi,\eta, \phi)}  - \abs{x' - \Phi(\xi,\eta, \phi)}\right), \\
& \det \Phi(\xi, \eta, \phi) = (\xi + D \eta) (\xi - D \eta).
\end{align*}
Thus, by change of coordinates,
\begin{align*}
\int \frac{e^{ \i \kappa \abs{x-y}}  V(y)  e^{  \i \kappa \abs{x'-y}} }{\abs{x-y}^{1-n_1} \abs{x'-y}^{1-n_2} } \d y  &= \int e^{2\i \kappa \xi} V( \Phi(\xi, \eta, \phi)) (\xi + D \eta)^{n_1} (\xi - D \eta)^{n_2} \d (\xi,\eta,\phi).
\\ &= \int_D^{\infty} e^{2\i \kappa \xi} h(\xi)  \d \xi,
\end{align*}
where $h(\xi) := \int V( \Phi(\xi, \eta, \phi))  (\xi + D \eta)^{n_1} (\xi - D \eta)^{n_2} \d (\eta, \phi)$. Let $E := \frac{1}{2} \abs{x+x'}$. Notice that by direct computation, for $\xi \geq D + E$, 
\[
\abs{ \Phi(\xi, \eta, \phi) } \geq \xi - D - E.
\]
Thus, we get that $h(\xi) = 0$ for $\xi \geq R + D + E$. Then, by integration by parts, 
\begin{align*}
\int_D^{\infty} e^{2\i \kappa \xi} h(\xi) \d \xi &= \frac{1}{2 \i \kappa}  \int_D^{R + D + E} \partial_\xi  \left( e^{2\i \kappa \xi}\right) h(\xi) \d \xi  \\
&= \frac{1}{2 \i \kappa}  \bigg(  - h(D)  e^{2\i \kappa D} - \int_D^{R + D + E}  e^{2\i \kappa \xi} \partial_\xi h(\xi) \d \xi  \bigg).
\end{align*}
As $D,E \leq R$ are bounded, so is the first term. For the second one notice that
\begin{align}
\partial_\xi h(\xi) &= \int \sc{ \nabla V( \Phi(\xi, \eta, \phi)), \partial_\xi \Phi(\xi,\eta,\phi)}  (\xi + D \eta)^{n_1} (\xi - D \eta)^{n_2} \d (\eta, \phi) \label{eq:xih first} \\
\qquad &+ \int V( \Phi(\xi, \eta, \phi)) \partial_\xi ( (\xi + D \eta)^{n_1} (\xi - D \eta)^{n_2} ) \d (\eta, \phi). \label{eq:xih second}
\end{align}
The term \eqref{eq:xih second} is clearly bounded by a constant depending only on $R$. The term \eqref{eq:xih first} is bounded up to a constant by
\begin{align*}
\sup_{\eta,\phi} \abs{ \partial_\xi \Phi(\xi,\eta, \phi) } \leq C\left( 1 + \frac{\xi}{\sqrt{\xi^2 - D^2}} \right),
\end{align*}
for some constant $C$. This is integrable and the integral is also bounded by a constant only depending on $R$:
\[
\int_D^{R+D+E} \frac{\xi}{\sqrt{\xi^2 - D^2}} \d \xi = \sqrt{(R+D+E)^2 - D^2}. 
\qedhere
\]
\end{proof}

\begin{lemma}
\label{th:pm star}
Let $p \in \IN$, $V_1, \ldots, V_{p} \in \Cci(\IR^3)$ and  $n_1, \ldots, n_p \in \IN_0$. Then there exists a constant $C$ such that for all $k,x \in \IR^3$, continuous bounded functions $\psi$  on $\IR^3$,
\begin{align*}
\abs{ (T^{(n_1)}_{V_1,\abs k}   \cdots    T^{(n_p)}_{V_p,\abs k} \psi )(x)  } \leq  \frac{C  (1+ \sc x^{n_1-1}) \nn{\psi}_\infty}{1+ \abs k^{\lfloor \frac{p-1}{2} \rfloor }}.
\end{align*}
\end{lemma}
\begin{proof}
First we assume that $p = 2p^*+1$. Then by Lemma \ref{iteratedformula2} 
\begin{align}
(T^{(n_1)}_{V_1,\abs k} &  \cdots    T^{(n_p)}_{V_p,\abs k} \psi )(x) \nonumber \\
&= \int  \frac{e^{\i \abs k \abs{x- y_1}}}{\abs{x - y_{1}}^{1-n_{1}}} V_{1}(y_1) \label{eq:first integral} \\
&\qquad \left\lbrace \prod_{l=1}^{p^*} \frac{e^{\i \abs k \abs{y_{2l-1} - y_{2l}}}}{\abs{y_{2l-1} - y_{2l}}^{1-n_{2l}}} V_{2l}(y_{2l})  \frac{e^{\i \abs k \abs{y_{2l} - y_{2l+1}}}}{\abs{y_{2l} - y_{2l+1}}^{1-n_{2l+1}}}  V_{2l+1}(y_{2l+1}) \right\rbrace \label{eq:next integrals} \\ &\qquad \psi( y_p) \d(y_1, y_2, y_3, \ldots, y_p). \nonumber
\end{align}
In the following let $C$ denote different constants depending only on $V_l$ and $n_l$, $l=1,\ldots,p$.
We estimate the terms in \eqref{eq:next integrals} for $l=1,\ldots, p^*$ by
\[
\abs{ \int \frac{e^{\i \abs k \abs{y_{2l-1} - y_{2l}}}}{\abs{y_{2l-1} - y_{2l}}^{1-n_{2l}}} V_{2l}(y_{2l})  \frac{e^{\i \abs k \abs{y_{2l} - y_{2l+1}}}}{\abs{y_{2l} - y_{2l+1}}^{1-n_{2l+1}}}  \d y_{2l} } \leq \frac{C}{1 + \abs k}
\]
using \thref{th:eve lemma}, the term \eqref{eq:first integral} by
\[
 \abs{ \int  \frac{e^{\i \abs k \abs{x- y_1}}}{\abs{x - y_{1}}^{1-n_{1}}} V_{1}(y_1) \d y_1 } \leq C (1+ \sc x^{n_1-1}) ,
\]
and thus we find 
\begin{align*}
&\abs{ (T^{(n_1)}_{V_1,\abs k}   \cdots    T^{(n_p)}_{V_p,\abs k} \psi)(x) } \
\\ &\qquad \leq \frac{C (1+ \sc x^{n_1-1})  }{(1+\abs k)^{p^*}} \int \left\lbrace \prod_{l=0}^{p^*-1} \abs{ V_{2l+2}(y_{2l+2}) } \right\rbrace \psi( y_p) \d(y_2, y_4, \ldots, y_p) \\
  &\qquad  \leq  \frac{ C (1+ \sc x^{n_1-1})  \nn{\psi}_\infty }{(1+\abs k)^{p^*}}.
\end{align*}
In case $p$ is even, we  estimate the first $p-1$ factors as in the odd case   
and the remaining expression we estimate using  inequality  \eqref{eq:boundonTs},
which implies  that 
there is a constant $C$ independent of $k$ such that 
$\| \ind_{B_R(0)}     T^{(n_p)}_{V_p,\abs k} \psi \|_\infty \leq C \| \psi \|_\infty$. 
\end{proof}

\begin{prop}
\label{th:final:phir}
Let $n \in \IN_0$, $m \in \{0,1\}$, $j \in \{1,2,3\}$. 
\begin{enumerate}[label = (\alph*)]
\item \label{eq:linear combination}  For all $k \not= 0$, the expression 
\begin{align*}
\partial^{m}_{k_j} \hDk^n  T_{V,\abs k}   \cdots    T_{V,\abs k} \phi(k,\cdot)
\end{align*}
can be written as linear combination of terms 
\begin{align}
\label{eq:tk r derivative}
f(k) T^{(n_1)}_{V,\abs k}   \cdots    T^{(n_p)}_{V,\abs k} \partial^{m'}_{k_j} \hDk^{n'} \phi(k,\cdot),
\end{align}
where $f$ is a bounded function on $\IR^3 \setminus \{ 0 \}$, $0 \leq n' \leq n$, $0 \leq m' \leq m$, and $n_1+ \cdots + n_p + m' + n' = m+n$. 
\item \label{eq:estimate after linear combination}
For any $p \in \IN$, there exists a constant $C$ such that we have for all $k \not= 0$, $x \in \IR^3$,
\[
\abs{ \partial^{m}_{k_j} \hDk^n   \phir^{(p)}(k,x) } \leq \frac{C  (1+\sc{x}^{n+m-1}) }{1+ \abs k^{\lfloor \frac{p-1}{2} \rfloor }}.
\]
\end{enumerate}
\end{prop}
\begin{proof} Part 
\ref{eq:linear combination} follows by  the product rule from \autoref{iteratedformula2}.  Part  \ref{eq:estimate after linear combination} follows from  \ref{eq:linear combination},  \thref{th:pm star}, \thref{th:phi derivatives bounded}, and   the fact that $V$ has compact support. 
\end{proof}

%
\begin{lemma}
\label{th:phir phi0}
Let $p,m,r \in \IN$ with $m \geq 2r+4$, $V_1, \ldots, V_p$, $W_1, \ldots, W_m \in \Cci(\IR^3)$, and $n_1, \ldots, n_p$,  $n'_1, \ldots, n'_m \in \IN_0$. Then there exists a constant $C$, $n_0 \in \IN_0$, such that for all $k$, $k'$, continuous bounded functions $\psi$, and $\chi \in \S(\IR^3)$,
\begin{align}
\nonumber
&\abs{ \sc{T^{(n_1)}_{V_1,\abs k}   \cdots    T^{(n_p)}_{V_p,\abs k} e_k ,\chi   T^{(n'_1)}_{W_1,\abs {k'}}   \cdots    T^{(n'_m)}_{W_m,\abs {k'}}\psi    }  } \\ & \qquad  \leq \frac{C \sup_{|\alpha|\leq r}  \nn{ (1+ \sc{\cdot}^{n_0}) \partial^\alpha \chi  }_1   \nn{\psi}_\infty}{(1+ \abs k^{\lfloor \frac{p-1}{2} \rfloor + r} )(1+ \abs {k'}^{\frac{m}{2} -2 -r }) }. \label{eq:k k prime bound}
\end{align}
\end{lemma}

\begin{proof}
To shorten the notation we assume that $n_1 = \cdots = n_p = n'_1 = \cdots = n'_m  = 0$. The proof works   for the other cases   with the obvious change of notation (in fact the bounds in the proof given 
below only improve). 
The bound for $r=0$ follows from Lemma 
\ref{th:pm star}. Now let us consider  $r \geq 1$. For this we show the following identity. 
Fix $j=1,2,3$. 
 We claim that for each $n \in \IN_0$ there  exist coefficients $c(\cdots )$ such that for all
 ${V}_1, \ldots, {V}_{p}$, ${W}_1, \ldots, {W}_{m} \in \Cci(\IR^3)$,  $\psi \in \Cb(\IR^3)$,
 and $\chi \in \mathcal{S}(\IR^3)$ we have 
\begin{align}
k_j^n &\sc{T_{{V}_1,\abs k}   \cdots    T_{{V}_{{p}},\abs k} e_k ,\chi   T_{{W}_1,\abs {k'}}   \cdots    T_{{W}_{{m}},\abs {k'}} {\psi} } \label{explicitformuladecay} 
\\ &= \sum_{l=1}^{m+1} \sum_{\substack{ \underline{\mu} \in  \IN_0^m   }}  \sum_{\substack{ \underline{\nu} \in \{0,1\}^m 
 }}   
\sum_{\underline{s} \in \IN_0^p  , \underline{t} \in \IN_0^m , u \in \IN_0 } \nonumber  \\
& 1_{  \nu_i = \mu_i =  0 , \forall  i <  l } \ 1_{ \mu_l + \nu_l \geq  1 }  \ 1_{ \nu_i = 1 ,  \forall i  >  l }  \
 1_{|\underline{\mu}| + |\underline{\nu}| + |\underline{s}| + |\underline{t}| + u = n } 
 \label{coeffcond}  \\
&  c(n, l,\underline{\nu},\underline{\mu},\underline{s},\underline{t},u)  \nonumber  \\
& 
\sc{ \left\{ \prod_{a=1}^p  T_{{V}_a^{[s_a]},\abs k} \right\}   e_k ,\chi^{[u]}  
 \left\{ \prod_{b=1}^{m}
 S_{W_{b}^{[t_{b}]},\abs {k'}}^{(\mu_{b},\nu_b)} \right\}  \psi } \nonumber 
\end{align}
where we used the  multi-index notation $\underline{s} = (s_1,\ldots,s_p)$,  denoted by  $f^{[n]}$   the $n$-th partial derivative in the $j$-th coordinate direction, 
and defined 
$
S^{(\mu,\nu)}_{W,\kappa}$ as  the  operator with integral kernel
$$
S^{(\mu,\nu)}_{W,\kappa}(x,y) := ( \i \kappa)^\mu  e^{\i \kappa \abs{x -y}} \partial_{x_j}^\nu \left(  \frac{1}{\abs{x-y}} \left( \frac{x_j-y_j}{|x-y|} \right)^\mu \right)W(y)    , \quad x, y \in \R^3 .
$$
Oberve that $S^{(0,0)}_{W,\kappa} = T_{W,\kappa}$. 
We prove   \eqref{explicitformuladecay}   by induction in $n$. If $n=0$, we only need to  consider $l=m+1$
and $\underline{\mu}, \underline{\nu}, \underline{s}, \underline{t}, u$ equal zero. 
Suppose \eqref{explicitformuladecay} holds for $n$. We want to show, that it then also holds for $n+1$. 
For this we  fix an   $l$ and assume that  \eqref{coeffcond} is nonzero. 
This implies in particular that $(\mu_1,\nu_1), \ldots, (\mu_{l-1},\nu_{l-1})$ all equal  $(0,0)$
and $\mu_l + \nu_l \geq  1 $.  Using  \eqref{eq:tk formula-1} we find 
\begin{align}
& k_j
\sc{ \left\{ \prod_{a=1}^p  T_{{V}_a^{[s_a]},\abs k} \right\}   e_k ,\chi^{[u]}  
 \left\{ \prod_{b=1}^{m}
 S_{W_{b}^{[t_{b}]},\abs {k'}}^{(\mu_{b},\nu_b)} \right\}  \psi } \label{PIinteginduc1} 
\\
& \int \frac{e^{-\i \abs k \abs{v_1}}}{\abs{v_1}} V_1^{[s_1]}(v_1 + x) \frac{e^{-\i\abs k \abs{v_2}}}{\abs{v_2}} V_2^{[s_2]}(v_2 + v_1 + x) \ldots  \frac{e^{-\i \abs k \abs{v_p}}}{\abs{v_p}} V_p^{[s_p]}\left(\sum_{l=1}^p v_l + x \right) \nonumber  \\
&\qquad  e^{-\i k \sum_{l=1}^{p} v_l} ( \i \nabla_{x} e^{- \i k x} ) \chi^{[u]}(x)  \frac{e^{\i \abs {k'} \abs{x-x_1}}}{\abs{x-x_1}} W_1^{[t_1]}(x_1) \frac{e^{\i \abs {k'} \abs{x_1-x_2}}}{\abs{x_1-x_2}} W_2^{[t_2]}(x_2)  \nonumber 
\\ &\qquad \ldots  \frac{e^{\i \abs {k'} \abs{x_{l-2}-x_{l-1}}}}{\abs{x_{l-2}-x_{l-1}}} W_{l-1}(x_{l-1}) 
S_{W_{l}^{[t_{l}]},\abs {k'}}^{(\mu_{l},\nu_l)}(x_{l-1},x_l) 
\left(  \left\{ \prod_{b=l+1}^m  S_{W_{b}^{[t_{b}]},\abs {k'}}^{(\mu_{b},\nu_b)}
\right\}  \psi \right)(x_l)  \nonumber  \\
& \qquad  \d (v_1, \ldots, v_p , x, x_1, \ldots, x_l) .   \nonumber 
\end{align}
We now use integration by parts with respect to $x$.  
By the product rule, we have a linear combination of several different terms. The ones with derivatives of the potentials $V_1^{(s_1)}, \ldots, V_p^{(s_p)}$ and $\chi^{(u)} $ will simply increase exactly
one of the indices $s_1,\ldots,s_p,u$ by one. 
 For the remaining  term containing $x$ we use 
\[
\nabla_x  \frac{e^{\i \abs {k'} \abs{x-x_1}}}{\abs{x-x_1}} = \nabla_{x_1} \frac{e^{\i \abs {k'} \abs{x-x_1}}}{\abs{x-x_1}}.
\]
Then we use again  integration  by parts and obtain  a term 
involving $\nabla W_1$, which can be treated  as  before, and a term involving 
\[
\nabla_{x_1}  \frac{e^{\i \abs {k'} \abs{x_1-x_2}}}{\abs{x_1-x_2}}.
\]
We  repeat this procedure  until we arrive at the right-hand side at the term 
\[
\nabla_{x_{l-2}}  \frac{e^{\i \abs {k'} \abs{x_{l-2}-x_{l-1}}}}{\abs{x_{l-2}-x_{l-1}}}.
\]
If $\nu_l=1$, we use 
\[
\partial_{x_j}  \frac{e^{\i \abs {k'} \abs{x -y}}}{\abs{x  - y }} W(y) = 
 S_W^{(1,0)}(x,y) + S_W^{(0,1)}(x,y)  , 
\]
which follows by the product rule. 
If $\nu_l = 0$, then we do once more integration by parts and use the relation, which is 
straightforward to verify, 
\[
 \partial_{x_j}  S_W^{(\mu_l,0)}(x,y) = S_W^{(\mu_l+1,0)}(x,y) + S_W^{(\mu_l,1)}(x,y) . 
\]
We conclude that there exist coefficients $d(\cdots)$ such that 
\begin{align}
 & \text{l.h.s. of } \eqref{PIinteginduc1}  \nonumber   \\
& =  \sum_{\tilde{l}=l-1}^{l}   \sum_{\substack{ \tilde{\underline{\mu}} \in  \IN_0^m   }}  \sum_{\substack{ \tilde{\underline{\nu}} \in \{0,1\}^m 
 }} 
\sum_{\tilde{\underline{s}} \in \IN_0^p  , \tilde{\underline{t}} \in \IN_0^m , \tilde{u} \in \IN_0 } \nonumber  \\
& 1_{  \tilde{\nu}_i = \tilde{\nu}_i  = 0  , \forall  i  < \tilde{ l} } \ 1_{ \tilde{\mu}_{\tilde{l}} + \tilde{\nu}_{\tilde{l}} \geq  1 }  \  1_{ \tilde{\nu}_i = 1 ,  \forall i > \tilde{ l}   }  \
 1_{|\tilde{\underline{\mu}}| + |\tilde{\underline{\nu}}| + |\tilde{\underline{s}}| + |\tilde{\underline{t}}| +
\tilde{ u}  = n+1 }    \nonumber  \\
&  d(n, l,\underline{\nu},\underline{\mu},\underline{s},\underline{t},u;  \tilde{l},\tilde{\underline{\nu}},\tilde{\underline{\mu}},\tilde{\underline{s}},\tilde{\underline{t}},\tilde{u}  )  \nonumber  \\
& 
\sc{ \left\{ \prod_{a=1}^p  T_{{V}_a^{[\tilde{s_a}]},\abs k} \right\}   e_k ,\chi^{[\tilde{u}]}  
 \left\{ \prod_{b=1}^{m}
 S_{W_{b}^{[\tilde{t_{b}}]},\abs {k'}}^{(\tilde{\mu_{b}},\tilde{\nu_b})} \right\}  \psi } . \nonumber
\end{align} 
Using  the above relation for every summand, we find that  \eqref{explicitformuladecay} 
holds for $n+1$.  

Now using \eqref{explicitformuladecay}, we  show 
\eqref{eq:k k prime bound}. For all $W$ and  $\mu$,  $\nu \in \{0,1\}$ there exists 
a constant $C$ such that  
\begin{equation} 
\label{eq:PIinductyy}  
\| S_{W,\kappa}^{(\mu,\nu)}  \psi \|_\infty   \leq C \kappa^\mu  \| \psi \|_\infty  
\end{equation} 
for all $\kappa \geq 0$ and $\psi \in \Cb(\IR^3)$.  This follows from \autoref{th:Tk properties} (a) 
by observing that   for each $\mu \in \IN_0$ there exists as constant $C_\mu$ such that 
for all $y \neq x$ we have $$\left| \partial_{x_j} \left(  \frac{1}{\abs{x-y}} \left( \frac{x_j-y_j}{|x-y|} \right)^\mu \right) \right| \leq C_\mu |x-y|^{-2} . $$
Now using \autoref{th:pm star} together with  \eqref{eq:PIinductyy}, we find  for fixed $l$  a constant $C$, such that  whenever   \eqref{coeffcond} equals one  we have 
\begin{align*}
& \left| 
\sc{ \left\{ \prod_{a=1}^p  T_{{V}_a^{[s_a]},\abs k} \right\}   e_k ,\chi^{[u]}  
 \left\{ \prod_{b=1}^{m}
 S_{W_{b}^{[t_{b}]},\abs {k'}}^{(\mu_{b},\nu_b)} \right\}  \psi } \right|  \\
& \leq  \frac{C   \sup_{|\alpha| \leq n} \| ( 1+ \sc x^{n_0} ) \partial^\alpha \chi  \|_1  \nn{\psi}_\infty}{(1+ \abs k^{\lfloor \frac{p-1}{2} \rfloor} )(  1+ \abs{ k'}^{\lfloor \frac{l-2}{2} \rfloor })    } 
 |k'|^{|\underline{\mu}|}  .
\end{align*}
Now since  \eqref{coeffcond} equal one implies  $|\underline{\nu}| \geq m-l$ and 
$
|\underline{\mu}|  + |\underline{\nu}| \leq n 
$
we find  that 
$$
|\underline{\mu}|  \leq n - m  +  l  .
$$
Thus we find  for all $l=1,\ldots,m+1$  
$$
\left\lfloor \frac{l-2}{2} \right\rfloor - |\underline{\mu}|  
\geq     \frac{m}{2} - 2 - n  .
$$
Since $j = 1,2,3$ was arbitrary, the  claimed inequality now follows.
\end{proof}

\begin{prop}
\label{th:phir phi0 final}
Let $s \in \IN_0$, $n \in \IN$,  $\eta \in \S(\IR^3)$, $X \in\{  \Id, \nabla_k, \nabla_{k'} \}$,   $Y \in \{ k \nabla_k + k' \nabla_k', \eta(k) k \nabla_k \}$. Then  there exists a  constant $m_0 \in \IN$, such for all $m \geq m_0$ and $p \in \IN$, there are  $n_1, n_2 \in \IN$, $C$, such that for all  $\chi \in \S(\IR^3)$, $k,k' \not= 0$ 
\begin{align*}
&\abs{ X  Y^s \sc{ \phiz^{(p)}(k,\cdot), \chi \phir^{(m)}(k',\cdot) } } \leq \frac{C  \sup_{\abs \alpha \leq n_1}  \nn{\sc{\cdot}^{n_2} \partial^\alpha \chi}_1   }{(1+\abs{k}^n)(1+\abs{k'}^n)} . 
\end{align*}
\end{prop}
\begin{proof}
By applying \autoref{th:k nabla lemma 2} and \eqref{eq:tk r derivative} for the left and right part of the inner product, respectively,  we can write
\[
X Y^s \sc{ \phiz^{(p)}(k,\cdot), \chi \phir^{(m)}(k',\cdot) }
\]
for all given $X$ and $s$ as a linear combination of expressions
\[
k^\alpha (k')^\beta  f(k,k') \sc{T^{(n_1)}_{V_1,\abs k}   \cdots    T^{(n_p)}_{V_p,\abs k} e_k ,\chi   T^{(n'_1)}_{W_1,\abs {k'}}   \cdots    T^{(n'_m)}_{W_m,\abs {k'}} \phi(k',\cdot)    } ,
\]
where $\alpha,\beta$ are multi-indices with $\abs \alpha, \abs \beta \leq s$, $f$ is a bounded function on $\IR^3 \times \IR^3$,  $V_1, \ldots, V_p$, $W_1, \ldots, W_m \in \Cci(\IR^3)$, and $n_1, \ldots, n_p$,  $n'_1, \ldots, n'_m \in \IN_0$. Now we can estimate these expressions with \autoref{th:phir phi0}. 
\end{proof}


\subsection{Commutator with the Interaction}
\label{subsec:commutator with interaction}
This part provides the key for the proof of \thref{th:G commutator conditions}.  In the following we omit for the moment the regularity function $\kappa$  of the coupling and work with multiplication operators $H(\omega,\Sigma)$, $(\omega,\Sigma) \in \IR_+ \times \IS^2$. Throughout this section 
we shall always assume 
 \begin{equation}\label{eq:defofH}
H(\omega,\Sigma)(x) = \chi(x)  \tilde H(\omega,\Sigma)(x)
\end{equation}   where 
$\chi \in \S(\IR^3)$ and 
$\tilde{H}$ is a function on $I \times \IS^2 \times \IR^3$, where $I = (0,\infty)$ or $I = [0,\infty)$, such that  for some $s \in \IN_0$ the following holds.
\begin{itemize}
\item[($\text{J}_s$)] \label{Js}
 For all  $n \in \{0,\ldots, s\}$
and $\alpha \in \IN_0^3$   the partial derivatives  
 $ \partial_x^\alpha  \partial^n_\omega \tilde{H}$ exist  and are continuous on $I \times \IS^2  \times \IR^3$,
and 
 there exists a polynomial $P$ and $M \in \IN_0$ such that
\[
\abs{  \partial_x^\alpha  \partial^n_\omega \tilde{H}(\omega,\Sigma)(x) } \leq P(\omega) \langle x \rangle^M , \qquad (\omega,\Sigma,x) \in I \times \IS^2 \times  \IR^3. 
\]
\end{itemize} 
To show that a commutator $[T,\Apse]$, for a bounded  operator $T$ on $L^2(\IR^3)$,  is bounded, we  shall make of use the following decomposition on $\ell^2(N) \oplus L^2(\IR^3)$,
\[
\Vcd [T,\Apse] \Vcd^* = \mm{  0 &  \Vd T \Vc^* \ee \Adil \ee \\ -\ee \Adil \ee \Vc T \Vd^* & [\Vc T \Vc^*, \ee \Adil \ee]  },
\]
where $\Vcd$ is the unitary operator defined in \eqref{eq:Vcd} and $N \in \IN$ is the number of linearly 
independent eigenfunctions of $H_\p$. We treat the off-diagonal terms in \thref{th:psi_d bounded} and the term on the diagonal in \thref{th:psi_c bounded}.
\begin{prop}
\label{th:psi_d bounded} Suppose $\tilde{H}$ satisfies \refJ{0}. Then 
for all $n \in \IN_0$,  $j \in \{1,2,3\}$, $(\omega,\Sigma)$, the operators
\begin{enumerate}[label=(\arabic*)]
\item \label{pdsic_op1}
 $  \Adil ^n \Vc \Hos \Pd $,
\item \label{pdsic_op2}
 $\q_j \Adil ^n \Vc \Hos \Pd $,
\item \label{pdsic_op3}
 $ \po_j  \Adil ^n \Vc \Hos \Pd$,
\end{enumerate}
are well-defined,  their norms  can be estimated  uniformly in $\Sigma$ by a polynomial in $\omega$, 
and they are  continuous  in $(\omega,\Sigma)$.  Furthermore, if  $\tilde{H}$ satisfies \refJ{s},
then  \ref{pdsic_op1}--\ref{pdsic_op3} are $s$ times continuously differentiable with respect to $\omega$ in the operator norm  topology 
and  
\begin{align*} 
 \partial_\omega^s  \Adil ^n \Vc \Hos \Pd   & =   \Adil ^n \Vc  \partial_\omega^s \Hos \Pd  , \\
  \partial_\omega^s \q_j \Adil ^n \Vc \Hos \Pd   & = \omega \q_j \Adil ^n \Vc \partial_\omega^s  \Hos \Pd , \\
    \partial_\omega^s  \po_j  \Adil ^n \Vc \Hos \Pd  & =   \po_j  \Adil ^n \Vc  \partial_\omega^s  \Hos \Pd .
\end{align*} 
\end{prop}
\begin{proof}  Let $m \in \{0,1\}$ and $n \in \IN_0$. 
Choose $N$ big enough so that we find by means of \thref{th:final:phir} a constant $C$  and an $n_0  \in \IN_0$ such
that  
\begin{align}
\label{eq:d2:pre}
\abs{ \int  \partial^{m}_{k_j} \hDk^n     \phir^{(N)}(k,x) f(x) \psi_\d(x) \d x } \leq \frac{C \| \langle \cdot \rangle^{n_0} f \| \nn{\psi_\d}}{1+ \abs k^6}, 
\end{align}
for all  $f \in \mathcal{S}(\IR^3)$ ,  $\psi_\d \in \ran \Pd$ and $k \not= 0$.
Expanding $\phi(k,x)$ using \autoref{prop:expansionofscatter}  we obtain for $\psi_\d \in \ran \Pd$, $k \not= 0$,
\begin{align} \label{eq:forpointwisecont1}
\Vc \Hos \psi_\d (k) &=  (2\pi)^{-3/2}  \int \overline{\phi(k,x)} H(\omega,\Sigma)(x) \psi_\d(x)  \d x \\
&= T_0(\omega,\Sigma, k) + T_{\mathsf{R}}(\omega,\Sigma,k) , \nonumber 
\end{align} 
where 
\begin{align} 
 T_0(\omega,\Sigma; k)   & :=  (2\pi)^{-3/2}   \sum_{l=0}^{N-1} \int \overline{\phiz^{(l)}(k,x)} H(\omega,\Sigma)(x) \psi_\d(x) \d x \label{eq:d1} , \\
T_{\mathsf{R}}(\omega,\Sigma; k)  & := 
 (2\pi)^{-3/2} \int \overline{\phir^{(N)}(k,x)}H(\omega,\Sigma)(x) \psi_\d(x) \d x  \label{eq:d2} . 
\end{align}
The terms which appear if we apply $\Adil^n$, $\q_j$, $\po_j$, $n \in \IN_0$,  $j \in \{1,2,3\}$  to 
\eqref{eq:d1}  can be estimated by means of   \thref{th:final:phiz2} with the result 
 that for some constant $C$ and $n_1, n_2 \in \IN_0$ 
\begin{align} \label{eq:ineqdomconv1}
& |  \Adil^n  T_0(\omega,\Sigma;  k) |  , \  | \q_j \Adil^n  T_0(\omega,\Sigma ;  k)   | , \    | \po_j  \Adil^n T_0(\omega,\Sigma ;  k) | \  \\
& \quad  \leq     \frac{ C  \sup_{\abs \alpha \leq n_1} \nn{\sc{\cdot}^{n_2} \partial_x^\alpha ( H(\omega,\Sigma) \psi_\d)}_1  }{1 + \abs k^2}   \nonumber 
\end{align}
for all  $(\omega, \Sigma)$,   $k \neq 0$, and  $\psi_\d \in \ran \Pd$. 
 The terms  coming from \eqref{eq:d2} can be estimated using  \eqref{eq:d2:pre} such that for some constant $C$ and $n_1 \in \IN_0$,
\begin{align} \label{eq:ineqdomconv2}
& |  \Adil^n  T_{\mathsf{R}}(\omega,\Sigma ;  k) |  , \  | \q_j   \Adil^n  T_{\mathsf{R}}(\omega,\Sigma ;  k)   | , \    | \po_j  \Adil^n  T_{\mathsf{R}}(\omega,\Sigma ; k) | \ \\
& \quad   \leq       \frac{C \| \sc{\cdot}^{n_1} H(\omega,\Sigma)  \| \nn{\psi_\d}  }{1+ \abs k^2} 
\nonumber  
\end{align}
for all  $(\omega, \Sigma)$,   $k \neq 0$, and  $\psi_\d \in \ran \Pd$. 
Now observe that by elliptic regularity (cf. \cite[IX.6]{rs2}) we have $\psi_\d \in C^\infty(\IR^3)$ and $\partial^\alpha \psi_\d \in L^2(\IR^3)$  for all $\alpha \in \IN_0^3$. Thus, as  $\tilde{H}$ satisfies \refJ{0}, it follows that for fixed $\psi_\d$  and $H$  there exists a polynomial $P$ and $n_1, n_2 \in \IN_0$ such that for all $(\omega,\Sigma)$ 
\begin{align*}
   \sup_{\abs \alpha \leq n_1} \nn{\sc{\cdot}^{n_2} \partial_x^\alpha ( H(\omega,\Sigma) \psi_\d)}_1   , \ 
 \| \langle \cdot \rangle^r H(\omega,\Sigma)  \| 
\leq P(\omega) ,
\end{align*} 
using Cauchy-Scharz and standard estimates involving Schwartz functions. 
Collecting esimates and  using that the discrete spectrum is finite we see  that the operators \ref{pdsic_op1}, \ref{pdsic_op2} and \ref{pdsic_op3} are well-defined and their norms can be estimated by a polynomial in $\omega$.  Continuity in $(\omega,\Sigma)$ with respect to the operator norm topology 
now follows from  linearity, the bounds \eqref{eq:ineqdomconv1} and \eqref{eq:ineqdomconv2}, and the fact that   $\tilde{H}$ satisfies \refJ{0} (and again  standard estimates involving  Schwartz functions). 
If $s=1$, an analogous argument implies differentiability in $\omega$ 
with the derivative given by replacing $H$ by $\partial_\omega H$. Now the claim for arbitrary $s$ 
follows by induction. 
\end{proof}
\begin{lemma}
\label{th:integral kernels bounded} Suppose $\tilde{H}$ satisfies \refJ{0}.  Then 
for all $(\omega,\Sigma)$ and $k,k' \in \IR^3$ let
\begin{align}
\label{eq:K defn}
K_{\omega,\Sigma}[H](k,k') := \int \overline{\phi(k,x)} \Hos(x) \phi(k',x) \d x .
\end{align}
Then for all  $Z \in \{ k \nabla_k + k' \nabla_k', \eta_1(k)  k \nabla_k + \eta_2(k) + 
\eta_1(k') k' \nabla_{k'} + \eta_2(k')   \}$, where $\eta_1, \eta_2  \in \mathcal{S}(\IR^3)$,  $j \in \{1,2,3\}$, and $s \in \IN_0$ 
there exists a polynomial $P$ such that  the absolute values of
\begin{enumerate}[label=(\arabic*)]
\item \label{kn_first} $Z^s  K_{\omega,\Sigma}[H](k,k')$,
\item \label{kn_lastg} $\partial_{k_j} Z^s  K_{\omega,\Sigma}[H](k,k')$, $\partial_{k'_j} Z^s  K_{\omega,\Sigma}[H](k,k')$,
\item \label{kn_last}  $ (k_j - k_j') Z^s  K_{\omega,\Sigma}[H](k,k')$,
\end{enumerate}

are  bounded from above by
\begin{align}
\label{eq:integral kernel estimate} 
P(\omega) \left( \frac{1}{(1+\abs k^2)(1+ \abs {k'}^2)} + \frac{1}{1+\abs{k-k'}^4} \right) 
\end{align}
 for all   $(\omega,\Sigma)$,  $k,k' \not=0$ .  Furthermore the following is satisfied.
\begin{enumerate}[label=(\alph*)]
\item \label{partaofintegralkerest}  For fixed $k,k' \not= 0$, the functions $\IR_+ \times \IS^2 \rightarrow \IC$ mapping $(\omega,\Sigma)$ to the expressions \ref{kn_first}--\ref{kn_last}, are continuous.
If  $\tilde{H}$ satisfies \refJ{s}, these functions  are  $s$ times continuously differentiable  in $\omega$ and  the $s$-th partial
derivative  with respect to $\omega$ is obtained  by replacing $ H$     by   $  \partial_\omega^s H$. 
\item  \label{partbofintegralkerest}   The  integral kernels  \ref{kn_first}--\ref{kn_last} define bounded operators in $L^2(\R^3)$ whose  norms are uniformly bounded in $\Sigma$ by a polynomial
in $\omega$. With respect to the operator norm toplogy the following holds. These  operators depend continuously on $(\omega,\Sigma)$. If  $\tilde{H}$ satisfies \refJ{s}, these operators   are  $s$ times continuously differentiable  in $\omega$ and  the $s$-th partial
derivative  with respect to $\omega$ is obtained  by replacing $ H$     by   $  \partial_\omega^s H$. 
\end{enumerate}
\end{lemma}
\begin{proof} 
 Let $X \in \{ \Id , \partial_{k_j}, \partial_{k'_j}, k_j - k_j'\}$.  Assume first that
 \[Y \in \{ k \nabla_k + k' \nabla_k', \eta_1(k)  k \nabla_k  \}.\]
 Fix $s \in \IN_0$. 
Using \autoref{prop:expansionofscatter}  we write 
\begin{align} \label{eq:forpointwisecont2}
K_{\omega,\Sigma}[H](k,k') &= \int \overline{\phi(k,x)} \Hos(x) \phi(k',x) \d x \\ 
&= \sum_{l,l'=0}^{N-1} \int  \overline{ \phiz^{(l')}(k,x) } \Hos(x) \phiz^{(l)}(k',x) \d x \label{eq:t1} \\
&\qquad + \sum_{l=0}^{N-1} \int \overline{ \phir^{(N)}(k,x)} \Hos(x) \phiz^{(l)}(k',x) \d x \label{eq:t2:1} \\ &\qquad +\sum_{l=0}^{N-1}  \int \overline{ \phiz^{(l)}(k,x)} \Hos(x)  \phir^{(N)}(k',x)  \d x   \label{eq:t2:2}  \\
&\qquad + \int \overline{ \phir^{(N)}(k,x) } \Hos(x) \phir^{(N)}(k',x) \d x.  \label{eq:t3} 
\end{align}
 By  \thref{th:phir phi0 final} we can choose    $N$ large enough such that   there exist constants $ n_1, n_2 \in \IN$, $C$, such that for all  $f \in \S(\IR^3)$ and $k,k' \not= 0$, and $p=1,\ldots,N$,
\begin{align}   \label{eq:intkerest0} 
&\abs{ X  Y^s \sc{ \phiz^{(p)}(k,\cdot), f  \phir^{(N)}(k',\cdot) } } \leq \frac{C  \sup_{\abs \alpha \leq n_1}  \nn{\sc{\cdot}^{n_2} \partial^\alpha f}_1   }{(1+\abs{k}^2)(1+\abs{k'}^2)} ,  
\end{align}
which  implies that for all $(\omega,\Sigma)$ 
\begin{align}  \label{eq:intkerest1} 
| X Y^s  \eqref{eq:t2:1} | , \ | X Y^s  \eqref{eq:t2:2} | 
\leq  \frac{C  \sup_{\abs \alpha \leq n_1}  \nn{\sc{\cdot}^{n_2} \partial^\alpha H(\omega,\Sigma)}_1  }{(1+|k|^2)(1+|k'|^2)}  , 
\end{align} 
Moreover, by   \thref{th:final:phiz}   there are constants $n_1, n_2 \in \IN$, $C$,  such that for all $k,k' \not= 0$,
\begin{align}  \label{eq:intkerest2} 
| X Y^s \eqref{eq:t1} | \leq    \frac{ C \sup_{\abs \alpha \leq n_1} \nn{\sc{\cdot}^{n_2} \partial_x^\alpha H(\omega,\Sigma)}_1}{1+\abs{k-k'}^4}  . 
\end{align} 
Finally, using   \thref{th:final:phir} we see, by possibly making $N$ 
larger, that there exist constants $ n_1 \in \IN$ and  $C$ such that 
\begin{align}  \label{eq:intkerest3} 
  | X Y^s  \eqref{eq:t3} | \
\leq  \frac{C \| \langle \cdot  \rangle^{n_1}  H(\omega,\Sigma) \|_1 }{(1+|k|^2)(1+|k'|^2)}  . 
\end{align} 
On the other hand since $\tilde{H}$ satisfies \refJ{0} and $H = \chi \tilde{H}$,  there exists for each $n_1 \in \IN_0$ and $\alpha \in \IN_0^3$ a polynomial $P$ such that 
\begin{equation} \label{eqestonkernel55} 
\nn{  \langle \cdot \rangle^{n_1} \partial_x^{\alpha} \Hos ) }_1 \leq P(\omega), \qquad (\omega,\Sigma) \in \IR_+ \times \IS^2 . 
\end{equation} 
It follows as a consequence of   \eqref{eq:intkerest0}--\eqref{eqestonkernel55} that 
\begin{align} \label{basickernelest1}
| X Y^s K_{\omega,\Sigma}[H](k,k')  | \leq \text{ r.h.s. of 
 \eqref{eq:integral kernel estimate} }  . 
\end{align} 
This shows  \ref{kn_first}--\ref{kn_last} in case $Z = k \nabla_k + k' \nabla_{k'} $. We note that 
 $Y =  \eta_1(k)  k \nabla_k $ will be used below. 

Let us now assume \begin{align} \label{eq:assumpforZ}
 Z = \eta_1(k)  k \nabla_k + \eta_2(k) + 
\eta_1(k') k' \nabla_{k'} + \eta_2(k') .
\end{align} 
  To estimate derivatives acting on both sides 
of \eqref{eq:forpointwisecont2} 
we  use \autoref{th:phi derivatives bounded}, with the result  that for all $r,r' \in \{0,1\}$ and $s,s' \in \IN_0$  there exist  $n_1 \in \IN_0$ and  $C$ such that 
for all nonzero $k,k'$ 
\begin{align} \label{eq:forpointwisecont2b}
|\partial_{k_j}^r ( k \nabla_k)^s \partial_{k'_j}^{r'} (k' \nabla_{k'})^{s'}K_{\omega,\Sigma}[H](k,k') | 
  \leq  C \| \langle \cdot \rangle^{n_1} H(\omega,\Sigma) \|_1 \langle  k \rangle^{s} \langle k' \rangle^{s'}  . 
\end{align}
To estimate the norm occurring on the right-hand side we shall use that for  $\tilde{H}$ satisfying \refJ{0} and $n_1 \in \IN_0$ there exists  a  polynomial $P$ such that 
\begin{align} \label{eq:forpointwisecont2c}
 \| \langle \cdot \rangle^{n_1} H(\omega,\Sigma) \|_1  \leq P(\omega)   . 
\end{align}
Let $W(k) = \eta_1(k)  k  \nabla_k  + \eta_2(k) $.   
Then by the binomial theorem 
\begin{align} 
 Z^n K_{\omega,\Sigma}[H](k,k')   & = \sum_{l=0}^n   \binom n l   W(k)^{l} W(k')^{n-l} K_{\omega,\Sigma}[H](k,k')  .  \label{eq:derofbound222}
\end{align} 
We see, after  commuting Schwartz functions to the left, that  for each $l \geq 1$ there exist functions $\eta^{(l,s)},  \tilde{\eta}^{(l,s)} \in \mathcal{S}(\IR^3)$, $0 \leq s \leq l$, such that 
\begin{align} \label{eq:commidenity}  
 W(k)^{l} = \sum_{s=0}^l \eta^{(l,s)} (\eta(k) k \nabla_k)^s  = \sum_{s=0}^l \tilde{\eta}^{(l,s)} (k \nabla_k)^s .
\end{align} 
Let us first consider the terms  in  \eqref{eq:derofbound222} for $l=0$ and $l=n$.  Using the first equality in \eqref{eq:commidenity} 
and  \eqref{basickernelest1}   (as well as its adjoint)  for $Y =  \eta_1(k)  k \nabla_k $,  we find 
\begin{align} \label{basickernelest12}
| X W^n(k)  K_{\omega,\Sigma}[H](k,k')  | , \  | X W^n(k')  K_{\omega,\Sigma}[H](k,k')  |  \leq \text{ r.h.s. of  \eqref{eq:integral kernel estimate} }  .
\end{align} 
The terms in \eqref{eq:derofbound222} for $l \in \{1,\ldots,n-1\}$ are estimated using  \eqref{eq:forpointwisecont2b},  the second equality
in   \eqref{eq:commidenity}   controlling the growth in $k$ and $k'$, and finally \eqref{eq:forpointwisecont2c}. Thus  we find  with  \eqref{basickernelest12}
\begin{align*} 
| X  Z^n K_{\omega,\Sigma}[H](k,k')  |   \leq \text{ r.h.s. of  \eqref{eq:integral kernel estimate} }  .
\end{align*} 
This shows  \ref{kn_first}--\ref{kn_last} in  the case  \eqref{eq:assumpforZ}. It remains to prove \ref{partaofintegralkerest} and \ref{partbofintegralkerest}.
\begin{enumerate}[label=(\alph*),wide,  labelindent=0pt]
\item
The continuity property in $(\omega,\Sigma)$  for fixed nonzero $k, k'$ can be seen from the  integral  \eqref{eq:forpointwisecont2},
using dominated convergence with the property that $\tilde{H}$ satisfies \refJ{0}. To this end, we note that the integrand
contains a Schwartz function  and that  the derivatives of the scattering functions are bounded by polynomials,
as shown in  \thref{th:phi derivatives bounded}.
If $s=1$, we conclude analogously differentiability in $\omega$, and furthermore, that the derivative 
is given by replacing $H$ with $\partial_\omega H$. 
For arbitrary  $s$  the claim then   follows by induction.
\item We first note, that operators with integral kernels satisfying  a bound  \eqref{eq:integral kernel estimate} are bounded by $P(\omega)$.
To this end, observe that an integral operator $T$ with integral kernel
\[
t  \: \quad (k,k') \mapsto \frac{1}{(1+\abs k^2)(1+ \abs {k'}^2)} 
\]
is Hilbert-Schmidt and its norm is estimated by  $\| T \| \leq \| t \|_2$, and that  an operator $S$  with integral kernel 
\[
 \quad (k,k') \mapsto \frac{1}{1+\abs{k-k'}^4} =: s(k-k') 
\]
is bounded by Young's inequality for convolutions: $\| S \psi \|_2  = \nn{s * \psi}_2 \leq \nn{s}_1 \nn{\psi}_2$, $s \in L^1(\IR^3)$, $\psi \in L^2(\IR^3)$. 
In view of this, continuity in $(\omega,\Sigma)$ with respect to the operator norm topology 
now follows from  linearity, the bounds   \eqref{eq:intkerest1}--\eqref{eq:intkerest3}
as well as    \eqref{eq:forpointwisecont2c}, and the fact that   $\tilde{H}$ satisfies \refJ{0} (and a standard estimate 
involving   Schwartz functions). If $s=1$, we conclude analogously differentiability in $\omega$, and   that the derivative 
is given by replacing $H$ with $\partial_\omega H$. 
For arbitrary  $s$  the claim then  follows by induction. \qedhere  
\end{enumerate}
\end{proof}
\begin{rem}  We note that for the proof of the main theorem we  will only use  
 Part \ref{partbofintegralkerest} of \autoref{th:integral kernels bounded}  and 
Part \ref{partaofintegralkerest} will not be needed.   We  nevertheless included   Part \ref{partaofintegralkerest} in \autoref{th:integral kernels bounded}, since in principle  we could work 
with a weaker topology.
\end{rem}

In the following lemma we estimate the coupling functions first in  ``scattering space''.

\begin{lemma}
\label{th:psi_c bounded} Suppose $\tilde{H}$ satisfies \refJ{0}. 
Then for all $\epsilon \geq 0$, $n \in \IN_0$,  $j \in \{1,2,3\}$,  $(\omega,\Sigma)$
\begin{enumerate}[label=(\arabic*)]
\item 
\label{item:psi_c Apn}
 $ \ad^{(n)}_{\ee \Adil \ee }(  \Vc \Hos   \Vc^* ) $,
\item 
\label{item:psi_c qj}
 $ \ad_{\q_j} \left(\ad^{(n)}_{\ee \Adil \ee }(  \Vc  \Hos  \Vc^* ) \right) $,
\item 
\label{item:psi_c pj}
 $ \ad_{\po_j} \left(\ad^{(n)}_{ \ee \Adil \ee }(  \Vc \Hos \Vc^*  )\right) $,
\item
\label{item:psci_c xj}
 $  \po_j   \ad^{(n)}_{\Adil }(  \Vc  \Hos  \Vc^*  ) $,
\end{enumerate}
are well-defined bounded operators in  $L^2(\R^3)$  and we can estimate their norms uniformly in $\Sigma$ by a polynomial in $\omega$.  With respect to the operator norm toplogy the following holds. 
\ref{item:psi_c Apn}--\ref{item:psci_c xj}
are continuous  $\L(L^2(\IR^3))$-valued functions of  $(\omega,\Sigma)$. Moreover if  $\tilde{H}$ satisfies \refJ{s}, then the functions 
\ref{item:psi_c Apn}--\ref{item:psci_c xj}  are $s$ times continuously differentiable with respect to $\omega$ and     the $s$-th partial
derivative of   \ref{item:psi_c Apn}--\ref{item:psci_c xj}  with respect to $\omega$ is obtained  by replacing $ H$    by     $  \partial_\omega^s H$. 
\end{lemma}

\begin{proof}
 From  \autoref{th:scattering functions properties} we see that  for all $(\omega,\Sigma)$,
\[
\Vc \Hos   \Vc^* \psi(k) = (2\pi)^{-3} \int K_{\omega,\Sigma}[H](k,k') \psi(k') \d k',
\]
with  $K_{\omega,\Sigma}$ defined  in \eqref{eq:K defn}.  Thus the lemma 
follows  directly from 
\autoref{th:integral kernels bounded}, observing that 
$\eta_\epsilon A_D \eta_\epsilon =
\frac{\i \eta_\epsilon^2(k)}{2}  k \nabla_k + \frac{ \i \eta_\epsilon(k)}{4} ( 2 k \nabla_{k} \eta_\epsilon(k) + 3 \eta_\epsilon(k) )$.   
\end{proof}

Let us now prove the  central proposition of this section, which can be thought of as a preliminary  version of \thref{th:G commutator conditions} but  without the cutoff  function $\kappa$. 
 For the proof  
we need the following auxiliary lemma.  
\begin{lemma} \label{th:ad n Hos transform}  Let $\HH_0$ and $\HH_1$ be Hilbert spaces.  Let $B$ be a bounded operator in $\HH_0$ and let $V \: \HH_0 \to \HH_1$ a partial isometry with  $\ran V = \HH_1$.
Let $P$ be the orthogonal projection onto the kernel of $V$. Suppose  $A$ is a self-adjoint operator in $\HH_1$  
such that   
for all $j=1,\ldots,n$   the set  $\ran  V (V^* A V)^{j-1} B P$ is contained in the domain of $A$ 
and the operators   $\ad_A^{(j)}(V B V^*)$ and  $ (V^* A V)^j    B P$ are  bounded. 
Then $\ad_{V^* A V}^{(n)}(B)$ is a bounded operator     on $\HH_0$   and  
\begin{align*}
\ad_{V^* A V}^{(n)}(B) =  V^* \ad_{ A }^{(n)}(V B V^*) V +   ( \i V^* A V )^n B  P +  P B ( -\i V^* A V )^n   .
\end{align*} 
\end{lemma}
\begin{proof} This follows by induction in $n$ and a straightforward calculation. 
\end{proof}

\begin{prop}
\label{th:G commutator conditions 2} Suppose $\tilde{H}$ satisfies \refJ{s}.  Then 
for all $\epsilon \geq 0$, $n \in \IN_0$, $j \in \{1,2,3\}$,  $(\omega,\Sigma)$,  $r=0,\ldots,s$ the operators
\begin{enumerate}[label=(\arabic*)]
\item 
\label{it:Aps}
 $ \ \partial_\omega^r \ad^{(n)}_{\Apse}(  \Hos  )$,
\item 
\label{it:Aps:ad qj}
 $ \  \partial_\omega^r  \ad_{\Vc^* \q_j \Vc}\left( \ad^{(n)}_{\Apse}(  \Hos  ) \right)$,
\item  
\label{it:Aps:ad pj} 
 $ \  \partial_\omega^r  \ad_{\Vc^* \po_j \Vc}\left( \ad^{(n)}_{\Apse}(  \Hos  ) \right)$,
\item
\label{it:Aps:pj}
 $ \  \partial_\omega^r \Vc^* \po_j \Vc   \ad^{(n)}_{\Aps}(  \Hos  ) $ and  $   \partial_\omega^r \ad^{(n)}_{\Aps}(  \Hos  )  \Vc^* \po_j \Vc   $,
\end{enumerate} 
where the derivative $\partial_\omega$ is understood with respect to the operator norm topology, 
are well-defined, bounded, and we can estimate their norms uniformly in $\Sigma$  by a polynomial in $\omega$. The operators \ref{it:Aps}--\ref{it:Aps:pj}
depend continuously on $(\omega,\Sigma)$  with respect to the operator norm topology.
\end{prop}
\begin{proof}
Follows directly from \autoref{th:psi_d bounded}, \autoref{th:psi_c bounded}  and
an application of 
\autoref{th:ad n Hos transform}, with $V = V_\c$, $P = P_{\mathrm{ disc}}$,
$A =  \Adil$, and $B = H(\omega,\Sigma)$. 
\end{proof}

\begin{proof}[Proof of \thref{th:G commutator conditions}] To show the proposition  we will use   \thref{th:G commutator conditions 2} and \autoref{lemma:gluing_integrable}.
First we consider  the  case where  \ref{I3:first}  of   \ref{assumption:I3} holds.
Let $F, \tilde{F}  \: \IR_+ \times \IS^2 \rightarrow \Lb(\H_\p)$, where  
  $F(\omega,\Sigma) =  \kappa(\omega) \tilde F(\omega,\Sigma)$ and 
$\tilde{F} $ is one of the functions
\begin{equation}
\begin{aligned}
\label{eq:adn operators}
&\ad^{(n)}_{\Apse}( \chi \tilde{G}(\cdot) ), ~\ad_{\Vc^* \q_j \Vc}\left(\ad^{(n)}_{\Apse}( \chi \tilde{G}(\cdot) )\right),~ \ad_{\Vc^* \po_j \Vc}\left(\ad^{(n)}_{\Apse}( \chi \tilde{ G}(\cdot))\right), 
 \\  &\Vc^* \po_j \Vc    \ad^{(n)}_{\Aps }( \chi \tilde{G}(\cdot) )  , 
\end{aligned}
\end{equation}
for  $j \in \{1,2,3\}$. Since  \ref{assumption:I1} implies that $\tilde{G}$ satisfies \refJ{3} 
we find that  the $ \Lb(\H_\p)$-valued functions 
\eqref{eq:adn operators} are   well-defined by \thref{th:G commutator conditions 2} as 
well as their first three partial   derivatives with respect to  $\omega \in \IR_+$. From   Leibniz' rule we find  for $m \leq 3$ 
\begin{align} \label{eq:leibniz4p2}
\partial_\omega^m F(\omega,\Sigma) = \sum_{l=0}^m \binom{m}{l} \partial_\omega^{l} \kappa(\omega)   \partial_\omega^{m-l} \tilde F(\omega,\Sigma)  . 
\end{align}
By  \thref{th:G commutator conditions 2}  there exists a 
 polynomial $P$ such that for all $l = 0, \ldots, m$, $(\omega,\Sigma)$,
\begin{align}
\label{eq:ft P}
\nn{   \partial_\omega^{m-l} \tilde F(\omega,\Sigma) } \leq P(\omega).
\end{align}
 Now Condition \ref{gluing:uv condition} of  \autoref{lemma:gluing_integrable} holds for $F$  by  \eqref{eq:ft P},  \eqref{eq:leibniz4p2}  and  \ref{assumption:I2}.  Condition \ref{gluing:ir condition}  of 
 \autoref{lemma:gluing_integrable} is seen to  hold  by  \eqref{eq:ft P},  \eqref{eq:leibniz4p2}  and   \ref{I3:first}  of   \ref{assumption:I3}.
Thus, by    \autoref{lemma:gluing_integrable}   $(u,\Sigma) \mapsto\partial^m_u \tau_\beta(F)(u,\Sigma)$ belongs to $L^2(\IR \times \IS^2, \Lb(H_\p))$ for all $0 \leq m \leq 3$, and moreover 
\eqref{eq:mainineqbound} holds.  Hence we have shown  \thref{th:G commutator conditions}
in    case  \ref{I3:first}  of   \ref{assumption:I3} holds. 

Let us now assume the case  where  \ref{I3:second}   of \ref{assumption:I3} holds.  To this end, 
let $F_0, \tilde{F}_0  \:  [0,\infty) \allowbreak  \times \IS^2 \rightarrow \Lb(\H_\p)$, where  
  $F_0(\omega,\Sigma) =  \kappa_0(\omega) \tilde F_0(\omega,\Sigma)$ and 
$\tilde{F}_0 $ is one of the functions
\begin{equation}
\begin{aligned}
\label{eq:adn operators2}
&\ad^{(n)}_{\Apse}( \chi \tilde{G}_0(\cdot) ), ~\ad_{\Vc^* \q_j \Vc}\left(\ad^{(n)}_{\Apse}( \chi \tilde{ G}_0(\cdot) )\right),~ \ad_{\Vc^* \po_j \Vc}\left(\ad^{(n)}_{\Apse}( \chi \tilde{G}_0(\cdot))\right), 
 \\  &
 \Vc^* \po_j \Vc    \ad^{(n)}_{\Aps }(\chi \tilde{ G}_0 )  +   \ad^{(n)}_{\Aps }( \chi \tilde{G}_0 )  \Vc^* \po_j \Vc  
\end{aligned}
\end{equation}
for  $j \in \{1,2,3\}$. The verification of  Assumption \ref{gluing:uv condition}   of  \autoref{lemma:gluing_integrable} for $F_0$ 
is analogous to the first case. 
Now  \ref{I3:second}   of \ref{assumption:I3} implies 
that $\tilde{G}_0$ satisfies \refJ{s} where $s = \max\{0,3-J\}$   
we find that  the $ \Lb(\H_\p)$-valued functions 
\eqref{eq:adn operators} are   well-defined and continuous  by \thref{th:G commutator conditions 2} as 
well as their first $s$  partial   derivatives with respect to  $\omega \in [0,\infty)$. 
By assumption \ref{I3:second}   of \ref{assumption:I3} it is straightforward to verify  that $F_0$ satisfies the
 Assumption 
\ref{alternative ir}  of  \autoref{lemma:gluing_integrable}, 
 noting 
 that  the  adjoints  are obtained by  replacing  $\kappa_0 \chi \tilde{G}_0$   by $\overline{ (\kappa_0 \chi \tilde{G}_0)}$. 
 Thus \thref{th:G commutator conditions}
now follows  by \autoref{lemma:gluing_integrable}, observing that 
for    \ref{item:pj additional}  we use the identity 
\begin{align*} 
&  \Vc^* \po_j \Vc    \ad^{(n)}_{\Aps }( \tau_\beta(G )) \\
& \quad = \frac{1}{2} \left( \Vc^* \po_j \Vc    \ad^{(n)}_{\Aps }( \tau_\beta(G ) ) +   \ad^{(n)}_{\Aps }( \tau_\beta(G )) \Vc^* \po_j \Vc    - \i \ad_{\Vc^* \po_j \Vc}(\ad^{(n)}_{\Aps}( \tau_\beta( G) )) \right) .
\end{align*} 

Finally, note that the proposition still holds true if we replace $G$ by $G^*$ since the conditions \ref{assumption:I1}--\ref{assumption:I3} obviously follow for $G^*$.
 \qedhere 
\end{proof}

\section{Proof of Positivity and of the Main Theorem}
\label{sec:positivity}
In this section the main estimates and the positivity proof of the commutator are discussed. First we introduce the two terms $A_0$ and $\CQ{Q}$ which we add to $C_1$ as already mentioned in the overview of the proof. Then we show how the main theorem is proven given that we know that the sum of all three terms is positive (\thref{th:positive}). Subsequently, in \Autoref{subsec:error} we show how we estimate these three terms separately and which error terms occur. With that we conclude by proving \thref{th:positive}.

In this section we  assume that  \ref{ass:h1},  \ref{ass:h2},  \ref{assumption:I1}--\ref{assumption:I3} 
hold, i.e., the assumptions  of \autoref{th:main} are satisfied.

%
%
\subsection{Putting Things together, Proof of the Main Theorem}
\label{subsec:main proof}
Remember that we have on $\Ds$, 
\begin{equation} \label{eq:calccomm} 
C_1 = \Vc^* \q^2 \Vc  \otimes \Id_\p \otimes \Id_\f +  \Id_\p \otimes  \Vc^* \q^2 \Vc    \otimes \Id_\f + \Nfh + \lambda W_1 , 
\end{equation} 
where $W_1$ was the commutator with the interaction, cf. \eqref{eq:cn_formula 2}. Obviously, $C_1$ is strictly positive on the orthogonal complement of the vacuum subspace for $\lambda = 0$ and its first two terms are positive on $(\ran (\Pd \otimes \Pd))^\perp \otimes \FF$. 

On the space $\ran \Pi$ we use the Fermi Golden Rule and introduce the corresponding conjugate operator $A_0$ in order to obtain a positive expression in \thref{th:fgr} as a commutator with $L_\lambda$.  Such an operator was  introduced  for zero temperature systems in \cite{a0firsttime} and later adapted to the positive temperature setting in \cite{merkli_positive}.
It is a bounded self-adjoint operator on $\H$, given by
\begin{equation} \label{eq:defofAzero} 
A_0 := \i \lambda (\Pi W R_\varepsilon^2 \oPi - \oPi R_\varepsilon^2 W \Pi),
\end{equation} 
where  $R_\varepsilon^2 := (L_0^2 + \varepsilon^2)^{-1}$ and $\varepsilon > 0$. One can show that $\ran A_0 \subseteq \Def(L_\lambda)$, since $\ran A_0 \subseteq \H_\p \otimes \H_\p \otimes \FF_\fin$ (cf. \autoref{th:a0 welldefined}). Furthermore, 
\muun
{
\label{eq:L A0}
\i [L_\lambda, A_0] = - \lambda [L_\lambda, \Pi W R_\varepsilon^2 \oPi - \oPi R_\varepsilon^2 W \Pi]
}
is bounded and extends to a self-adjoint operator as well, and we have \[\Pi \i[  L_\lambda, A_0] \Pi > 0\]for suitable $\varepsilon > 0$, see \thref{th:fgr}.

To obtain positivity on the remaining space $(\ker L_\p)^\perp \otimes \ran P_\Omega$ we introduce the bounded self-adjoint operator 
\begin{align}
\label{eq:CQ}
\CQ{Q} := Q \otimes P_\Omega    +   \frac{\lambda}{2}    W    ( L_\p^{-1}  Q \otimes  P_\Omega )   + \frac{\lambda}{2}    ( W ( L_\p^{-1}  Q \otimes P_\Omega)   )^*
\end{align}
on $\H$, where $L_\p^{-1}$ is to be understood in the sense of functional calculus as an unbounded operator, and
\begin{align}
\label{eq:Q defn}
Q := L_\p^2 \ind_{[-1,1]}(L_\p) + \ind_{(-\infty,-1) \cup (1,\infty)}(L_\p).
\end{align}
Notice that by construction $\ran Q \subseteq \Def(L_\p^{-1})$, $L_\p^{-1} Q$ is bounded and self-adjoint, so the definition of $\CQ{Q}$ makes sense. Furthermore, the first summand $ Q \otimes P_\Omega $  is indeed positive on $(\ker L_\p)^\perp \otimes \ran P_\Omega$. 

The goal is to show that the sum of the three operators \eqref{eq:calccomm}, \eqref{eq:L A0} and \eqref{eq:CQ} is  positive and has zero expectation with any element of the kernel of $L_\lambda$. 
To this end we define for $\psi \in \Def(\qce) \cap \Def(L_\lambda)$ and some $\theta > 0$, 
\begin{equation} \label{defoffinalcommutator} 
\qtot(\psi) := \qce(\psi) +  \theta \sc{\psi,\i [L_\lambda, A_0] \psi} + \sc{\psi,\CQ{Q} \psi},
\end{equation} 
where $\qce$ denoted the form corresponding to $C_1$. 
By the virial theorem for $C_1$ and by construction of the other two terms, this form is actually non-positive for any $\psi \in \ker L_\lambda$. This is the content of the following proposition. 
\begin{prop}
\label{th:virial}
For arbitrary $\lambda$, $\varepsilon$, $\theta$, and $\psi \in \ker L_\lambda$ we have $\psi \in \Def(\qce)$ and $\qtot(\psi) \leq 0$.
\end{prop} 
\begin{proof}
Let $\psi \in \ker L_\lambda$. By \thref{th:result virial} we know that $\psi \in \Def(\qce)$, and
\[
\qce(\psi) \leq 0.
\]
It is clear that  $\sc{\psi,\i [L_\lambda, A_0] \psi} = 0$. Furthermore, by construction, the last term vanishes:
\begin{align*}
0 &= \sc{ (L_\p^{-1} Q \otimes P_\Omega)  L_\lambda \psi, \psi} \\
&= \sc{ \psi, L_\lambda  (L_\p^{-1} Q \otimes P_\Omega)  \psi}\\
&= \sc{ \psi, (L_\p \otimes \Id_\f + \lambda W) (L_\p^{-1} Q \otimes P_\Omega)  \psi} \\
&= \sc{ \psi,  (Q \otimes P_\Omega) \psi  + \lambda W  (L_\p^{-1} Q \otimes P_\Omega )  \psi }.
\end{align*}
Then $\sc{\psi,\CQ{Q} \psi} = 0$ follows by adding the complex conjugate term. Thus,
\[
\qtot(\psi) = \qce(\psi) \leq 0.  \qedhere
\]
\end{proof}
We prove in the next subsection the following proposition which states that $\qtot$ is in fact positive. Remember that $\gb(\varepsilon)$ is the constant appearing in \ref{fgrc}.
\begin{prop}
\label{th:positive}
Let $\beta_0 > 0$ and $\varepsilon > 0$. Then there exists a constant $C > 0$ such that for all $\beta \geq \beta_0$ and $0 < \abs \lambda <  C \min\{ 1 , \gb^2(\varepsilon)\}$, we have $\qtot > 0$. That is, $\qtot \geq 0$ and $\qtot(\psi) = 0$ for some $\psi \in \Def(\qce) \cap \Def(L_\lambda)$ implies $\psi = 0$. 
\end{prop}

Let us  now give the  proof of the main theorem.  

\begin{proof}[Proof of  \thref{th:main}.] The theorem follows from  \autoref{th:positive} together with \autoref{th:virial}.
\end{proof} 

\subsection{Error Estimates}
\label{subsec:error}

In the following proposition we prove separate estimates from below for the three operators \eqref{eq:calccomm}, \eqref{eq:L A0} and \eqref{eq:CQ}.  We use the short-hand notation 
\[
\Poh := \Id_\p \otimes \Id_\p \otimes P_\Omega. 
\]
\begin{prop}
\label{th:total estimate steps}
The following holds.
\begin{enumerate}[label=(\alph*)]
\item
\label{it:C1 estimate}
There exist constants $c_1$, $c_2 > 0$ such that, for all $\lambda \in \IR$, we have in the sense of quadratic forms on $\Def(\qce)$,
\begin{align}
\label{eq:C1 estimate}
C_1 &\geq   \big[\Vc^* \q^2 \Vc  \otimes \Id_\p  + \Id_\p \otimes   \Vc^* \q^2 \Vc   \\ &- c_1 (1+ \beta^{-1}) \lambda^2 \big( \Vc^* \sc{\po}^{-2} \Vc \otimes \Id_\p +  \Id_\p  \otimes  \Vc^* \sc{\po}^{-2} \Vc + \PP{(\Pc \otimes \Pc)} \big)  \big]    \nonumber
\\ &\qquad  \otimes \Id_\f  + c_2  \PP{\Poh}. \nonumber
\end{align}
\item
For all $\varepsilon >0$ there exist constants $c_1, c_2, c_3 > 0$ (depending on $\varepsilon$) such that for $\abs \lambda < 1$,
\muu
{
\i [L_\lambda, A_0] &\geq (1- c_1 \abs \lambda) 2  \lambda^2 \Pi W R_\varepsilon^2  W \Pi  - c_2 (1+\beta^{-1})  \abs \lambda \PP{\Poh}  \\  &\qquad -  c_3 (1+\beta^{-1}) \lambda^2    \ind_{L_\p \not= 0} \big( \Vc^* \sc{\po}^{-2} \Vc     \otimes \Id_\p +  \Id_\p  \otimes \Vc^* \sc{\po}^{-2} \Vc \\ &\qquad + \PP{ (\Pc \otimes \Pc)}   \big) \ind_{L_\p \not= 0}  \otimes P_\Omega. 
}

\item
There exists a constant $c_1 > 0$ such that, for all $\lambda \in \IR$,
\muu
{
\CQ{Q} \geq (1-c_1 \abs \lambda (1+\beta^{-1}) ) Q \otimes P_\Omega - \abs \lambda \PP{\Poh}. 
}
\end{enumerate}
\end{prop}
Before we can give the proof of \thref{th:total estimate steps} we need some preparatory lemmas. 
It is convenient to introduce some further notation for the interaction and the commuted interaction. We separate them into parts which act on the left and right of the particle space tensor product, respectively, 
\begin{align*}
I^{(\l )}_1(u,\Sigma)  &:= (-\i \partial_u) \tau_\beta(G)(u,\Sigma)  +  \tau_\beta( \ad_{\Aps}(G))(u,\Sigma)  , \\ 
I^{(\r )}_1(u,\Sigma) &:=  (-\i \partial_u)  e^{-\beta u  /2}   \tau_\beta(\G^*)(u,\Sigma)    -   e^{-\beta u  /2}   \tau_\beta( \ad_{\Aps} (\G^*))(u,\Sigma).
\end{align*}
We introduce also integrated versions which will be used in the further estimates,
\begin{align*}
w &:=  \int I(u,\Sigma)^* I(u,\Sigma) \d(u, \Sigma), \\
w_1 &:=  \int I_1 (u,\Sigma)^* I_1(u,\Sigma) \d(u, \Sigma),
\end{align*}
and the left and right parts, for  $\alpha = \l,\r$, 
\begin{align*}
w^{(\alpha)} &:=  \int I^ {(\alpha)}(u,\Sigma)^* I^ {(\alpha)}(u,\Sigma) \d(u, \Sigma), \\
w_1^{(\alpha)}&:=  \int I_1^ {(\alpha)}(u,\Sigma)^* I_1^ {(\alpha)}(u,\Sigma) \d(u, \Sigma).
\end{align*}
By construction $W_1 = \Phi(I_1)$ and we have the decomposition
\muun
{
\label{eq:I1_decomposition}
I_1(u,\Sigma) &= I^{(\l)}_1(u,\Sigma) \otimes \Id_\p + \Id_\p  \otimes  I^{(\r)}_1(u,\Sigma).
}
First, we estimate the commuted interaction term $W_1$ appearing in \eqref{eq:C1 estimate}. 
For this, we prove a bound for $w$ and $w_1$.
\begin{lemma}
\label{th:I1_q}
There exist constants $C$ independent of $\beta$ such that, for  $\alpha = \l,\r$, 
\begin{enumerate}[label=(\alph*)]
\item $\nn{I^{(\alpha)}}^2_{ L^2(\IR \times \IS^2, \L(\H_\p) )} \leq C ( 1+ \beta^{-1})$,
\item $\nn{I_1^{(\alpha)}}^2_{ L^2(\IR \times \IS^2, \L(\H_\p) )}  \leq C ( 1+ \beta^{-1})$,
\item $\Pc w^{(\alpha)} \Pc \leq C ( 1+ \beta^{-1}) \Vc^* \sc{\po}^{-2} \Vc$,
\item $ \Pc w_1^{(\alpha)}\Pc \leq C ( 1+ \beta^{-1}) \Vc^*  \sc{\po}^{-2} \Vc$.
\end{enumerate}
\end{lemma}
\begin{proof}
By \thref{th:G commutator conditions}
we have that, for all $j \in \{1,2,3\}$, $\alpha = \l, \r$, 
\[ 
I^{(\alpha)},~
I^{(\alpha)}_1,~
I^{(\alpha)} \Vc^* \po_j \Vc, ~
I^{(\alpha)}_1 \Vc^* \po_j \Vc  \in L^2(\IR \times \IS^2, \L(\H_\p) ), 
\]
 and there is a constant $C$ independent of $\beta$ such that we can estimate the norm of these expressions by
\[
C (1+\beta^{-1}).
\]
Thus, the same applies for $ I^{(\alpha)} \Vc^* \sc{\po}  \Vc, ~ I^{(\alpha)}_1 \Vc^* \sc{\po} \Vc \in L^2(\IR \times \IS^2, \L(\H_\p) )$.
Consequently, we obtain, for $\alpha = \l, \r$, and some constant $C >0$ not depending on $\beta$,
\muu
{
\Pc &w^{(\alpha)}  \Pc  \\ &=   \Vc^* \sc{\po}^ {-1} \Vc  \int  (I^{(\alpha)}(u,\Sigma)  \Vc^* \sc{\po} \Vc  )^* I^{(\alpha)}(u,\Sigma)  \Vc^* \sc{\po} \Vc \d( u, \Sigma) \Vc^* \sc{\po}^ {-1} \Vc \\ &\geq C ( 1+ \beta^{-1}) \Vc^* \sc{\po}^{-2 } \Vc.
}
The proof for $w^{(\alpha)}_{1}$ is analogous. 
\end{proof}

\begin{lemma}
\label{lemma:W1_estimate}
For any $\delta > 0$  we have    
\begin{align} \label{estoncommint} 
\pm  \lambda  W_1    &\leq  \delta  \Nfh    + \frac{1}{\delta}  \lambda^2  w_1 \otimes \Id_\f
\end{align} 
in the sense of forms.
\end{lemma}
\begin{proof}
 The standard estimates
of the creation and annihilation operators lead to the inequality above. 
\end{proof}

Now, we estimate the second term involving $A_0$. As an immediate consequence of 
the  definition,  \eqref{eq:defofAzero}, we have 
\muu
{
\i [L_\lambda, A_0] = - \lambda [L_\lambda, \Pi W R_\varepsilon^2 \oPi - \oPi R_\varepsilon^2 W \Pi].
}
With respect to the different subspaces we obtain 
\muun
{
\Pi \i [ L_\lambda, A_0] \Pi &= 2  \lambda^2 \Pi W R_\varepsilon^2  W \Pi, \label{eq:pipi_normal}  \\
\oPi \i [ L_\lambda, A_0] \oPi &= -  \lambda^2 ( \oPi W \Pi W R_\varepsilon^2 \oPi + \oPi R_\varepsilon^2 W \Pi W \oPi), \label{eq:pipi} \\
\Pi \i [ L_\lambda, A_0] \oPi &=  \lambda \Pi W R_\varepsilon^2 \oPi L_\lambda \oPi.   \label{eq:pipip}
}
The Fermi Golden Rule condition implies strict positivity of  $\Pi \i [ L_\lambda, A_0] \Pi$, see \thref{th:fgr}. The other expressions can be potentially negative and are estimated in the following lemma, which we use later for a Birman-Schwinger argument. It contains slightly sharper estimates than a similar one in \cite{merkli2}. 
\begin{lemma}
\label{lemma:pi_estimate_new}
For all $\varepsilon > 0$ and all $\lambda \in \IR$ the following holds. 
\begin{enumerate}[label = (\alph*)]
\item
We have  $ \oPi \i [L_\lambda, A_0] \oPi =  \ind_{\Nfh=1}  \oPi \i [L_\lambda, A_0] \oPi \ind_{\Nfh=1}   $.  Moreover, 
\[
\nn{ \oPi \i [L_\lambda, A_0] \oPi }  \leq 2  \frac{\lambda^2}{\varepsilon^2} \nn{I}^2 . 
\]
\item  For arbitrary $\delta_1, \delta_2 > 0$,
\muu
{
 \oPi  &\i [L_\lambda, A_0] \Pi + \Pi  \i [L_\lambda, A_0] \oPi  \leq  
  (\abs\lambda \delta_1 + \delta_2 \lambda^2)  \Pi W R_\varepsilon^2  W \Pi  + \frac{  \abs\lambda}{ \delta_1}   \ind_{\Nfh = 1}     \\
  &\qquad  + 2  \frac{ \lambda^2} {\delta_2}   \bigg(      \ind_{\Nfh=2}  a^*(I)  R_\varepsilon^2 a(I)     \ind_{\Nfh=2}  \\ &\qquad \qquad +     \oPi  \int \frac{ I(u,\Sigma)^*  I(u,\Sigma)}{u^2 + \varepsilon^2}  \d(u, \Sigma) \otimes P_\Omega  \oPi  \bigg).
}
\end{enumerate}
\end{lemma}
\begin{proof}
\begin{enumerate}[label = (\alph*),wide, labelindent=0pt]
\item
Consider the first term in \eqref{eq:pipi}, 
\muu
{
 \oPi W \Pi W R_\varepsilon^2 \oPi = a^*(I) \Pi  a(I) R_\varepsilon^2 \oPi.
}
Clearly, this operator vanishes everywhere except on $\ran \ind_{\Nfh=1}$.
By standard estimates of creation and annihilation operators we obtain
 \[\nn{  a^*(I) \Pi  a(I)} \leq  \nn{I}^2  ,
\] thus
\muu
{
\nn{  \oPi W \Pi W R_\varepsilon^2 \oPi } \leq \frac{ \nn{I}^2  }{\varepsilon^2},
}
which proves the claim, as the second term in \eqref{eq:pipi} is just the adjoint of the first one. 
\item
Using \eqref{eq:pipip}, and the operator inequality \eqref{eq:op estimate} we get
\muun
{
\oPi & \i [L_\lambda, A_0] \Pi + \Pi  \i [L_\lambda, A_0] \oPi \nonumber  \\ &=  \lambda( \Pi W R_\varepsilon^2  L_0  \ind_{\Nfh=1} + \ind_{\Nfh=1}  L_0 R_\varepsilon^2 W \Pi) + \nonumber \\ &\qquad \lambda^2 ( \Pi W R_\varepsilon^2  \ind_{\Nfh=1} W \oPi + \oPi W \ind_{\Nfh=1}  R_\varepsilon^2 W \Pi)  \nonumber \\
&\leq   \abs\lambda ( \delta_1 \Pi W R_\varepsilon^2  W \Pi  + \delta_1^{-1} \ind_{\Nfh=1}  L_0  R_\varepsilon^2 L_0 \ind_{\Nfh=1} ) \label{eq:firstpiop} \\ &\qquad + \lambda^2( \delta_2 \Pi W R_\varepsilon^2  W \Pi + \delta_2^{-1} \oPi W \ind_{\Nfh=1}  R_\varepsilon^2 \ind_{\Nfh=1} W \oPi).   \label{eq:secondpiop}
}
We have
\[
\nn{ \ind_{\Nfh=1} L_0  R_\varepsilon^2 L_0  \ind_{\Nfh=1} } \leq 1,
\]
which yields a bound for the second operator in \eqref{eq:firstpiop}.  The second one in \eqref{eq:secondpiop} only operates on the space $\ran ( \ind_{\Nfh = 0}  + \ind_{\Nfh = 2})$, so we can write
\begin{align}
\oPi W \ind_{\Nfh=1} & R_\varepsilon^2 \ind_{\Nfh=1} W \oPi \nonumber  \\ &=  ( \ind_{\Nfh = 0}  + \ind_{\Nfh = 2}  ) \oPi W \ind_{\Nfh=1} R_\varepsilon^2 \ind_{\Nfh=1} W \oPi (  \ind_{\Nfh = 0}    + \ind_{\Nfh = 2}  ). \label{eq:local term}
\end{align}
Now, we can use again the operator inequality \eqref{eq:op estimate} 
and then the  pull-through formula to bound \eqref{eq:local term} by
\begin{align*}
& 2   a^*(I)   R_\varepsilon^2 a(I)  \ind_{\Nfh=2}  + 2  \oPi   a(I)   R_\varepsilon^2 \ a^*(I)   \oPi   \ind_{\Nfh=0}   \\
& \leq 
2  \left( a^*(I)   R_\varepsilon^2 a(I)  \ind_{\Nfh=2}  +  \oPi  \int  \frac{ I^*(u,\Sigma)  I(u,\Sigma)}{u^2 + \varepsilon^2}  \d(u,\Sigma) \otimes P_\Omega  \oPi  \right).
\qedhere
\end{align*}
\end{enumerate}
\end{proof}
Using the just proven  lemmas  we can now give the proof for the concrete error estimates  stated in \thref{th:total estimate steps}. 
\begin{proof}[Proof of \thref{th:total estimate steps}]
\begin{enumerate}[label=(\alph*),wide, labelindent=0pt]
\item First, we use \thref{lemma:W1_estimate} and the explicit form of $C_1$ to obtain on $\D$ for any $\delta_1 > 0$,
\muu
{
C_1 &= (\Vc^* \q^2 \Vc   \otimes \Id_\p + \Id_\p \otimes  \Vc^* \q^2 \Vc)    \otimes \Id_\f +  \Nfh + \lambda W_1 \\
&\geq ( \Vc^* \q^2 \Vc  \otimes \Id_\p + \Id_\p \otimes  \Vc^* \q^2 \Vc  )  \otimes \Id_\f +  (1 - \delta_1)  \PP{\Poh}  - \frac{\lambda^2}{\delta_1} w_1 .
}
Next, note that the operator inequality \eqref{eq:op estimate} yields
\[
w_1 \leq 2 ( w^{(\l)}_1 \otimes \Id_\p +  \Id_\p  \otimes w^{(\r)}_1  ).
\]
Then, using \eqref{eq:op estimate} again, and subsequently \thref{th:I1_q},  a decomposition into $\ran \Pc$ and $\ran \Pd$ gives, for $\alpha = \l, \r$, 
\begin{align*}
w_1^{(\alpha)} &\leq 2(\Pc w_1^{(\alpha)} \Pc +  \Pd w_1^{(\alpha)} \Pd ) \\
&\leq C (1+\beta^{-1}) ( \Vc^* \sc{\po}^{-2} \Vc + \Pd),
\end{align*}
where $C > 0$ is a constant not depending on $\beta$. Choosing any $0< \delta_1 < 1$ yields \eqref{eq:C1 estimate} on $\D$.  As $\D$ is a core for $C_1$, it is also a form core for $q_{C_1}$, so the operator inequality can be extended to the corresponding forms in the form sense on $\Def(q_{C_1})$. 
\item 
We have for all $\delta_1  \delta_2 > 0$ by \thref{lemma:pi_estimate_new},
\muu
{
\i [L_\lambda, A_0] &= \Pi i [ L_\lambda , A_0 ]   \Pi +   \oPi i [ L_\lambda , A_0 ]   \oPi +   \oPi i [ L_\lambda , A_0 ]   \Pi +   \Pi i [ L_\lambda , A_0 ]   \oPi  \\
 &\geq (1-  (\abs \lambda \delta_1 + \delta_2 \lambda^2)) \Pi i [ L_\lambda , A_0 ]   \Pi  -  2 \frac{{\lambda}^2}{\varepsilon^2} \nn{ I }^2  \ind_{\Nfh =1}   \\
& \qquad -  \frac{ \abs\lambda}{\delta_1} \ind_{\Nfh =1}  -  2  \frac{ \lambda^2}{\delta_2}  \left(  a^*(I)   R_\varepsilon^2 a(I)  \ind_{\Nfh =2}   + \frac{ 1}{\varepsilon^2}  \oPi  w \otimes P_\Omega  \oPi  \right).
}
Notice that the operator $  a^*(I)  R_\varepsilon^2 a(I)  \ind_{\Nfh =2}$ is in fact bounded
by  \autoref{th:I1_q}, and $\Poh \oPi = \ind_{L_\p \not= 0} \otimes P_\Omega$. 
The last term can be decomposed and estimated as above in the proof of \ref{it:C1 estimate}. 
\item
We have
\muu
{
W (L_\p^{-1} Q \otimes P_\Omega ) = a^*(I)  (L_\p^{-1} Q \otimes P_\Omega)  = a^*(I L_\p^{-1} Q   ) \Poh.
}
Thus, by the standard estimates for creation and annihilation operators, we obtain for all  $\delta > 0$ on $\Def(\Nfh)$,
\begin{align*}
&  W ( L_\p^{-1}  Q \otimes P_\Omega )    +   ( W ( L_\p^{-1}  Q \otimes P_\Omega )  )^*  \\
& \quad  \leq \delta \Nfh + \delta^{-1} \left(   \int
L_\p^{-1} Q  I(u,\Sigma)^* I(u,\Sigma) L_\p^{-1} Q   \d(u,\Sigma) \right) \otimes  P_\Omega.
\end{align*} 
Hence,  we get
\muun
{
\CQ{Q} &\geq Q \otimes  P_\Omega  - \abs \lambda \delta  \Nfh -  \abs \lambda \nn{ w }  \delta^{-1}   
( L_\p^{-1} Q )^2 \otimes P_\Omega   \nonumber  \\
 &\geq  (1 - 2 \abs \lambda (1+\beta^{-1}) \delta^{-1}) Q \otimes  P_\Omega  - \abs \lambda \delta  \Nfh,  \label{eq:final3}
}
where we used \autoref{th:I1_q} and the fact that the concrete choice of $Q$ as in \eqref{eq:Q defn} yields
\[
( L_\p^{-1} Q )^2 =  ( L_p^2 \ind_{[-1,1]}(L_\p) + L_\p^{-2} \ind_{(-\infty,-1) \cup (1, \infty)}(L_\p)) \leq 2 Q. 
\qedhere
\]
\end{enumerate}
\end{proof}


After the preparations we are now able to put all the estimates of this section together in order to prove positivity of $\qtot$. 
\begin{proof}[Proof of \thref{th:positive}]
First choose $\varepsilon > 0$ such that \thref{th:fgr} holds.
Using \thref{th:total estimate steps}, we obtain 
for the quadratic form defined in \eqref{defoffinalcommutator} 
in the sense of forms on $\Def(\qce)$ for all $\theta > 0$, 
\begin{align}
\nonumber \qtot &\geq  \Vc^* \left(  \q^2    - c_1  \lambda^2( 1+\theta ) (1+\beta^{-1})\sc{\po}^{-2}\right) \Vc \otimes \Id_\p \otimes \Id_\f  \\
\nonumber &\qquad  + \Id_\p \otimes   \Vc^* \left(   \q^2    - c_1 \lambda^2( 1+\theta )(1+\beta^{-1}) \sc{\po}^{-2} \right) \Vc \otimes \Id_\f  \\
\label{eq:Pomega} &\qquad + \left( c_2 - c_3 (1+\theta) \abs \lambda(1+\beta^{-1}) \right)  \PP{\Poh} \\
\label{eq:fgr estimate} &\qquad + 2  \theta \lambda^2  (1- c_4 \abs \lambda) \gb(\varepsilon) \Pi -  c_5 \lambda ^2(1+\beta^{-1}) \Pi \\
\label{eq:Cq estimate} 
\begin{split}
 &\qquad +  \big[ \left( 1-c_6 \abs \lambda(1+\beta^{-1}) \right) Q  \\ &\qquad  \qquad - c_7 \lambda ^2 (1 + \theta)(1+\beta^{-1}) \PP{( \Pc \otimes \Pc )}  \ind_{L_\p \not= 0} \big] \otimes P_\Omega, 
\end{split}
\end{align}
for constants $c_i > 0$, $i \in \{1,\ldots,7\}$, independent of $\lambda$ and $\beta$,
where we made use of  the identity  $(\Pc \otimes \Pc)^\perp \otimes P_\Omega = 
(\Pc \otimes \Pc)^\perp \ind_{L_\p \neq 0 } \otimes P_\Omega + \Pi$.
Now we set $\theta = \abs \lambda^{-1/2}$ in order to have a positive term in \eqref{eq:fgr estimate} of higher order.  Next, we make $\abs \lambda > 0$ sufficiently small in the following sense: First we make it so small such that, by the uncertainty principle lemma (cf. \cite[X.2]{rs2}),
\[
 \q^2    -  c_1  \lambda^2( 1+ \abs \lambda^{-\frac{1}{2}} )(1+\beta^{-1}) \sc{\po}^{-2} > 0.
\]
Furthermore, we can make it small enough such that we get strictly positive operators in \eqref{eq:Pomega} and \eqref{eq:fgr estimate} on $\ran \PP{P_\Omega}$ and $\ran \Pi$, respectively. Note that for \eqref{eq:fgr estimate} to be positive  
 we have to choose $\abs \lambda$ small  compared  to $\gb(\varepsilon)^2$. For the last term, \eqref{eq:Cq estimate}, we first observe that  
\begin{equation} \label{finalpos1}  
\PP{( \Pc \otimes \Pc  )}  \ind_{L_\p \not= 0} \leq \ind_{[\Xi,\infty)}(L_\p^2), 
\end{equation} 
where
\[
\Xi := \inf_{\substack{ \lambda  \in \sigma_\d(H_p) , \mu \in \sigma(H_p) \\ \lambda \neq \mu }}
(\lambda - \mu)^2> 0.
\]
Furthermore,   using  the   definition  \eqref{eq:Q defn}, $Q =  L_\p^2 \ind_{[0,1]}(L_\p^2) + \ind_{(1,\infty)}(L_\p^2)$, we see   that 
\begin{align}
Q - \delta  \ind_{[\Xi,\infty)}(L_\p^2) &=  L_\p^2 \ind_{[0,\Xi)} (L_\p^2) + (L_\p^2 - \delta ) \ind_{[\Xi,1]}(L_\p^2) + (1- \delta ) \ind_{(1,\infty)}(L_\p^2) \nonumber \\ &> 0   \label{finalpos2}  
\end{align}
on $\ran \ind_{L_\p \not= 0}$  whenever $\delta < {\min\{\Xi,1\}}$. Now   using  \eqref{finalpos1}   and  \eqref{finalpos2}  
we can achieve   that \eqref{eq:Cq estimate} is positive on  \[\ran (\ind_{L_\p \not= 0} \otimes P_\Omega)\] for $\abs \lambda$ small enough. The claim now follows in view of the decomposition \eqref{eq:ranpi}.
\end{proof}

\appendix

\section{Decay in the Glued Positive Temperature Space}

In \autoref{lemma:gluing_integrable} the decay behavior of functions and their derivatives in the positive temperature gluing representation is discussed. That is, sufficient decay conditions of a function are stated such that the derivatives of the transformation by $\tau_\beta$, defined as in \Autoref{subsec:liouvillian}, are integrable. This is required for the proof of \autoref{th:G commutator conditions}. The same kind of bounds were already used in \cite{merkli1,merkli2} but without explicitly mentioning them. 

Before that we state some properties of the decay behavior of the Planck distribution and its first derivative.  
\begin{lemma}[Properties of $\rho_\beta$]
\label{th:rho_beta properties}
There exists constants $C$ such that for all $\omega >  0$, $\beta > 0$, 
\begin{enumerate}[label=(\alph*)]
\item $\sqrt{\rho_\beta(\omega)} \leq \frac{1}{\sqrt{\beta \omega} }$,
\item $\sqrt{1+ \rho_\beta(\omega)} \leq 1+  \frac{1}{\sqrt{\beta \omega} }$,
\item $\partial_\omega \sqrt{\rho_\beta(\omega)} \leq C( \omega^{-1} + \beta^{-\frac{1}{2}} \omega^{-\frac{3}{2}} )$,
\item $\partial_\omega \sqrt{1+ \rho_\beta(\omega)} \leq C( \omega^{-1} + \beta^{-\frac{1}{2}} \omega^{-\frac{3}{2}} )$.
\end{enumerate}

\end{lemma}

\begin{proof}
Parts (a) and (b) follow from an elementary calculation. 
\begin{enumerate}[label=(\alph*),wide, labelindent=0pt]
 \addtocounter{enumi}{2}
\item
  $ \begin{aligned}[t]
\left| \partial_\omega \sqrt{\rho_\beta(\omega)} \right| &  = 
 \frac{\beta e^{\beta \omega}}{2 (e^{\beta \omega } - 1 )^{3/2} } = \omega^{-1} \frac{\beta \omega e^{ \beta \omega}}{2 (e^{\beta \omega } - 1 )^{3/2} } \\
& \leq  \omega^{-1}  C \left(   1_{\beta \omega \geq 1} + 1_{\beta \omega <  1}   (\beta \omega  )^{-1/2} \right)  . 
 \end{aligned} $
\item 
 $ \begin{aligned}[t]
\left| \partial_\omega \sqrt{1 + \rho_\beta(\omega)} \right| = 
 \frac{\beta e^{-\beta \omega}}{2 (1- e^{-\beta \omega }  )^{3/2} }
= \frac{\beta e^{\frac{1}{2}\beta \omega}}{2 (e^{\beta \omega } - 1   )^{3/2} }
 \leq \left| \partial_\omega \sqrt{\rho_\beta(\omega)} \right|.
 \end{aligned} $
 \qedhere
\end{enumerate}
\end{proof}

\begin{rem} \label{toplh} 
We note that in the following lemma  derivatives of the $\mathcal{L}(\H)$-valued functions
can be understood in the operator norm, strong operator, or weak operator topology in $\mathcal{L}(\H)$, respectively.
\end{rem} 

\begin{lemma}
\label{lemma:gluing_integrable}
Let $\H$ be a Hilbert space, $F \: \IR_+ \times \IS^2 \rightarrow \Lb(\H)$  a measurable function,  and $m \in \IN$.
Assume that for $j=0,\ldots,m$ the partial derivatives  $ \partial_\omega^j F(\omega,\Sigma)$ exist on $\R_+ \times \IS^2$,  see \autoref{toplh},
and that they are essentially bounded on compact subsets. 
Furthermore assume that there exist constants $\varepsilon > 0 $ and $ k, K,  C_1,  C_2 
\in (0,\infty) $ such that for all $\Sigma \in \IS^2$, $j = 0, \ldots , m$ 
\begin{enumerate}[label=(\arabic*)]
\item $\nn{ \partial_\omega^j F(\omega,\Sigma) } \leq  C_1 \omega^{m-1 + \varepsilon -j }$, for $\omega \in (0, k)$, \label{gluing:ir condition}
\item  $\nn{ \partial_\omega^j F(\omega,\Sigma) } \leq  C_2 \omega^{-\frac{3}{2} -\varepsilon} $, for $\omega \in ( K,\infty)$. \label{gluing:uv condition}
\end{enumerate}
Then the weak partial derivatives  $\partial_u^j \tau_\beta(F)$ are in 
$L^2(\IR \times \IS^2, \Lb(\H))$ for $j=0,\ldots,m$.
In particular, for $j=0,1$, there exists a constant $C_3$ such that for all  $\beta \in (0,\infty)$   
\begin{align}
\label{eq:beta eq}
\nn{\partial_u^j \tau_\beta(F)}_{L^2(\IR \times \IS^2, \Lb(\H))} \leq C_3 (1+ \beta^{-\frac{1}{2}}),
\end{align}
where $\tau_\beta$ was defined in \eqref{eq:taubeta defn}. 
 The assertion  also holds, if we assume
instead of Condition \ref{gluing:ir condition}: 
\begin{enumerate}[label=(\arabic*')] 
\item
\label{alternative ir} There exists a  $J  \in \IN_0$
such that 
$$F(\omega,\Sigma) = \omega^{- \frac{1}{2} +J }  F_0(\omega,\Sigma) ,  \quad \omega \in (0, k), $$ 
where  $F_0$ is an $\mathcal{L}(\H)$-valued function on $[0,k) \times \IS^2$ such that  for $j = 0,\ldots,\allowbreak \max\{0,m-J\}$ the partial derivatives 
 $\partial_\omega^j F_0$ exist,  
are uniformly bounded, 
and satisfy the relation 
  $$\partial_{\omega}^j F_0(\omega,\Sigma) |_{\omega=0} =    (-1)^{j+J+1} \partial_{\omega}^j  F_0(\omega,\Sigma)^*|_{\omega=0} . $$ 
\end{enumerate}
\end{lemma}
\begin{proof} We will treat small   and large $|u|$ separately.
In particular,  \eqref{eq:beta eq} will follow as a consequence of 
 \eqref{eq:inteoffact3}  and  \eqref{eq:inteoffact2}, below.  Let us start  
using  Leibniz' formula  
\begin{align}
\label{eq:tau beta product}
\partial_u^j( \tau_\beta F ) (u,\Sigma) = \begin{cases} \sum_{l=0}^j  \binom{j}{l} \partial_u^{l} \left( u \sqrt{1+ \rho_\beta(u)} \right) \partial_u^{j-l} F(u, \Sigma) & \text{ if }  u > 0, \\
\sum_{l=0}^j  \binom{j}{l} \partial_u^{l} \left(  u \sqrt{ \rho_\beta(-u)} \right) \partial_u^{j-l} F(-u, \Sigma)^* &\text{ if }  u < 0.
 \end{cases}
\end{align}
We first  consider  $|u|$ at infinity. 
 The first factor in the first line in \eqref{eq:tau beta product} is in $O(u)$ for $u\to\infty$, and the first factor in the second line decays faster than any polynomial. 
This implies with \ref{gluing:uv condition} that    there exist constants $C_j(\beta)$, such that, for all $j = 0,\ldots,m$ and all $u$ with $\abs u  > K$,
\begin{align*}
\nn{ \partial_u^j( \tau_\beta F ) (u,\Sigma)  } \leq C_j(\beta) \abs u^{-\frac{1}{2}  - \varepsilon}. 
\end{align*}
This yields  that
\begin{align*}
 u \mapsto \ind_{\IR \setminus [-K,K]}(u)  \partial_u^j( \tau_\beta F)(u, \Sigma)
\end{align*}
is in $L^2(\IR, \Lb(\H))$ for all $\Sigma\in \IS^2$ and $j=0, \ldots, m$. In particular for $j=0,1$
we find in view of \thref{th:rho_beta properties} that $C_0(\beta)$ and $C_1(\beta)$ can be bounded 
by a constant times $1 + \beta^{-1/2}$. Thus we conclude that 
\begin{align} \label{eq:inteoffact3} 
 \| \ind_{\IR \setminus [-K,K]}(u) \partial_u^j( \tau_\beta F) \|_{L^2(\IR \times \IS^2, \Lb(\H))}  \leq C (1 + \beta^{-\frac{1}{2}} ) , \quad j = 0, 1 . 
\end{align}

Let us now  check the decay behavior near zero. 
First assume that \ref{gluing:ir condition} is satisfied. 
We  extend  $\tau_\beta F$ to a function  on $\IR \times \IS^2$
 by setting $\tau_\beta F(0,\Sigma) := 0$.
The first factors in the sums in \eqref{eq:tau beta product}, $\partial_u^{l} \left( u \sqrt{1+ \rho_\beta(u)} \right)$ and $\partial_u^{l} \left( u \sqrt{\rho_\beta(u)} \right)$,  are in $O(\abs u^{\frac{1}{2} - l})$ for $u \to 0$.   The norms of the second factors are in $O(\abs u^{m-1-(j-l) + \varepsilon})$ by Assumption   \ref{gluing:ir condition}. Hence, there exist constants $ c_j(\beta)$, such that, for $0 < \abs u  < k$ and $\Sigma \in \IS^2$,  we have
\begin{align}\label{estonderFbeta}
\nn{ \partial_u^j( \tau_\beta F ) (u,\Sigma)  } \leq c_j(\beta)  \abs u^{-\frac{1}{2} + m - j + \varepsilon} , \quad  j=0,\ldots,m . 
\end{align} 
We  conclude  from \eqref{estonderFbeta} that for each 
$\Sigma$ the function $u \mapsto \tau_\beta F(u,\Sigma)$   is 
$m-1$ times continuously differentiable  on $\R$ with $\partial_u^j \tau_\beta F(0,\Sigma) = 0$ for $j=1,\ldots,m-1$. Moreover, we see from \eqref{estonderFbeta}
that  for each $\Sigma$, the function 
$u \mapsto \partial_u^{m-1}\tau_\beta F(u,\Sigma)$ is weakly differentiable. 
 Furthermore,  we infer from  \eqref{estonderFbeta}  and the boundedness  of $F$ and its partial derivatives on compact subsets of $\R_+ \times \IS^2$ 
  that for any $R > 0$  
\begin{align} \label{eq:inteoffact} 
 (u,\Sigma)  \mapsto \ind_{[-R,R]}(u) \partial_u^j( \tau_\beta F)(u, \Sigma)
\end{align}
is in $L^2(\IR \times \IS^2, \Lb(\H))$ for all $j=0, \ldots, m$.  In particular, for $j=0,1$ 
we find in view  of   \thref{th:rho_beta properties} that we can 
bound    $c_0(\beta)$ and $c_1(\beta)$ by a constant 
times $1 + \beta^{-1/2}$, and so   
\begin{align} \label{eq:inteoffact2} 
 \| \ind_{[-R,R]}(u) \partial_u^j( \tau_\beta F) \|_{L^2(\IR \times \IS^2, \Lb(\H))}  \leq C (1 + \beta^{-\frac{1}{2}} ) , \quad j = 0, 1 . 
\end{align}

Let us now assume the  alternative condition  \ref{alternative ir} is satisfied. In that case we first oberve that 
 for $F(\omega,\Sigma) := \omega^{-\frac{1}{2} + J }  F_0(\omega,\Sigma)$  we can write for $\abs u < k$,
\begin{align*}
(\tau_\beta F)(u, \Sigma) &:= \begin{cases}  u^{J}  \sqrt{\sigma_\beta(-u)} F_0(u, \Sigma), & u > 0, \\
						-	(	-u)^{J}  \sqrt{\sigma_\beta(-u)} F_0(-u ,\Sigma)^* , & u < 0 , \end{cases}
\end{align*}
where we defined  the function
$$
\IR \to \IR , \quad x \mapsto \sigma_\beta(x) = \left\{ \begin{array}{ll} x \rho_\beta(x), &  x \neq 0, \\  \frac{1}{\beta}, &  x = 0 . \end{array} \right.  
$$
The function $\sigma_\beta$ is  $C^\infty$ and positive, which for large $|x|$ is obvious and for small  $|x| \neq 0$ can be 
seen  from  the power series expansion 
$$
\sigma_\beta(x) = \beta^{-1}  \left( 1 + \sum_{n=1}^\infty \frac{(\beta x )^n}{(n+1)!} \right)^{-1}  .
$$
 It is now  straightforward to verify the claimed  differentiability and boundendness property 
if the   assumed  conditions are satisfied, by continuously extending the function 
at  zero and  applying  the product rule.
We infer from boundedness the $L^2$--integrablity of \eqref{eq:inteoffact}.  In particular, for $j=0,1$ we again 
obtain a bound of the form \eqref{eq:inteoffact2},  by noting that $\sigma_\beta(x) = \beta^{-1}\sigma_1(\beta x)$, and so $\sigma_\beta'(x) = \sigma_1'(\beta x)$,  
  $\sup_y |\sigma_1(y)| \langle y \rangle^{-1} < \infty$, and  $\sup_y |\sigma_1'(y)| < \infty$.   
\end{proof}

\section{Fermi Golden Rule}

In this part we review the result \cite[Proposition 3.2]{merkli2} -- how the Fermi Golden Rule condition \ref{fgrc} implies the positivity of the commutator with $A_0$  --  and generalize it with the obvious modifications to the coupling considered in this paper. This will be used in \Autoref{sec:positivity} for the proof of positivity.

First we state some elementary properties of the conjugate operator $A_0$, which was introduced in \Autoref{subsec:main proof}. In the following we use the symbols as defined in \Autoref{sec:model}.
\begin{lemma}
\label{th:a0 welldefined}
The operator 
\[
A_0 = \i \lambda (\Pi W R_\varepsilon^2 \oPi - \oPi R_\varepsilon^2 W \Pi),
\]
is bounded, self-adjoint and $\ran A_0 \subseteq \Def(L_\lambda)$ for any $\lambda \in \IR$ and $\varepsilon > 0$. 
\end{lemma}
\begin{proof}
Note that $\Pi$ contains the projection to the vacuum subspace, so the creation operators yield  bounded 
contributions and the annihilation operators vanish. Thus, the operator is indeed bounded and self-adjoint by construction.

Furthermore, the range of the first summand of $A_0$ equals $\ran \Pi$ and the range  of the second  summand 
equals  $\Def(L_0^2) \cap \FF_\fin$, which are clearly subsets of $\Def(L_\lambda)$. 
\end{proof}

\begin{prop}
\label{th:fgr}
For all $\varepsilon > 0$ we have
\muu
{
 \Pi \Pw W  R_\varepsilon^2  W \Pw \Pi \geq \gb(\varepsilon) \Pi,
}
where $\gb(\varepsilon)$ is defined as in \Autoref{subsection:main}.
\end{prop}
\begin{proof} Notice that
\muu
{
\Pi = \ind_{L_\p\not=0} \otimes P_\Omega = \sum_{\substack{E \in \sigma_\d(H_\p) }} p_E \otimes p_E \otimes P_\Omega.
}
Then we compute
\begin{align*}
\Pi W R_\varepsilon^2 W \Pi & \geq \Pi W R_\varepsilon^2 (\Pc \otimes \Pd \otimes \Id_\f) W \Pi \\
&=  \Pi ( a(\tau_\beta(G \otimes \Id_\p)) - a(e^{-\beta \cdot /2}  \tau_\beta( \Id_\p \otimes  \G^*)) )  \frac{\Pc \otimes \Pd \otimes \Id_\f }{L_0^2 + \varepsilon^2}  \\ &\qquad  \qquad ( a^*( \tau_\beta(G \otimes \Id_\p)) - a^*(  e^{-\beta \cdot /2}  \tau_\beta( \Id_\p \otimes \G^*)) )  \Pi  \\
&=  \Pi  a(\tau_\beta(G \otimes \Id_\p) )   \frac{\Pc \otimes \Pd \otimes \Id_\f }{L_0^2 + \varepsilon^2}  a^*( \tau_\beta(G \otimes \Id_\p))  \Pi  \\
&= \sum_{E \in \sigma_\d(H_\p)}   \Pi  a(\tau_\beta(G \otimes \Id_\p) )   \frac{\Pc \otimes p_E \otimes \Id_\f }{ (H_\p \otimes \Id_\p \otimes \Id_\f - E + \widehat{\dG(u)} )^2  + \varepsilon^2} \\  &\qquad  \qquad a^*( \tau_\beta(G \otimes \Id_\p))  \Pi,  
\end{align*}
with $\widehat{\dG(u)} := \Id_\p \otimes \Id_\p \otimes \dG(u)$, and where we used the usual pull-through formula in the last step. Evaluating $\Pi$ and using the definition of $\tau_\beta$, we arrive at
\begin{align*}
\Pi W R_\varepsilon^2 W \Pi &  \geq  (\ind_{L_\p\not=0} \sum_{E \in \sigma_\d(H_\p)}  \int_\IR \int_{\IS^2}   \tau_\beta(G^* \otimes \Id_\p)(u, \Sigma)    \frac{\Pc}{ (H_\p - E + u)^2  + \varepsilon^2}  \otimes p_E    \\&\qquad  \tau_\beta(G \otimes \Id_\p) (u, \Sigma) \d \Sigma  \d u  \ind_{L_\p\not=0} ) \otimes  P_\Omega  \\
&= \sum_{E \in \sigma_\d(H_\p)}  p_E( F^{(1)}_\beta(E,\varepsilon) +  F^{(2)}_\beta(E,\varepsilon) ) p_E \otimes p_E \otimes P_\Omega,
\end{align*} 
with $F^{(1)}_\beta(E,\varepsilon)$, $F^{(2)}_\beta(E,\varepsilon)$ defined as in \Autoref{subsection:main}.
%
\end{proof}

\bibliography{main}

\begin{thebibliography}{10}

\bibitem{a0firsttime}
V.~Bach, J.~Fr{\"o}hlich, I.M. Sigal, and A.~Soffer.
\newblock Positive commutators and the spectrum of {Pauli}--{Fierz}
  {Hamiltonian} of atoms and molecules.
\newblock {\em Communications in Mathematical Physics}, 207(3):557--587, Nov
  1999.

\bibitem{rte1}
V.~Bach, J.~Fröhlich, and I.M. Sigal.
\newblock Return to equilibrium.
\newblock {\em Journal of Mathematical Physics}, 41(6):3985--4060, 2000.

\bibitem{photoelectric1}
V.~Bach, F.~Klopp, and H.~Zenk.
\newblock Mathematical analysis of the photoelectric effect.
\newblock {\em Advances in Theoretical and Mathematical Physics}, 5:969--999,
  11 2002.

\bibitem{derezinski_rte}
J.~Derezi{\'n}ski and V.~Jak{\v{s}}i{\'c}.
\newblock Return to equilibrium for {Pauli}-{Fierz} systems.
\newblock In {\em Annales Henri Poincar{\'e}}, volume~4, pages 739--793.
  Springer, 2003.

\bibitem{dj}
J.~Dereziński and V.~Jakšić.
\newblock Spectral theory of {Pauli}–{Fierz} operators.
\newblock {\em Journal of Functional Analysis}, 180(2):243 -- 327, 2001.

\bibitem{merkli1}
J.~Fr{\"o}hlich and M.~Merkli.
\newblock Thermal ionization.
\newblock {\em Mathematical Physics, Analysis and Geometry}, 7(3):239--287, Aug
  2004.

\bibitem{merkli2}
J.~Fr{\"o}hlich, M.~Merkli, and I.~M. Sigal.
\newblock Ionization of atoms in a thermal field.
\newblock {\em Journal of Statistical Physics}, 116(1):311--359, Aug 2004.

\bibitem{froehlich_invariance}
J.~Fröhlich.
\newblock Application of commutator theorems to the integration of
  representations of {L}ie algebras and commutation relations.
\newblock {\em Comm. Math. Phys.}, 54(2):135--150, 1977.

\bibitem{rte2}
J.~Fröhlich and M.~Merkli.
\newblock Another return of 'return to equilibrium'.
\newblock {\em Communications in Mathematical Physics}, 251, 11 2004.

\bibitem{photoelectric3}
M.~Griesemer and H.~Zenk.
\newblock On the atomic photoeffect in non-relativistic {QED}.
\newblock {\em Communications in Mathematical Physics}, page 615–639, 10
  2009.

\bibitem{ikebe}
T.~Ikebe.
\newblock Eigenfunction expansions associated with the {Schroedinger} operators
  and their applications to scattering theory.
\newblock {\em Archive for Rational Mechanics and Analysis}, 5(1):1, Jan 1960.

\bibitem{jakpillet2}
V.~Jak{\v{s}}i{\'{c}} and C.-A. Pillet.
\newblock On a model for quantum friction, {II}. {F}ermi's golden rule and
  dynamics at positive temperature.
\newblock {\em Communications in Mathematical Physics}, 176(3):619--644, Mar
  1996.

\bibitem{jaksic_pillet_multiple_fermionic}
V.~Jak{\v{s}}i{\'c} and C.-A. Pillet.
\newblock Non-equilibrium steady states of finite quantum systems coupled to
  thermal reservoirs.
\newblock {\em Communications in mathematical physics}, 226(1):131--162, 2002.

\bibitem{jakpillet3}
V.~Jakšić and C.-A. Pillet.
\newblock On a model for quantum friction. {III. Ergodic} properties of the
  spin-boson system.
\newblock {\em Comm. Math. Phys.}, 178(3):627--651, 1996.

\bibitem{kato1}
T.~Kato.
\newblock Growth properties of solutions of the reduced wave equation with a
  variable coefficient.
\newblock {\em Communications on Pure and Applied Mathematics}, 12(3):403--425,
  1959.

\bibitem{merkli_positive}
M.~Merkli.
\newblock Positive commutators in non-equilibrium quantum statistical
  mechanics.
\newblock {\em Communications in Mathematical Physics}, 223(2):327--362, Oct
  2001.

\bibitem{instability_eqstates}
M.~Merkli, M.~M{\"u}ck, and I.M. Sigal.
\newblock Instability of equilibrium states for coupled heat reservoirs at
  different temperatures.
\newblock {\em Journal of Functional Analysis}, 243(1):87--120, 2007.

\bibitem{ness_resonances}
M.~Merkli, M.~M{\"u}ck, and I.M. Sigal.
\newblock Theory of non-equilibrium stationary states as a theory of
  resonances.
\newblock In {\em Annales Henri Poincar{\'e}}, volume~8, pages 1539--1593.
  Springer, 2007.

\bibitem{MorseFeshbach}
P.M.C. Morse and H.~Feshbach.
\newblock {\em Methods of Theoretical Physics}.
\newblock Number~1 in International series in pure and applied physics.
  McGraw-Hill, 1953.

\bibitem{mourre}
E.~Mourre.
\newblock Absence of singular continuous spectrum for certain self-adjoint
  operators.
\newblock {\em Communications in Mathematical Physics}, 78(3):391--408, Jan
  1981.

\bibitem{mueck_phd}
M.~Mück.
\newblock {\em Thermal relaxation for particle systems in interaction with
  several bosonic heat reservoirs}.
\newblock Univ., Mainz, 2004.

\bibitem{newton}
R.G. Newton.
\newblock {\em Inverse Schr{\"o}dinger Scattering in Three Dimensions}.
\newblock Theoretical and Mathematical Physics. Springer Berlin Heidelberg,
  2012.

\bibitem{rs2}
M.~Reed and B.~Simon.
\newblock {\em {II: Fourier} Analysis, Self-Adjointness}.
\newblock Methods of Modern Mathematical Physics. Elsevier Science, 1975.

\bibitem{rs4}
M.~Reed and B.~Simon.
\newblock {\em Methods of Modern Mathematical Physics: Vol.: 4. : Analysis of
  Operators}.
\newblock Academic Press, 1978.

\bibitem{rs3}
M.~Reed and B.~Simon.
\newblock {\em {III}: Scattering Theory}.
\newblock Methods of Modern Mathematical Physics. Elsevier Science, 1979.

\bibitem{simon_trace}
B.~Simon.
\newblock {\em Trace ideals and their applications}.
\newblock Number 120. American Mathematical Soc., 2010.

\bibitem{simon}
B.~Simon.
\newblock {\em Quantum Mechanics for Hamiltonians Defined as Quadratic Forms}.
\newblock Princeton Series in Physics. Princeton University Press, 2015.

\bibitem{ZemachKlein}
C.~Zemach and A.~Klein.
\newblock The born expansion in non-relativistic quantum theory.
\newblock {\em Il Nuovo Cimento (1955-1965)}, 10:1078--1087, 1958.

\bibitem{photoelectric2}
H.~Zenk.
\newblock Ionization by quantized electromagnetic fields: the photoelectric
  effect.
\newblock {\em Reviews in Mathematical Physics - RMP}, 20:367--406, 05 2008.

\end{thebibliography}
\addcontentsline{toc}{section}{References}

\end{document}